\documentclass[reqno,12pt]{amsart}

\usepackage[top=0.9in,left=1in,right=1in,bottom=0.9in]{geometry}
\usepackage{subcaption}
\usepackage{amssymb}
\usepackage[dvipsnames]{xcolor}
\usepackage[colorlinks=true,linkcolor=Red,citecolor=Green]{hyperref}

\usepackage{tikz}
\usepackage{autonum} 
\usepackage{dsfont} 
\usepackage{extarrows} 

\usepackage{environ} 
\usepackage{floatrow} 
\usepackage{bm} 

\newcommand{\R}{\mathbb{R}}

\newcommand{\N}{\mathbb{N}}

\newcommand{\p}{{\partial}}

\newcommand{\DD}{{\mathcal{D}}}

\newcommand{\Tr}{{\operatorname{Tr}}}

\newcommand{\te}{\theta}

\newcommand{\FF}{\mathcal{F}}

\newcommand{\supp}{\mathrm{supp}}

\newcommand{\tOmega}{\tilde{V}}

\newcommand{\Z}{\mathbb{Z}}
\newcommand{\Dd}{\mathbb{D}}

\newcommand{\CC}{\mathcal{C}}
\newcommand{\Ss}{\mathbb{S}}

\newcommand{\VV}{{\mathcal{V}}}

\newcommand{\CCC}{\mathcal{C}}

\newcommand{\II}{\mathcal{X}}

\newcommand{\UU}{{\mathcal{U}}}
\newcommand{\KK}{{\mathcal{K}}}

\newcommand{\Hh}{{\mathbb{H}}}

\newcommand{\ove}[1]{{\overline{#1}}}
\newcommand{\systeme}[1]{\left\{ \begin{matrix} #1 \end{matrix} \right.}

\newcommand{\Ind}{{\operatorname{Ind}}}

\newcommand{\C}{\mathbb{C}}

\renewcommand{\Im}{\operatorname{Im}}

\newcommand{\Aa}{\mathbb{A}}

\newcommand{\1}{\mathds{1}}

\newcommand{\de}{ \ \mathrel{\stackrel{\makebox[0pt]{\mbox{\normalfont\tiny def}}}{=}} \ }

\newcommand{\Int}{\operatorname{int}}
\newcommand{\AAA}{{\mathcal{A}}}
\newcommand{\tgamma}{{\tilde{\gamma}}}
\newcommand{\tGamma}{\widetilde{\Gamma}}

\numberwithin{equation}{section}

\usepackage[alphabetic,nobysame]{amsrefs}

\newcommand{\PUV}{{\Psi_{U,V}}}
\newcommand{\KUV}{{K_{U,V}}}
\newcommand{\SUV}{{\sigma_b^{U,V}}}
\newcommand{\oI}{{\operatorname{I}}}
\newcommand{\oII}{{\operatorname{II}}}

\newcommand{\bfx}{{\bm{x}}}
\newcommand{\bfy}{{\bm{y}}}
\newcommand{\bfz}{{\bm{z}}}
\newcommand{\bfw}{{\bm{w}}}
\newcommand{\ME}{{\mathcal E}}
\newcommand{\MG}{{\mathcal G}}

\title{The bulk-edge correspondence for curved interfaces}
\author{Alexis Drouot}
\address[Alexis Drouot]{University of Washington, Seattle, USA.} 
\email{adrouot@uw.edu}
\author{Xiaowen Zhu}
\address[Xiaowen Zhu]{University of Washington, Seattle, USA.} 
\email{xiaowenz@uw.edu}

\NewEnviron{equations}{%
\begin{equation}\begin{gathered}
  \BODY
\end{gathered}\end{equation}
}
\newtheorem{thm}{Theorem}
\newtheorem{assumption}{Assumption}

\newtheorem{definition}{Definition}
\newtheorem{lemma}{Lemma}[section]
\newtheorem{cor}[lemma]{Corollary}

\newtheorem{prop}{Proposition}
\newtheorem{theorem}[thm]{Theorem}
\newtheorem{claim}{Claim}
\theoremstyle{definition}

\begin{document} 

\begin{abstract} The bulk-edge correspondence is a condensed matter theorem that relates the conductance of a Hall insulator in a half-plane to that of its (straight) boundary. In this work, we extend this result to domains with curved boundaries.
    Under mild geometric assumptions, we prove that the edge conductance of a topological insulator sample is an integer multiple of its Hall conductance. This integer counts the algebraic number of times that the interface (suitably oriented) enters the measurement set. This result provides a rigorous proof of a well-known experimental observation: arbitrarily truncated topological insulators support edge currents, regardless of the shape of their boundary. 
\end{abstract}
\maketitle



\section{Introduction and Main results}
\subsection{Introduction} 

The study of topological insulators is a central topic in condensed matter physics. These materials are insulating phases of matter, described by a Hamiltonian with a spectral gap, to which one can associate a quantized topological invariant: the Hall conductance. When two insulators with distinct topological invariants are glued together, protected gapless currents emerge along the interface: the material becomes a conductor. For straight interfaces, the interface conductance is also quantized and equals the difference of the bulk topological invariants. This fundamental result is called the bulk-edge correspondence.

The quantization of the bulk conductance (i.e. the fact that it takes values in a discrete set) was first observed in the quantum Hall effect \cites{L81,TKNN, H82}. The characterization of Hall conductances as a Chern number marked the birth of topological phases of matter. Haldane \cite{H88} demonstrated on a famous model of magnetic graphene that this phenomenon is not restricted to quantum Hall systems. Hatsugai \cite{Hatsugai} then showed that topological materials support gapless states along their boundary; this is known today as the bulk-edge correspondence. Since then, the bulk-edge correspondence has been extended to various situations: see \cites{KM05a, KM05b,FKM07} for $\Z_2$-topological insulators, \cites{GT18, ST19} for Floquet topological systems, and \cites{HR08, RH08, YGS15,DMV17,PDV19,F19} for physical fields beyond quantum science, such as photonics, accounstics and fluid mechanics. By now mathematical proofs of the bulk-edge correspondence have spanned a wide variety of situations: see for instance \cites{Hatsugai, KS02, EG02, KS04,  L15, Ku17, Br18, LT22, L23} for discrete models, \cites{EGS05, LH11, GS18, ST19} for disordered systems, \cites{GP13, ASV13, FSSWY20} for $\Z_2$-topological insulators, and \cites{B19, D21b, D21, SW22a} for continuous Hamiltonians. 

Various experimental results have suggested that the bulk-edge correspondence is an extremely stable principle: it does not depend on the fine details of the interface \cite{WMTHXHW17, NHCR18, FSSWY20, TJCIBBP21}. However, with the exception of some $K$-theoretic approaches \cites{Ku17, LT22, L23}, most proofs of the bulk-edge correspondence have focused on straight interfaces. In this work, we provide a spectral proof of the bulk-edge correspondence for curved interfaces. We show equality between the edge conductance and an integer multiple of the Hall conductance; this goes beyond the aforementioned $K$-theoretic results where the bulk index takes the form of a geometric Hall conductance. We provide a physical interpretation of the emerging integer as an intersection number between the boundary and the measurement set. To the best of our knowledge, this is the first time that such a quantity emerges in the study of topological insulators.

\subsection{Main result}\label{sec-1.2}
We briefly review standard facts from condensed matter physics. Electronic propagation in a quantum material follows the Schr\"odinger equation with a selfadjoint operator $H$ on $\ell^2(\Z^2,\C^m)$. The operator $H$ describes the hopping of the electron between atomic sites. Its spectrum $\Sigma(H)$ characterizes the electronic nature of the material: $H$ is a conductor at energy $\lambda$ if and only if $\lambda \in \Sigma(H)$; and an insulator otherwise. In this paper, we will work with Hamiltonians that model limited hopping:

\begin{definition}[ESR]\label{def-1} We say that a selfadjoint operator $H$ on $\ell^2(\Z^2,\C^m)$ is \textit{exponentially-short-range (ESR)} if its kernel satisfies, for some $\nu >0$,
\begin{equation}\label{eq-9z}
    |H(\bfx, \bfy)|\leq \nu^{-1}e^{-2\nu d_1(\bfx,\bfy)}, \qquad \forall \bfx, \bfy\in \Z^2.
\end{equation}
\end{definition}

In \eqref{eq-9z}, $d_1(\bfx, \bfy)$ denotes the $\ell^1$-distance between $\bfx, \bfy \in \Z^2$.
When $H$ is an insulator at a certain energy, we can compute an invariant that characterizes the topological phase described by $H$: the Hall conductance. It is the intrinsic conductance of the material originating from the quantum Hall effect, see \cite{TKNN}.

\begin{definition}[Hall conductance] \label{def-1a} Let $H$ be an ESR Hamiltonian with a spectral gap $\MG$, $\lambda \in \MG$ and $P$ be the spectral projection below energy $\lambda$: $P := \1_{(-\infty, \lambda)}(H)$. The \textit{Hall conductance} of $H$ in the gap $\MG$ is:
\begin{equation}\label{eq-1t}
    \sigma_b(H) := -i \Tr \left(P \big[[P,\1_{\{x_2>0\}}],[P,\1_{\{x_1>0\}}]\big]\right).
\end{equation}
\end{definition}

The gap condition $\MG\subset\Sigma(H)^c$ ensures that \eqref{eq-1t} is well-defined; see e.g. \cite{EGS05} (or Proposition \ref{prop-2a} below). This work aims to prove the bulk-edge correspondence: for edge systems interpolating between two insulators, the conductance of the edge is equal to the difference between the two Hall conductances. We turn to the definition of edge systems:

\begin{definition}[Edge operator]\label{a2} Let $H_\pm$ be two ESR Hamiltonians and $U \subset \R^2$. An edge Hamiltonian $H_e$ associated to $H_+$ and $H_-$ in $U$ and $U^c$ is a selfadjoint operator on $\ell^2(\Z^2,\C^m)$ that satisfies the kernel estimate:
\begin{equation}\label{eq-1a}
\forall \bfx, \bfy \in \Z^2, \qquad
    \big| E(\bfx,\bfy) \big| \leq \nu^{-1} e^{-2\nu d_1(\bfx,\partial \Omega)}, \qquad E :=  H_e - \1_U H_+ \1_U - \1_{U^c} H_- \1_{U^c}.
\end{equation}
\end{definition}

The condition \eqref{eq-1a} means that $H_e$ describes two insulators glued along the edge $\p U$. To measure the conductance of $H_e$ along $\p U$, we now discuss sets \textit{transverse} to $U$.

\begin{definition}[Transversality] \label{def-1b}
    We say that two sets $U,V \subset \R^2$ are transverse if
\begin{equation}\label{eq-0a}
    \liminf_{|\bfx |\to +\infty} \frac{\ln \PUV(\bfx)}{\ln |\bfx|} >0, \qquad \PUV(\bfx) \de 1 + d_1(\bfx,\p U) + d_1(\bfx,\p V).
\end{equation}
\end{definition}

In \eqref{eq-0a}, $d_1(\bfx,\p U)$ denotes the $\ell^1$-distance between $\bfx \in \R^2$ and the set $\p U$. Transversality is a geometric condition on the relative position of the boundaries $\p U$ and $ \p V$. It demands that these two sets get away from each other at a relatively mild rate, typically $\PUV(\bfx) \gtrsim |\bfx|^\alpha$ for some $\alpha>0$. Under the transversality condition \eqref{eq-0a}, we can define the conductance of $H_e$ along $\p U$ into the measurement set $V$:

\begin{definition}[Edge conductance] \label{def-1c} Let $H_e$ be an edge operator associated to two Hamiltonians $H_+, H_-$ with a joint spectral gap $\MG$, in the sets $U, U^c$. Assume that  $U,V \subset \R^2$ are transverse. The \textit{edge conductance} of $H_e$ into $V$ for energies in the bulk spectral gap $\MG$ is: 
    \begin{equation}
        \label{eq-1u}
    \sigma_e^{U,V}(H_e)  =  i \Tr\big(\rho'(H_e)[H_e, \1_V]\big),
    \end{equation}
where $\rho\in C^\infty(\R;[0,1])$ is a function such that:
\begin{equation}
    \label{eq-1v}
\rho(x) = \begin{cases}
    1, & x\geq \sup \MG,\\
    0, & x\leq \inf \MG.
\end{cases}
\end{equation}
\end{definition}

In \eqref{eq-1u}, $[H_e, \1_V]$ measures the number of particles moving into the set $V$ per unit time, while $\rho'(H_e)$ is an energy density in $\MG$. As a result,  $\sigma_e(H_e)$ captures the expected charge moving along $\p U$ into $V$ per unit time and per unit energy: it is the conductance of $H_e$ along $\p U$. Transversality of $(U,V)$ guarantees that $\sigma_e(H_e)$ is well defined; see Proposition \ref{prop-2b} below. We also mention that $\sigma_e(H_e)$ does not depend on $\rho$ satisfying \eqref{eq-1v} -- this follows, for instance, from Theorem \ref{thm-main} below. Definitions \ref{def-1} - \ref{def-1c} serve as the basis for our setup:

\begin{assumption}\label{ass-all} In this work:
\begin{enumerate}
    \item[($\mathcal A$1)] $H_\pm$ are two selfadjoint ESR operators on $\ell^2(\Z^2,\C^m)$ with a joint spectral gap $\MG$; $P_\pm = \1_{(-\infty, \lambda)}(H_\pm)$ are the spectral projectors below energy $\lambda\in \MG$. 
    \item[($\mathcal A$2)] $U, V$  are transverse subsets of $\R^2$.
    \item[($\mathcal A$3)] $H_e$ is an edge Hamiltonian associated to $H_+, H_-$ in $U, U^c$.
\end{enumerate}
\end{assumption}

We are now ready to state our main result:

\begin{thm}[Bulk-edge correspondence]\label{thm-main} Under Assumption \ref{ass-all},
    \begin{equation}
        \label{eq-1w}
    \sigma_e^{U,V}(H_e) = \II_{U,V} \cdot \big(\sigma_b(P_+) - \sigma_b(P_-)\big),
    \end{equation}
    where $\II_{U,V}$ is an integer that depends exclusively on the sets $U, V$. 
\end{thm}

In rough terms, the integer $\II_{U, V}$ emerging in the formula \eqref{eq-1w}, which we call the \textit{intersection number,} counts algebraically how many times the boundary of $U$ (suitably oriented) enters $V$. See Figures \ref{fig: 1} and \ref{fig: 2} for a brief description on how to compute $\II_{U,V}$ and \S\ref{sec-5.2}, \S\ref{sec-7.2} for detailed definitions.

\begin{figure}[b]
     \centering
     \begin{subfigure}[ht]{0.3\textwidth}
         \centering
         \includegraphics[width=\textwidth]{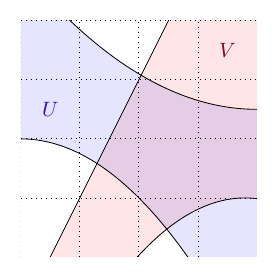}
         \caption{}
         \label{fig: 1(a)}
     \end{subfigure}
     \hfill
     \begin{subfigure}[ht]{0.3\textwidth}
         \centering
         \includegraphics[width=\textwidth]{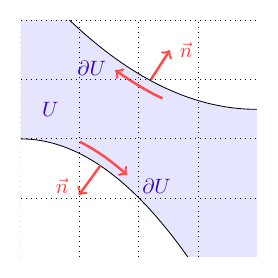}
         \caption{}
         \label{fig: 1(b)}
     \end{subfigure}
     \hfill
     \begin{subfigure}[ht]{0.3\textwidth}
         \centering
         \includegraphics[width=\textwidth]{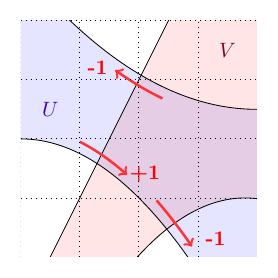}
         \caption{}
         \label{fig: 1(c)}
     \end{subfigure}
        \caption{We define the intersection number $\mathcal X_{U, V}$ between transverse simple sets $U, V$ in two steps. We first orient $\p U$ such that $U$ is to its left according to the outward-pointing normal, see (a) and (b); then we count how many times the oriented $\p U$ enters $V$, see (c). Here $\mathcal X_{U, V} = +1 - 1 - 1 = -1$.}
        \label{fig: 1}
\end{figure}

Theorem \ref{thm-main} is a version of the bulk-edge correspondence that goes beyond half-planes: it applies to sets $(U,V)$ with arbitrary geometry. 
In particular, leveraging flexibility on $V$ improves our knowledge on edge currents. When $\p U$ is disconnected, our result shows that each connected component of $\p U$ supports a current with conductance  $|\sigma_b(P_+) - \sigma_b(P_-)|$, propagating according to the orientation of $\p U$; see Figures \eqref{fig: 2(a)} and \eqref{fig: 2(b)}. Another application is when $U$ is a strip and $V$ is perpendicular to $U$: Theorem \ref{thm-main} implies that the edge conductance vanishes since the intersection number is $0$; see Figure \eqref{fig: 2(c)}. Both these facts were heuristically known but rigorous proofs were missing.

\subsection{Sketch of Proof}\label{sec-1.4}
Our proof of Theorem \ref{thm-main} consists of two independent components.

\textbf{Part 1 (\S\ref{sec-2}-\S\ref{sec-4}).} This part reduces the edge conductance $\sigma_e(H_e)$ to a bulk quantity (i.e. that depends only on $H_+$, $H_-$). The key observation is that one can formally think of $\sigma_e(H_e)$ as a \textit{singular trace:} 
\begin{equation}\label{eq-0y}
    \sigma_e^{U, V}(H_e) = \Tr \big( i \rho'(H_e) [H_e,\1_V] \big) \simeq \Tr \big( i[ \rho(H_e), \1_V ] \big).
\end{equation}
We call it ``singular'' because the operator $[\rho(H_e), \1_V]$ is not technically  trace-class. However, if $W$ is a $R$-tubular neighborhood of $\p U$ then $[\rho(H_e)\1_W,\1_V]$ is trace-class, and algebraic manipulations show that this trace vanishes. Therefore, as singular traces,
\begin{equation}\label{eq-6b}
    \sigma_e^{U, V}(H_e) 
    \simeq \Tr \big( [ \rho(H_e) \1_{W^c}, \1_V ] \big),
\end{equation}
which is essentially a bulk term: for $R \gg 1$, $\rho(H_e) \1_{W^c}$ depends essentially on the bulk Hamiltonians $H_+$ and $H_-$. 

At the heuristic level, our proof is based on the above informal observation but its mathematical execution is more sophisticated. We follow a strategy due to Elgart--Graf--Schenker \cite{EGS05}, tailored to the more complex geometry under focus here. To avoid trace-class issues, we inject the cutoff $\1_W$ in the very first step: in the formula $\Tr ( i \rho'(H_e) [H_e,\1_V] )$ rather than in the ill-defined quantity $\Tr ([ \rho(H_e), \1_V ] )$. Getting to the rigorous analogue of \eqref{eq-6b} produces commutators which we analyze using Helffer--Sj\"ostrand formulas and cyclicity. Edge terms vanish as predicted, and this eventually yields
\begin{equation}\label{eq-0z}
    \sigma_e^{U, V}(H_e) = \sigma_b^{U, V}(P_+) - \sigma_b^{U, V}(P_-), \qquad \sigma_b^{U,V}(P) := -i\Tr \big( P \big[[P,\1_U],[P,\1_V] \big] \big),
\end{equation}
see Theorem \ref{thm-4a} below.
We call the emerging qauntities $\sigma_b^{U,V}(P_\pm)$ \textit{geometric bulk conductances.}

\begin{figure}[t]
     \centering
     \begin{subfigure}[b]{0.3\textwidth}
         \centering
         \includegraphics[width=\textwidth]{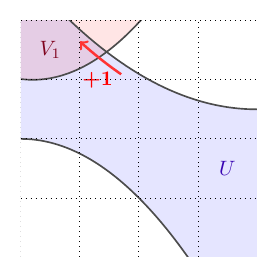}
         \caption{}
         \label{fig: 2(a)}
     \end{subfigure}
     \hfill
     \begin{subfigure}[b]{0.3\textwidth}
         \centering
         \includegraphics[width=\textwidth]{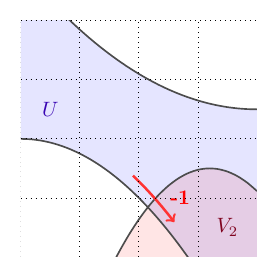}
         \caption{}
         \label{fig: 2(b)}
     \end{subfigure}
     \hfill
     \begin{subfigure}[b]{0.3\textwidth}
         \centering
         \includegraphics[width=\textwidth]{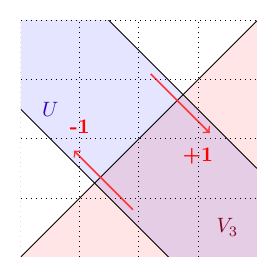}
         \caption{}
         \label{fig: 2(c)}
     \end{subfigure}
        \caption{The intersection number in subfigure \eqref{fig: 2(a)}, \eqref{fig: 2(b)}, \eqref{fig: 2(c)} are $+1$, $-1$ and $0$; thus the edge conductance $\sigma_e^{U, V_1}(H_e) = \sigma_b(P_+) - \sigma_b(P_-) = -\sigma_e^{U, V_2}(H_e)$, $\sigma_e^{U, V_3}(H_e) = 0$. }
        \label{fig: 2}
\end{figure}

\textbf{Part 2 (\S\ref{sec-5}-\S\ref{sec-7}).} In this part we reduce the geometric bulk conductances to integer multiples of Hall conductances (see Theorems \ref{thm:1} and \ref{thm:2} below); we give an interpretation of the emerging integer as an intersection number between $\p U$ and $\p V$. The key observation is that $\sigma_b^{U,V}(P)$ is the trace of a commutator $[A,B]$ where $AB$ and $BA$ are not separately trace-class:
\begin{equation}\label{eq-1y}
    \sigma_b^{U,V}(P) = -i\Tr \big( P \big[[P,\1_U],[P,\1_V] \big]\big) = -i\Tr \big( [P\1_U P,P\1_V P]\big).
\end{equation}
In particular, a compact perturbation $U'$ of $U$ does not affect $\sigma_b^{U,V}(P)$, because
\begin{equation}
    \sigma_b^{U',V}(P) - \sigma_b^{U,V}(P) = \sigma_b^{U' \Delta U,V}(P) = \Tr \big( [P\1_{U' \Delta U} P,P\1_V P]\big) = 0,
\end{equation}
where the last identity comes from cyclicity, see Proposition \ref{prop-2c} below.

When $U$ and $V$ are \textit{simple sets} (i.e. their boundaries are connected), our strategy consists of using the robustness of the geometric bulk conductance to locally deform $U$ and $V$ to half-planes and recover the Hall conductance. In details, we construct sets $U_n$, $V_n$, that (up to permutation of $U$ and $V$) resemble $\{x_2>0\}$, $\{x_1>0\}$ in $\Dd_{2n}(0)$ (the disk of radius $n$, centered at the origin) and $U$,$V$ outside $\Dd_{4n}(0)$. Since $U_n$, $V_n$ resemble $U$, $V$ outside $\Dd_{4n}(0)$, the robustness of geometric bulk conductances (Proposition \ref{prop-2c}) guarantees that
\begin{equation}
    \sigma_b^{U, V}(P) = \sigma_b^{U_n, V_n}(P) = \lim_{n\to \infty} \sigma_b^{U_n, V_n}(P).
\end{equation}
On the other hand, $\1_{U_n}(\bfx)$, $\1_{V_n}(\bfx)$ approach $\1_{\{x_2>0\}}(\bfx)$, $\1_{\{x_1>0\}}(\bfx)$ when $n\to \infty$ for every $\bfx$. So formally,
\begin{equation}\label{eq-1x} 
    \sigma_b^{U, V}(P) = \lim_{n\to \infty} \sigma_b^{U_n, V_n}(P) = \sigma_b(P).
\end{equation}

The main challenge is to make \eqref{eq-1x} rigorous. We will rely on an application of the dominated convergence theorem on specifically designed deformations of $(U,V)$. Because \eqref{eq-1y} is a trace, the involved kernel is negligible outside large enough balls. Our construction of $U_n, V_n$ preserve this negligibility uniformly, and this yield a rigorous proof of \eqref{eq-1x} for simple sets $U, V$; see Proposition \ref{lem-10}.

In \S\ref{sec-7} we will use additivity properties of the bulk conductances to extend \eqref{eq-1x} to arbitrary transverse sets $U, V$, driving the emergence of the intersection number between $U$ and $V$:
\begin{equation}\label{eq-9y}
    \sigma_b^{U, V}(P) = \II_{U,V} \cdot \sigma_b(P).
\end{equation}
See Theorem \ref{thm:2}. 
Theorem \ref{thm-main} will follow from combining \eqref{eq-0z} and \eqref{eq-9y}.

The paper is organized as follows: 
\begin{itemize}
    \item In \S \ref{sec-2}, we provide basic estimates on short-range operators and transverse sets.
    \item In \S \ref{sec-3}, we introduce of bulk and edge conductance and detail their basic properties.
    \item In \S \ref{sec-4}, we reduce the edge conductance to the difference between geometric bulk conductances (Theorem \ref{thm-4a}). 
    \item In \S \ref{sec-5}, we introduce the intersection number for simple transverse sets (transverse sets whose boundaries are connected). 
    \item In \S \ref{sec-6}, we prove Theorem \ref{thm:1}:  for simple transverse sets, the geometric bulk conductance is equal to the intersection number times the Hall conductance. 
    \item In \S \ref{sec-7}, we extend Theorem \ref{thm:1} to any pair of transverse sets. This leads to the emergence of the intersection number $\II_{U,V}$ and completes the proof of Theorem \ref{thm-main}.
\end{itemize}

\subsection{Relation to existing results}\label{sec-1e} The bulk-edge correspondence has been extensively studied in the past thirty years. Various perspectives have been developed around it: $K$-theoretic approaches \cites{KS02, KS04, T14, PS16, BKR, Ku17, Br18}; spectral flow methods \cites{Hatsugai, Br18}; and tracial techniques \cites{EG02, EGS05, GP13, GS18, GT18, ST19,B19, D21}. 

While many papers have focused on extending the bulk-edge correspondence to increasing degree of generality on the model -- from Landau Hamiltonian to $\Z_2$-insulators, from discrete to continuous, from periodic to disordered -- very few go beyond straight edges. The work \cite{BBDKLW23} gave an explicit formula for edge states of an adiabatically modulated Dirac operator with a weakly curved interface. The result relies on semiclassical methods and gives precise results but is restricted to the adiabatic regime; see also \cite{D22,B22,PPY22}. Our previous work \cite{DZ23} shows the emergence of edge spectrum for topological insulators as long as both the domain and its complement contain arbitrarily large balls, but did not provide a formula for the conductance. Thiang \cite{T20} derived similar results using $K$-theory and coarse geometry.

Later in \cites{LT22, L23}, Ludewig and Thiang proved a K-theoretic version of the bulk-edge correspondence on flasque spaces. These are spaces that satisfy coarse-geometric properties; in contrast with our condition \eqref{eq-0a}, flasqueness \cite[Definition 9.3]{R96} can be challenging to check on concrete examples beyond perturbations of half-spaces. Their K-theoretic bulk-edge correspondence result \cite[Theorem V.4]{L23} and edge-traveling interpretation \cite[Theorem VII.7]{L23} together connect the edge conductance to the geometric bulk conductance \eqref{eq-0z}. Our approach goes beyond \cites{LT22, L23} by connecting the geometric conductances themselves to a topological marker independent of the geometry: the Hall conductance \eqref{eq-1t}. This drove the emergence of the intersection number $\mathcal{X}_{U, V}$.

\subsection{Open problems}
One natural question is whether our curved bulk-edge correspondence \eqref{eq-1w} holds for $\Z_2$-topological insulators. This would require a space-based trace formula for $\Z_2$-bulk indices, such as \eqref{eq-1t} for Hall insulators. At this point it is unclear whether such a formula is available or even possible; see for instance discussions in \cites{LM19, FSSWY20}.

Another interesting problem concerns the type of the edge spectrum; it is widely expected to include an absolutely continuous part. For instance, \cite{BW22} proved that the edge spectrum is filled with  absolutely continuous spectrum when the edge is straight, combining a form of the bulk-edge correspondence derived in \cite{FSSWY20} with a result on the structure of unitary operators \cite{MM21}. The present work provides the first step in the context of curved edges. In a follow-up project we will show that the emerging edge spectrum is absolutely continuous, extending the result of \cite{BW22}.

In \cite{CB23}, Chen and Bal showed that for a straight edge, the bulk index is equal to the difference between the number of modes reflected and transmitted. It would be very interesting to extend this result to curved edges. Yet this promises to be challenging, because the scattering techniques developed there rely on translational invariance.

Several authors \cite{HM15, GMP17, BBDF18} have shown that the Hall conductance of interacting systems of particles is quantized (i.e. takes values in a discrete set). In our proof a \textit{geometric bulk conductance} naturally emerges, see \eqref{eq-0z}; it would be interesting to investigate whether the analogue of this quantity in interacting systems is quantized as well, and given by the integer $\II_{U,V}$ times the  Hall conductance.

\subsection{Funding and/or competing interests} We gratefully acknowledge support from the National Science Foundation DMS 2054589 (AD) and the Pacific Institute for the Mathematical Sciences (XZ). The contents of this work are solely the responsibility of the authors and do not necessarily represent the official views of PIMS.

\subsection{Notations}
We use the following notations:
\begin{itemize}
    \item Regular font $x$, $x_j$, $z$  denotes scalars in $\R$, $\Z$, or $\C$.
    \item Bold font $\bfx = (x_1, x_2), \bfy = (y_1, y_2)$ denotes points in $\Z^2$.
    \item $d(\cdot, \cdot)$ denote the Euclidean metric.
    \item $d_1(\bfx, \bfy) := |x_1 - y_1| + |x_2 - y_2|$ denote the $\ell^1$-metric on $\R^2$.
    \item $d_l(\bfx, \bfy) = \ln(1 + d_1(\bfx, \bfy))$ denote the logarithm metric on $\R^2$. 
    \item $\Vert \cdot \Vert$ and $\Vert \cdot \Vert_1$ denote the operator norm and the trace-class norm, respectively.
    \item Throughout the paper, $C$ denotes a constant that is allowed to vary from line to line, that may depend on $\nu$ and $\rho$ (though this dependence is not emphasized). We will occasionally use $C_N$ to emphasize dependence on an integer $N$.
    \item For $A \subset \R^2$, $A^\CC$ denotes the open set $\Int(A^c)$.
    \item $\Dd_R(\bfz):= \{\bfx\in \R^2: d_1(\bfx, \bfz) <R\}$ denotes the $\ell^1$-disk centered at $z$, of radius $R$.
    \item $\Hh_1 = \{x_1 > 0\}$ and $\Hh_2 = \{x_2 > 0\}$.
    \item If $K$ is an operator on $\ell^2(\Z^2)$, $K(\bfx,\bfy)$ denotes the Schwartz kernel of $K$.
    \end{itemize}

\section{Preparation}\label{sec-2}
In this section, we give some basic operator estimates that will be needed but leave the proof to Appendix \ref{app-ESR}. 

\subsection{Operator estimates} Under short-range assumptions, we give here estimates on resolvents and product of operators. We start by introducing the notion of textit{polynomially-short-range (PSR)} Hamiltonians.

\begin{definition}[PSR]
    An operator $H$ on $\ell^2(\Z^2,\C^m)$ is \textit{polynomially-short-range (PSR)} if for any $N>0$, there is $C_N>0$ such that 
\[
    |H(\bfx, \bfy)|\leq C_N(1 + d_1(\bfx, \bfy))^{-N} = C_N e^{-N d_l(\bfx, \bfy)}, \qquad \forall \bfx, \bfy\in \Z^2.
\]    
\end{definition}

Note that any ESR operator is also PSR. The next result shows that (1) resolvents of ESR for insulating energies are ESR; (2) functionals of ESR operators are PSR.

\begin{lemma}\label{lemma-1a} Assume $H$ is self-adjoint and ESR. 
    \begin{enumerate}
        \item[(a)] For $0<|\Im z|< 1$, $(H-z)^{-1}$ is ESR. More precisely,
        \begin{equation}\label{eq-9a}
        \left|(H - z)^{-1}(\bfx, \bfy)\right|\leq \frac{2}{|\Im z|}e^{-c|\Im z|d_1(\bfx, \bfy)}, \qquad c = \dfrac{\nu^4}{32}.
        \end{equation}
        \item[(b)] For any $g\in C_c^\infty(\R)$, $g(H)$ is PSR. 
    \end{enumerate}
\end{lemma}

We now give estimates on products of ESR / PSR operators. We state them in a unifying way using the two metrics $d_1, d_\ell$.

\begin{lemma}
    \label{lemma-1b} Let $U \subset \R^2$. Let $S$ be a selfadjoint operator on $\ell^2(\Z^2,\C^m)$ that satisfies the kernel estimate
    \[
    |S(\bfx, \bfy)|\leq Ce^{-cd_*(\bfx, \bfy)}.
    \]
    Then 
        \begin{align}
        &\|\1_U S \1_U(\bfx, \bfy)|\leq C e^{-cd_*(\bfx, \bfy) - cd_*(\bfx, U) - cd_*(\bfy, U)},\label{eq-2p}\\
           &|\1_ U  S \1_{ U ^c}(\bfx, \bfy)|\leq Ce^{-\frac{c}{3}d_*(\bfx, \bfy) - \frac{c}{3}d_*(\bfx, \p  U ) - \frac{c}{3}d_*(\bfy,\p U )},\label{eq-2b}\\
           &|[\1_ U , S](\bfx, \bfy)|\leq 2Ce^{-\frac{c}{3}d_*(\bfx, \bfy) - \frac{c}{3}d_*(\bfx, \p  U ) - \frac{c}{3}d_*(\bfy, \p  U )}.\label{eq-2c}
        \end{align}
    \end{lemma}

\begin{lemma}\label{lemma-1c}
       Let $ U ,  V \subset\R^2$. Assume that $S_1, \dots, S_n$ are self-adjoint operators on $\ell^2(\Z^2,\C^m)$ that satisfy the following estimates:
    \begin{equations}
        \label{eq-2g}
    \forall j \in [1,n], \qquad |S_j(\bfx, \bfy)|\leq C_j e^{-cd_*(\bfx, \bfy)}; \qquad \qquad   \qquad \qquad
    \\
    \exists p \in [1,n], \qquad  |S_p(\bfx, \bfy)|\leq  C_p e^{-cd_*(\bfx, \bfy) - cd_*(\bfx, \p U ) - cd_*(\bfy, \p U )}; 
    \\
    \exists q \in [1,n], \qquad |S_q(\bfx, \bfy)|\leq  C_qe^{-cd_*(\bfx, \bfy) - cd_*(\bfx, \p V ) - cd_*(\bfy, \p V)}.
    \end{equations}
    Then we have
        \begin{equation}
            \label{eq-2d}
            \left|\left(\prod_{j = 1}^n S_j \right)(\bfx, \bfy) \right|\leq \left(\prod_{j = 1}^n C_j \right) M_{d_*, \frac{c}{4}}^{n - 1} e^{-\frac{c}{4}d_*(\bfx, \bfy) -\frac{c}{4} d_*(\bfx, \p U ) -\frac{c}{4} d_*(\bfy, \p U ) -\frac{c}{4} d_*(\bfx, \p V ) -\frac{c}{4} d_*(\bfy, \p V )},
        \end{equation}
        where $M_{d_*, a} := \sum\limits_{\bfx\in \Z^2} e^{-a d_*(\bfx, \bfy)}$, $*\in \{1, l\}$.
\end{lemma}

Note that $M_{d_*,a}$ is independent of $\bfy\in \Z^2$ and 
\begin{equation}
    \label{eq-2e}
    \begin{split}
    &M_{d_1, a} \leq \frac{16}{a^2}, \quad \forall a>0,\quad \text{~and~} \quad M_{d_l, a} = \sum\limits_{\bfx \in \Z^2} \big(1 + |\bfx|_1\big)^{-a}<+\infty, \quad \forall a>2.
    \end{split}
\end{equation}

\subsection{Transversality and trace-class}
Recall that we say $U, V \subset \R^2$ if 
\begin{equation}
    \liminf_{|\bfx |\to +\infty} \frac{\ln \PUV(\bfx)}{\ln |\bfx|} >0, \qquad \PUV(\bfx) \de 1 + d_1(\bfx,\p U) + d_1(\bfx,\p V).
\end{equation}
Note that it is equivalent to 
\begin{equation}\label{eq-0b}
    \exists c\in (0,1) \text{ such that }  \forall |\bfx|\geq c^{-1}, \PUV(\bfx) \geq |\bfx|^c.
\end{equation}
Meanwhile, because of the inequality
\[
2\ln (1 + a + b) \geq \ln(1 + a) + \ln (1 + b)\geq \ln(1 + a + b), \qquad a,b >0,
\]
we have:  
\begin{equation}
    \label{eq-0x}
2\ln \Psi_{U, V}(\bfx) \geq d_l(\bfx, \p U) + d_l(\bfx, \p V) \geq \ln \Psi_{U, V}(\bfx).
\end{equation}
Therefore for any $N>\frac{2}{c}$, by \eqref{eq-0b}
\begin{equation}
    \label{eq-0v}
\sum_{\bfx\in \Z^2} e^{-N(d_l(\bfx, \p U) + d_l(\bfx, \p V))} \leq \sum_{\bfx\in \Z^2} \Psi_{U,V}(\bfx)^{-N} <+\infty.
\end{equation}

The next result states that product of PSR operators that decay away from the boundaries of transverse sets $U,V$ are trace-class.

\begin{cor}[trace-class]\label{cor-1a}
    Assume $U, V \subset \R^2$ are transverse sets. Assume $S_j$, $j = 1, \cdots, n$, are PSR and for some $p,  q\in [1,n]$ and for any $N$, there is $C_N$ such that
    \[
    |S_p(\bfx, \bfy)|\leq  C_Ne^{- Nd_l(\bfx, \p  U ) - Nd_l(\bfy, \p  U )},   \ \    |S_q(\bfx, \bfy)|\leq  C_Ne^{ - N d_l(\bfx, \p  V ) - N d_l(\bfy, \p  V )}.
    \]
    Then  $\prod\limits_{j = 1}^n S_j$ is trace-class. 
\end{cor}

\section{Bulk and edge conductance}\label{sec-3}
\begin{center}
   \textbf{From now on, we always assume that Assumption \ref{ass-all} holds.} 
\end{center}

In this section, we show that the Hall conductance \eqref{eq-1t}, the geometric bulk conductance \eqref{eq-0z}, and the edge conductance \eqref{eq-1u}, are well-defined. 

\subsection{Geometric bulk conductances}

\begin{prop}\label{prop-2a}
    Under Assumption \ref{ass-all}, the geometric bulk conductances
    \begin{equation}
        \label{eq-3v}
        \sigma_b^{U,V}(P_\pm):=  -i \Tr\big (P_\pm\big[[P_\pm, \1_U],[P_\pm, \1_V]\big]\big)
    \end{equation}
    are well-defined and independent of $\lambda \in \MG$.
\end{prop}

Because the right and upper half-planes are transverse sets, the Hall conductance $\sigma_b(H):= -i \Tr\big(P\big[[P, \1_{\{x_2>0\}}],[P, \1_{\{x_1>0\}}]\big]\big)$ is well-defined.

\begin{proof}
    Since $\MG$ is a spectral gap, for any $\lambda, \lambda'\in \MG$, the spectral projectors $\1_{(-\infty,\lambda)}(H)$ and $\1_{(-\infty,\lambda')}(H)$ are equal. This proves independence.
    
    We then show that $\sigma_b^{U,V}(P)$ is well-defined. Without loss of generalities we may assume that $\MG$ is an interval; pick $\lambda$ as the midpoint and $f\in C_c^\infty(\R)$ such that $f(H) = P$. By Lemma \ref{lemma-1a}(a),  $P = f(H)$ is PSR (one can actually prove that it is ESR). By \eqref{eq-2c}, for any $N$, there exists $C_N > 0$ such that 
    \begin{equation}
        \label{eq-3f}
    \begin{split}
        &|P(\bfx, \bfy)|\leq C_N e^{-4Nd_l(\bfx, \bfy)},\\
        &\big|[P, \1_U](\bfx, \bfy)\big|\leq C_Ne^{-4N(d_l(\bfx, \bfy) + d_l(\bfx, \p U) + d_l(\bfy, \p U))},\\
        &\big|[P, \1_{V}](\bfx, \bfy)\big|\leq C_Ne^{-4N(d_l(\bfx, \bfy) + d_l(\bfx, \p V) + d_l(\bfy, \p V))}.
    \end{split}
    \end{equation}
    Introduce 
    \begin{equation}\label{eq-9f}
        \KUV \de P[[P, \1_U], [P, \1_V]].
    \end{equation} 
    By  \eqref{eq-2d} in Lemma \ref{lemma-1c},
    \begin{equation}
        \label{eq-2k}
    |\KUV(\bfx, \bfy)|\leq C_N e^{-N (d_l(\bfx, \bfy) + d_l(\bfx, \p U) + d_l(\bfx, \p V) + d_l(\bfy, \p U) + d_l(\bfy, \p V))}.
    \end{equation}
    Since $U, V$ are transverse subsets of $\R^2$, by Corollary \ref{cor-1a}, $K_{U, V} = P\big[[P, \1_U][P, \1_V]\big]$ is trace-class; thus $\sigma_b^{U,V}(P)$ is well-defined.  This completes the proof. \end{proof}

 In particular, we have by \eqref{eq-0v} and \eqref{eq-2k} 
    \begin{equation}
        \label{eq-2q}
        |\SUV(P)| = \sum\limits_{\bfx\in \Z^2}|K_{U, V}(\bfx, \bfx)| 
    \leq C_N \left(\sum_{\bfx\in \Z^2} \Psi_{U, V}(\bfx)^{-N}\right)^2.    
    \end{equation}

Given two sets $A_1, A_2$, we introduce the symmetric difference
\begin{equation}
    A_1 \Delta A_2 = (A_1 \setminus A_2) \cup (A_2 \setminus A_1).
\end{equation}
One of the key properties of geometric bulk conductances is:

\begin{prop}\label{prop-2c}[Robustness of geometric bulk conductances]
    Under Assumption \ref{ass-all}, assume $U'$ is such that for some $R>0$,
    \begin{equation}\label{eq-2r}
        U \Delta U' \subset \{ \bfx \in \R^2: d_l(\bfx, \p U) \leq R\}.
    \end{equation}
    Then $\sigma_b^{U,V}(P_\pm) = \sigma_b^{U',V}(P_\pm)$. 
\end{prop}

\begin{proof} [Proof of Proposition \ref{prop-2c}] 
First, $U', V$ are transverse and $\sigma_b^{U', V}(P_\pm)$ are well-defined. Indeed, by \eqref{eq-0x} and \eqref{eq-2r}, 
\[
\begin{split}
    2\ln \Psi_{U', V}(\bfx) &\geq d_l(\bfx, \p U') + d_l(\bfx, \p V)\geq d_l(\bfx, \p U) + d_l(\bfx, \p V) - d_l(\p U, \p U')\\
    &\geq \ln \Psi_{U, V}(\bfx) - R \geq c\ln |\bfx| - R \geq c' \ln |\bfx|
\end{split}
\]
for any $c'>c$ and large enough $|\bfx|$. Hence $U', V$ are transverse sets and $\sigma_b^{U', V}(P_\pm)$ is well-defined. It remains to show:
\[
\sigma_b^{D, V}(P) 
= 0 \quad \text{when} \quad D = U\Delta U', \qquad P = P_\pm.
\]
By direct expansion of definition \eqref{eq-3v} and $1 = P + P^\perp$,  we get 
\[
\sigma_b^{D, V}(P) = \Tr\big(P\1_D P \1_V P - P \1_V P \1_D P\big) = \Tr\big( P \1_V P^\perp \1_D P - P\1_D P^\perp \1_V P\big). 
\]
Recall the cyclicity of trace-class operators -- see e.g. \cite[Proposition 7.3]{K89}: 
\begin{equation}
    \label{eq-0l}
    \Tr(AB) =\Tr(BA), \text{~when~} AB \text{~and~} BA \text{~are~trace~class.}
\end{equation}
Let us first assume that every operator below is trace-class so we can freely use cyclicity above. We get 
\[
\begin{split}
    \sigma_b^{D,V}(P) &= \Tr\big( P \1_V P^\perp \1_D - P^\perp \1_V P \1_D\big) = \Tr\big([P, \1_V]\1_D\big)\\
    &= \Tr\big(\1_{V^c} P\1_V\1_D - 1_{V} P \1_{V^c} \1_D\big)\\
    &= \Tr\big(P\1_V\1_D \1_{V^c} - P \1_{V^c} \1_D1_{V}\big) = 0,
\end{split}
\]
where we used $[P, \1_V] = \1_{V^c} P \1_V - \1_V P \1_{V^c}$. 

\smallskip 

Now we prove every operator mentioned above is trace-class. Note that $P^\perp \1_V P = P^\perp [1_V, P]$, $P \1_V P^{\perp} = [P, \1_V] P^{\perp}$ and recall that $P$ is PSR (see the proof of Proposition \ref{prop-2a}). By \eqref{eq-2c}, 
\[
    \big|P \1_V P^{\perp}(\bfx, \bfy)\big|, \big|P^\perp \1_V P(\bfx, \bfy)\big|\leq C_N e^{-6N d_l(\bfx, \bfy) - 6N d_l(\bfx, \p V) - 6N d_l(\bfy, \p V)}. 
\]
Meanwhile, $\1_D(\bfx)\leq e^{-6Nd_l(\bfx, D)}\leq e^{6N\ln (1 + R)}e^{-6Nd_l(\bfx,  \p U)}$. Hence 
\[
\big|\1_DP(\bfx, \bfy)\big|, \big|P\1_D(\bfx, \bfy)\big|\leq C_N(1 + R)^{6N} e^{-3Nd_l(\bfx, \bfy) - 3Nd_l(\bfx, \p U) - 3Nd_l(\bfy, \p U)}.
\]
Therefore by Corollary \ref{cor-1a}, $P \1_V P^\perp \1_D$ and $P^\perp \1_V P \1_D$ are both trace-class. Meanwhile, by \eqref{eq-2b} of Lemma \ref{lemma-1b} and $\1_D(\bfx) \leq (1 + R)^{6N}e^{-6Nd_l(\bfx, \p U)}$, hence
\[
\big|\1_{V^c} P\1_V\1_D(\bfx, \bfy)\big|\leq C_{6N}(1 + R)^{6N}e^{-2Nd_l(\bfx, \bfy) - 2Nd_l(\bfx, \p V) - 2Nd_l(\bfy, \p V) - 6Nd_l(\bfx, \p U)}. 
\]
By \eqref{eq-0v}, 
\[
\sum\limits_{\bfx, \bfy}\big|\1_{V^c} P\1_V\1_D(\bfx, \bfy)\big| \leq C_N(1 + R)^{6N}\sum\limits_{\bfx}e^{-2Nd_l(\bfx, \p V) - 2Nd_l(\bfx, \p U)} \left(\sum\limits_{\bfy} e^{-2Nd_l(\bfx, \bfy)}\right)<+\infty.
\]
Hence $\1_{V^c}P\1_V\1_D$, and similarly $1_{V} P \1_{V^c}\1_D$, are also trace-class. This completes the proof.
\end{proof}

\subsection{Edge conductance} 
Now we can introduce the edge conductance. 

\begin{prop}[Edge conductance]\label{prop-2b}
    Under Assumption \ref{ass-all}, the edge conductance into $V$ 
    \begin{equation}
        \label{eq-3w}
    \sigma_e^{U,V}(H_e)  =  i \Tr\big(\rho'(H_e)[H_e, \1_V]\big)
    \end{equation}
    is well-defined.
\end{prop}

For the rest of the paper, we consider an arbitrary but fixed pair of transverse sets $(U, V)$ in $\R^2$ and omit the superscripts $U, V$ from $\sigma_e^{U,V}(H_e)$ for convenience.

\begin{proof}[Proof of Proposition \ref{prop-2b}] 
1. By definition, $\supp(\rho')\subset \MG\subset \Sigma(H_+)^c\cap \Sigma(H_-)^c$; hence $\rho'(H_\pm) = 0$. Therefore, we have:
\begin{align}
    \rho'(H_e) & = \rho'(H_e) - \rho'(H_+)\1_U - \rho'(H_-)\1_{U^c}
    \\ & = \big(\rho'(H_e) - \rho'(H_+)\big)\1_U + \big(\rho'(H_e) - \rho'(H_-)\big) \1_{U^c}. \label{eq-9e}  
\end{align}

2. We now estimate $\rho'(H_e) - \rho'(H_+)$. By Helffer-Sj\"ostrand formula \eqref{eq: HS_1}, we have
    \begin{equation}
        \label{eq-2o}
    \begin{split}
        \rho'(H_e) - \rho'(H_+) &= \frac{1}{2\pi i}\int_\C \frac{\p \tilde \rho'}{\p \bar{z}} \left((H_e - z)^{-1} - (H_+ - z)^{-1} \right) dz\wedge d\bar{z}\\
        &= -\frac{1}{2\pi i}\int_\C \frac{\p \tilde \rho'}{\p \bar{z}} (H_e - z)^{-1} (H_e - H_+) (H_+ - z)^{-1} dz\wedge d\bar{z}.
    \end{split}
    \end{equation}
    Note that 
    \[
    \begin{split}
            H_e - H_+ &= E + \1_U H_+\1_U + \1_{U^c} H_-\1_{U^c} - H_+\\
            &= E - \big(\1_{U^c} H_+ \1_U + \1_U H_+\1_{U^c}\big) + \1_{U^c} (H_-  - H_+) \1_{U^c} := E - F + G.
    \end{split}    
    \]
And we have the following estimates:
    \begin{enumerate}
        \item By \eqref{eq-9a} in Lemma \ref{lemma-1a}, when $|\Im z|< 1$, we have:
    \begin{equation}
        \label{eq: est_kernel}
    \begin{split}
        \left|(H_+ - z)^{-1}(\bfx, \bfy)\right|, \ \left|(H_e - z)^{-1}(\bfx, \bfy)\right|\leq 2|\Im z|^{-1} e^{-\frac{\nu^4|\Im z|}{32}d_1(\bfx, \bfy)}.
    \end{split}
    \end{equation}
        \item By Lemma \ref{lemma-1b}, 
        \begin{equations}
        |G(\bfx, \bfy)| \leq 2\nu^{-1}e^{-2\nu \big(d_1(\bfx, \bfy) + d_1(\bfx, U^c) + d_1(\bfy, U^c)\big)}.
        \\
        |F(\bfx, \bfy)|\leq 2\nu^{-1}e^{-\frac{2\nu}{3}\big(d_1(\bfx, \bfy) + d_1(\bfx, \p U) + d_1(\bfy, \p U)\big)}.
        \end{equations}
    \item By short-range of $H_\pm$, $H_e$ and definition of $E$ in \eqref{eq-1a}, $E$ is ESR and decays away from $\p U$. By interpolating the two associated bounds, we get:
    \[
    \begin{split}
    |E(\bfx, \bfy)|\leq 2\nu^{-1}e^{-\frac{\nu}{2}\big(d_1(\bfx, \bfy) + d_1(\bfx, \p U) + d_1(\bfy, \p U)\big)}.
    \end{split}
    \]
\end{enumerate}

Note that all the decay rate in the above bounds can be relaxed to $\frac{\nu^4|\Im z|}{32}$ because $\nu < 1$ and $|\Im z| < 1$. Denote $r  := \frac{\nu^4}{32\times 4}$. Since $d_1(\bfx, \p U) \geq d_1(\bfx, U^c)$, we have
\begin{align}
    \big|(H_e - H_+)(\bfx, \bfy)\big| &= \big|(E + F + G)(\bfx, \bfy)\big|\leq 6\nu^{-1}e^{-\frac{\nu}{2}\big(d_1(\bfx, \bfy) + d_1(\bfx, U^c) + d_1(\bfy, U^c)\big)}\\
    &\leq 6\nu^{-1}e^{-r|\Im z|\big(d_1(\bfx, \bfy) + d_1(\bfx, U^c) + d_1(\bfy, U^c)\big)} \label{eq-6a}
\end{align}
By \eqref{eq-2d} in Lemma \ref{lemma-1c} and \eqref{eq-2e}, 
\[
\begin{split}
\big|(H_e - z)^{-1} (H_e - H_+) (H_+ - z)^{-1}(\bfx, \bfy)\big|&\leq \frac{CM_{d_1, r |\Im z|}^2}{|\Im z|^2} e^{-r |\Im z| \big(d_1(\bfx, \bfy) + d_1(\bfx, U^c)  +d_1(\bfy, U^c)\big)}\\
&\leq \frac{C}{|\Im z|^6}e^{-r  |\Im z|\big(d_1(\bfx, \bfy) + d_1(\bfx, U^c)  +d_1(\bfy, U^c)\big)}.
\end{split}
\]
As a result, using almost analyticity \eqref{eq: HS_1},  when $\bfx \neq \bfy$, let $w = |\Im z|$, 
\begin{align}
    \big| \big(\rho'(H_e) - \rho'(H_+)\big)(\bfx, \bfy)|&\leq   \int_{\supp(g) \times [-1,1]} \frac{C |\p_{\bar{z}}\tilde g|}{|\Im z|^6}e^{-r  |\Im z|\big(d_1(\bfx, \bfy) + d_1(\bfx, U^c) + d_1(\bfy, U^c)\big)}|dz\wedge d\bar{z}|\\
    &\leq C_{N}\int_0^\infty w^{N - 6} e^{-w r  \big(d_1(\bfx, \bfy) + d_1(\bfx, U^c) + d_1(\bfy, U^c)\big)} dw \\
    &\leq   C_{N}\big(d_1(\bfx, \bfy) + d_1(\bfx, U^c) + d_1(\bfy, U^c)\big)^{N - 5}.
    \end{align}
 When $\bfx = \bfy$, $\big|(\rho'(H_e) - \rho'(H_+))(\bfx,\bfy)\big|\leq \Vert \rho'(H_e) - \rho'(H_+)\Vert \leq 2\Vert g\Vert_\infty<+\infty$. Therefore, at the cost of potentially increasing $C_N$, we have for any $\bfx, \bfy$: 
\begin{align}
    \big|\big(\rho'(H_e) - \rho'(H_+)\big)(\bfx, \bfy)\big| 
    & \leq C_N\big(1 + d_1(\bfx, \bfy) + d_1(\bfx, U^c) + d_1(\bfy, U^c)\big)^{-3N} \\
& \leq C_Ne^{-N(d_l(\bfx, \bfy) + d_l(\bfx, U^c) + d_l(\bfy, U^c))}.
\label{eq-9d} 
\end{align}

3. Therefore, from \eqref{eq-9d}:
\begin{align}
    \big| \big( (\rho'(H_e) - \rho'(H_+))\1_U \big) (\bfx, \bfy)\big| &\leq C_N e^{-Nd_l(\bfx, \bfy) - Nd_l(\bfx, U^c) - Nd_l(\bfy, U^c)}\1_U(\bfy)\\
    &\leq C_N e^{-Nd_l(\bfx, \bfy) - Nd_l(\bfx, U^c) - Nd_l(\bfy, U^c)}\cdot e^{-Nd_l(\bfy, U)}\\
    &\leq C_N e^{-\frac{N}{2}d_l(\bfx, \bfy) - \frac{N}{2}d_l(\bfx, U^c) - \frac{N}{2}d_l(\bfy, U^c) - \frac{N}{2}d_l(\bfy, U) - \frac{N}{2}d_l(\bfx, U)}\\
    &\leq C_N e^{-\frac{N}{2}d_l(\bfx, \bfy) - \frac{N}{2}d_l(\bfx, \p U) - \frac{N}{2}d_l(\bfy, \p U)}.
\end{align}
The same upper bound holds for $\big(\rho'(H_e)-\rho'(H_-)\big) \1_{U^c}$. Going back to \eqref{eq-9e}, we obtain
\begin{equation}\label{eq-9h}
    \big| \rho'(H_e)(\bfx,\bfy) \big| \leq C_N e^{-\frac{N}{2}d_l(\bfx, \bfy) - \frac{N}{2}d_l(\bfx, \p U) - \frac{N}{2}d_l(\bfy, \p U)}.
\end{equation}
By \eqref{eq-2c} in Lemma \ref{lemma-1b}, for any $N$, there is $C_{N}$ such that 
    \begin{equation}
        \label{eq-1d}
    \big|[H_e, \1_V](\bfx,\bfy)\big| \leq C_{N}e^{-Nd_l(\bfx, \bfy) - Nd_l(\bfx, \p V) - Nd_l(\bfy, \p V)}.
    \end{equation}
Since $U, V$ are transverse sets, by Corollary \ref{cor-1a}, $\rho'(H_e)[H_e, \1_U]$ is trace-class; hence $\sigma_e(H_e)$ is well-defined.
\end{proof}

\section{Equality}\label{sec-4}

\begin{thm}\label{thm-4a}
    Under Assumption \ref{ass-all}, 
    \begin{equation}
        \label{eq-3p}
    \sigma_e^{U,V}(H_e) = \sigma_b^{U, V}(P_+) - \sigma_b^{U, V}(P_-).
    \end{equation}
\end{thm}

\begin{proof} 
\begin{figure}[b]
     \centering
     \begin{subfigure}[t]{0.3\textwidth}
         \centering
         \includegraphics[width=\textwidth]{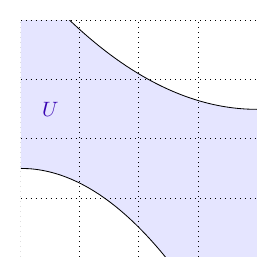}
         \caption{}
     \end{subfigure}
\qquad \qquad \begin{subfigure}[t]{0.3\textwidth}
         \centering
         \includegraphics[width=\textwidth]{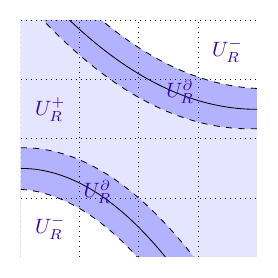}
         \caption{}
     \end{subfigure}
        \caption{The white, light blue, and dark blue region in $(b)$ represents $U_R^-$, $U_R^+$, and $U_R^\p$ respectively.}
        \label{fig-P4c}
\end{figure}
   We follow an approach due to Elgart--Graf--Schenker \cite{EGS05}, tailored here to our geometic setup. Define  
    \[
    U_R^{\p} :=  \{ \bfx : d(\bfx,\p U) < R\}, \qquad U_R^+ := \big(U_R^{\p}\big)^c \cap U, \qquad U_R^- :=\big(U_R^{\p}\big)^c \cap U^c,
   \]
    and correspondingly, 
        \begin{equation}
        \label{eq-4k}
    \zeta_R^{*}(H) := \Tr\left(\int \frac{\p {\tilde \rho}}{\p \bar{z}} (H - z)^{-1}[H, \1_V]\big[(H - z)^{-1}, \1_{U_R^{*}}\big]dz \wedge d\bar{z}\right), \qquad *\in\{\p, +, -\}.
    \end{equation}
    This quantity $\zeta_R^*(H)$ plays an intermediate role in proving equality \eqref{eq-3p}. In fact, we prove \eqref{eq-3p} by proving the following three equalities:
    \begin{equation}
        \label{eq-3u}
    \begin{split}
    \sigma_e(H_e) \xlongequal{\text{Claim~}\ref{claim-4a}} \lim_{R\to \infty} \zeta_R^\p(H_e) \xlongequal{\text{Claim~}\ref{claim-4b}} \lim_{R\to \infty} - \zeta_R^+(H_+) - \zeta_R^-(H_-) \xlongequal{\text{Claim~}\ref{claim-4c}} \sigma_b^{U, V}(P_+) - \sigma_b^{U, V}(P_-).
    \end{split}
    \end{equation}
    As explained in \S\ref{sec-1e}, to prove Claim \ref{claim-4a}, we inject a cutoff near $\p U$ and justify that terms away from $\p U$ are negligible. Claim \ref{claim-4b} reduces this near-$\p U$ edge quantity, $\zeta_R^\p(H_e)$, to the difference of two bulk quantities, $\zeta_R^\pm(H_\pm)$. Claim \ref{claim-4c} deform the bulk quantities, $\zeta_R^\pm(H_\pm)$ (in the limit), to the geometric bulk conductances, $\sigma_b^{U, V}(P_\pm)$. 
    
    Now we state and prove these three claims.

    \begin{claim}\label{claim-4a}
    \[\sigma_e(H_e) = \lim\limits_{R\to \infty} \zeta_R^\p(H_e).\]
\end{claim}
\begin{proof}[Proof of Claim \ref{claim-4a}]
    Since $\R^2 = U_R^\p \sqcup U_R^+\sqcup U_R^-$, we split $\sigma_e(H_e)$ into two parts:
    \begin{equation}
        \label{eq: I_and_II}
        2\pi i [H_e, \1_V] \rho'(H_e) = 2\pi i [H_e, \1_V]\big( \1_{U_R^\p} + (\1_{U_R^+} + \1_{U_R^-}) \big)\rho'(H_e):= \oI + \oII.
    \end{equation}    
    We show that $\oII$ is trace-class and that $\Tr(\oII)\to 0$ when $R\to +\infty$. By \eqref{eq-9h}, we have 
        \begin{align}
        \big|\1_{U_R^\pm} \rho'(H_e)(\bfx, \bfy)\big|&\leq C_N e^{- 2Nd_l\left(\bfx, (U_R^\p)^c\right) -2Nd_l(\bfx, \bfy) -2Nd_l(\bfx, \p U) - 2Nd_l(\bfy, \p U) }\\
        &\leq C_N (1 + R)^{-N} e^{-Nd_l(\bfx, \bfy) -Nd_l(\bfx, \p U) - Nd_l(\bfy, \p U) }, \label{eq-9g}
        \end{align}
        where we used $\big|\1_{U_R^\pm}(\bfx)\big|\leq e^{-2Nd_l \left(\bfx, (U_R^\p)^c \right)}$ and $d_l\left(\bfx, (U_R^\p)^c\right) + d_l(\bfx, \p U) \geq \ln (1 + R)$. Because of \eqref{eq-1d} and \eqref{eq-9g}, Corollary \ref{cor-1a} implies that $\oII$ is trace-class. We now estimate its trace. By \eqref{eq-1d} and \eqref{eq-2d} in Lemma \ref{lemma-1c}, we get 
        \begin{align}
        |\text{II}(\bfx, \bfy)| & \leq C_N (1 + R)^{-N} e^{-\frac{N}{4}d_l(\bfx, \bfy) -\frac{N}{4}d_l(\bfx, \p U) - \frac{N}{4}d_l(\bfy, \p U) - \frac{N}{4}d_l(\bfx, \p V)- \frac{N}{4}d_l(\bfy, \p V)} 
        \end{align}
        As a result, 
        \[
        \big| \Tr(\oII) \big| \leq \sum\limits_{\bfx} |\oII(\bfx, \bfx)|\leq C_N  (1 + R)^{-N} \cdot \sum_{\bfx} e^{- \frac{N}{2} \ln \Psi_{U,V}(\bfx)}.
        \]
        By \eqref{eq-0v}, the sum is finite for $N> 4/c$. It follows that $\Tr(\oII)\to 0$ as $R \rightarrow \infty$. 
        
        It remains to prove $\Tr(\oI) = \zeta_R^\p(H_e)$. Without loss of generalities, $\widetilde {\p_z \rho} = \p_z \tilde \rho$ -- see e.g. \cite[(2.3.5)]{D21}. Combining the Helffer--Sj\"ostrand formula \eqref{eq: HS_1} with integration by parts, we obtain 
    \begin{align}
        \oI = 2\pi i[H_e, \1_V] \1_{U_R^\p} \cdot \rho'(H_e) &= [H_e, \1_V] \1_{U_R^\p} \cdot \int \frac{\p {\tilde\rho}'}{\p \bar{z}}   (H_e - z)^{-1} dz\wedge d\bar{z}\\
        &= [H_e, \1_V] \1_{U_R^\p}  \cdot \int \frac{\p^2 {\tilde\rho}}{\p z\p \bar{z}}  (H_e - z)^{-1} dz \wedge d\bar{z}\\
        &= - \int \frac{\p {\tilde\rho}}{\p \bar{z}}  [H_e, \1_V] \1_{U_R^\p} \cdot (H_e - z)^{-2} dz \wedge d\bar{z}.\label{eq-6c}
    \end{align}   
   
    If $\bfx\in U_R^\p$, then $d_l(\bfx, \p U)\leq \ln(1 +R)$. Thus for any $N$, 
    \begin{equation}
        \label{eq-0n}
    |\1_{U_R^\p}(\bfx)|\leq  e^{N\big(\ln(1 + R)- d_l(\bfx, \p U)\big)}\leq (1 + R)^N e^{-Nd_l(\bfx, \p U) }.
    \end{equation}
    Combining with \eqref{eq-1d} and using Corollary \ref{cor-1a}, we see that $[H_e, \1_V] \1_{U_R^\p} $ is trace-class. Using $\| (H-z)^{-1} \| \leq |\Im z|^{-1}$, we have:
    \begin{equation}
      \left\|  [H_e, \1_V] \1_{U_R^\p} (H_e - z)^{-2} \right\|_1 \leq \dfrac{C}{|\Im z|^2}.
    \end{equation}
    Therefore, we can take the trace on both sides of \eqref{eq-6c} and switch trace and integral (using almost analyticity of $\tilde\rho$). This yields:
    \begin{equation}
        \Tr (\oI) =  - \int \frac{\p {\tilde\rho}}{\p \bar{z}}  \Tr \big( [H_e, \1_V] \1_{U_R^\p} (H_e - z)^{-2} \big) dz \wedge d\bar{z}.
    \end{equation}
    By \eqref{eq-1d}, \eqref{eq-0n}, and Corollary \ref{cor-1a}, the operators below are all trace-class. Thus we can apply cyclicity \eqref{eq-0l} to get
    \[
    \begin{split}
    \Tr(\oI) &  = -\int \frac{\p {\tilde\rho}}{\p \bar{z}} \Tr\left((H_e - z)^{-1} [H_e , \1_V]  \1_{U_R^\p} (H_e - z)^{-1}\right) dz \wedge d\bar{z}\\
    & = -\int \frac{\p {\tilde\rho}}{\p \bar{z}} \Tr\left( (H_e - z)^{-1}[H_e, \1_V](H_e - z)^{-1}  \1_{U_R^\p} \right) dz \wedge d\bar{z}\\
    & \ \ \ + \int \frac{\p {\tilde\rho}}{\p \bar{z}} \Tr\left((H_e - z)^{-1}[H_e, \1_V]\big[(H_e - z)^{-1},  \1_{U_R^\p} \big]\right)dz \wedge d\bar{z}
    \\ 
    & = \int \frac{\p {\tilde\rho}}{\p \bar{z}} \Tr\left( \big[(H_e - z)^{-1}, \1_V\big] \1_{U_R^\p} \right) dz \wedge d\bar{z} + \zeta_R^\p(H_e)\\
    &= \Tr\big(\big[\rho(H_e), \1_V\big] \1_{U_R^\p} \big) + \zeta_R^\p(H_e),
    \end{split}
    \]
 where we used Helffer-Sj\"ostrand formula \eqref{eq: HS_1} for the last equality. Finally, 
 \[
    \begin{split}
         \Tr\big(\big[\rho(H_e), \1_V\big] \1_{U_R^\p} \big) &= \Tr\big(\1_{V^c} \rho(H_e) \1_V  \1_{U_R^\p}
        \big)  - \Tr\big(\1_V  \rho(H_e)\1_{V^c} \1_{U_R^\p} \big)\\
        &= \Tr\big(\1_V \1_{V^c} \rho(H_e)  \1_{U_R^\p}\big)   - \Tr\big(\1_V \1_{V^c} \rho(H_e) \1_{U_R^\p}  \big) = 0
    \end{split}
    \]    
    as $\1_{V^c}\rho(H_e)\1_V  \1_{U_R^\p} $,  $\1_V\rho(H_e)\1_{V^c}  \1_{U_R^\p} $, and $\1_V\1_{V^c}\rho(H_e) \1_{U_R^\p}$ are all trace-class by Lemma \ref{lemma-1a}(b) and Corollary \ref{cor-1a}. Thus we get $ \Tr(\oI) = \zeta_R^\p(H_e)$; this completes the proof of Claim \ref{claim-4a}.
\end{proof}

\begin{claim}\label{claim-4b} We have:
\begin{equation}
    \label{eq-3s}
    \lim\limits_{R\to \infty} \big(\zeta_R^\p(H_e) +\zeta_R^+(H_+) + \zeta_R^-(H_-)\big) = 0.
\end{equation}
\end{claim}

\begin{proof}[Proof of Claim \ref{claim-4b}]
    Recall that $ \1_{U_R^\p} + \1_{U_R^+}  +\1_{U_R^-} = 1 $. Since $\big[1, (H_e - z)^{-1}\big] = 0$, from the definition \eqref{eq-4k}, we get $\zeta_R^\p(H_e) + \zeta_R^+(H_e) + \zeta_R^-(H_e) =0$. 
    To show \eqref{eq-3s}, it is enough to show 
    \begin{equation}
        \label{eq-3t}
    \lim_{R\to\infty} |\zeta_R^\pm(H_e) - \zeta_R^\pm(H_\pm)|= 0.
    \end{equation}
    Denote $A_R^+(H) = (H - z)^{-1}[H, \1_V]\big[(H - z)^{-1}, \1_{U_R^+} \big]$. We have:
    \begin{align}
        \big|\zeta_R^\pm(H_e) - \zeta_R^\pm(H_\pm)\big| & = \left| \Tr \int\frac{\p {\tilde\rho}}{\p \bar{z}} \left(A_R^+(H_e) - A_R^+(H_+)\right) dz \wedge d\bar{z} \right|  \\
        & \leq \int \left|\frac{\p {\tilde\rho}}{\p \bar{z}} \right| \left\Vert \left(A_R^+(H_e) - A_R^+(H_+)\right) \right\Vert_1 |dz \wedge d\bar{z}|. \label{eq-6d}
    \end{align}

    We rewrite
    \begin{equation}
        \label{eq-4b}
    \begin{split}
        A_R^+(H_e) - A_R^+(H_+) =& \left((H_e - z)^{-1} - (H_+ - z)^{-1}\right)[H_e, \1_V]\big[(H_e - z)^{-1}, \1_{U_R^+} \big] \\
        +& (H_+ - z)^{-1}[H_e - H_+, \1_V]\big[(H_e - z)^{-1}, \1_{U_R^+}  \big]\\
        +& (H_+ - z)^{-1}[H_+, \1_V]\big[(H_e - z)^{-1} - (H_+ - z)^{-1}, \1_{U_R^+}  \big]\\
        =& (H_e - z)^{-1}(H_+ - H_e)(H_+ - z)^{-1}[H_e, \1_V]\big[(H_e - z)^{-1}, \1_{U_R^+} \big] \\
        +& (H_+ - z)^{-1}[H_e - H_+, \1_V]\big[(H_e - z)^{-1}, \1_{U_R^+}  \big]\\
        +& (H_+ - z)^{-1}[H_+, \1_V]\big[(H_e - z)^{-1}(H_+ - H_e)(H_+ - z)^{-1}, \1_{U_R^+}  \big]\\
        =:& A_1 + A_2 + A_3.
    \end{split}
    \end{equation}

    Note that:
    \begin{enumerate}
   \item By \eqref{eq-2c} and using that $H_\pm$, $H_e$ are PSR, for any $N$, there exists $C_N$ with
    \[
    \big|[H_*, \1_V](\bfx, \bfy)\big|\leq C_N e^{-8Nd_l(\bfx, \bfy) - 8Nd_l(\bfx, \partial V) - 8Nd_l(\bfx, \partial V)},\qquad *\in \{\pm,\text{e}\}.
    \]
    \item By an argument similar to that leading to \eqref{eq-6a}, for any $N$, there is $C_N$ such that 
    \[
    \big|(H_e - H_+)(\bfx, \bfy)\big|\leq C_Ne^{-8Nd_l(\bfx, \bfy) - 8Nd_l(\bfx, U^c) - 8d_l(\bfy, U^c)}. 
    \]
    \item Recall that $\1_{U_R^+} (\bfx) \leq e^{-8Nd_l(\bfx, U_R^+)}$. 
    \end{enumerate}
    
    Note that each $A_j$, $j\in \{1,2,3\}$ contains $H_e - H_+$, $[H_*, \1_V]$ and $\1_{U_R^+} $. By $(1)$, $(2)$, $(3)$ above, the control of the resolvent norm provided by Lemma \ref{lemma-1a}(a) and \eqref{eq-2d} in Lemma \ref{lemma-1c}, we have 
    \[
    \begin{split}
        |A_j(\bfx, \bfy)|&\leq \frac{C_N}{|\Im z|^3} e^{-2N\big(d_l(\bfx, \bfy) + d_l(\bfx, U^c) + d_l(\bfy, U^c) + d_l(\bfx, \p V) + d_l(\bfy, \p V) + d_l(\bfx, U_R^+) + d_l(\bfy, U_R^+)\big)}\\
        &\leq \frac{C_N}{|\Im z|^3}(1 + R)^{-2N} e^{-N(d_l(\bfx, \bfy)+ d_l(\bfx, \partial U) + d_l(\bfx, \p V) + d_l(\bfy, \p U) + d_l(\bfy, \p V))}
        \end{split}
    \]
    where we used $d_l(\bfx, U^c) + d_l(\bfx, U_R^+) \geq d_l(U^c, U_R^+)\geq  \ln(1 + R)$ and $d_l(\bfx, U^c) + d_l(\bfx, U_R^+) \geq d_l(\bfx, \partial U)$.  Therefore,
    \[
    \big|\big(A_R^+(H_e) - A_R^+(H_+)\big)(\bfx, \bfy)\big|\leq \frac{C_N}{|\Im z|^3} (1 + R)^{-2N} e^{-N(d_l(\bfx, \bfy)+ d_l(\bfx, \partial U) + d_l(\bfx, \partial V) + d_l(\bfy, \p U) + d_l(\bfy, \p V))}.
    \]
    Since $U, V$ are transverse sets, when $N \geq 2/c$, by \eqref{eq-0v}, we obtain:
    \[
    \left \Vert A_R^+(H_e) - A_R^+(H_+)\right \Vert_1 \leq \sum\limits_{\bfx, \bfy}\big|\big(A_R^+(H_e) - A_R^+(H_+)\big)(\bfx, \bfy)\big|\leq \frac{C_N}{|\Im z|^3} (1 + R)^{-2N}.
    \]
    By almost analyticity, 
    \[
    \int \left|\frac{\p{\tilde\rho}}{\p \bar{z}} \right| \left\Vert \left(A_R^+(H_e) - A_R^+(H_+)\right) \right\Vert_1 |dz \wedge d\bar{z}| \leq C_N (1+R)^{-2N}.
    \]
    It suffices to go back to \eqref{eq-6d} and take $R \rightarrow \infty$ to obtain \eqref{eq-3t}. This completes the proof.
\end{proof}

\begin{claim}\label{claim-4c}
    For any $R>0$, 
    \begin{equation}
        \label{eq-3q}
    \zeta_R^\pm(H_\pm) = -\sigma_b^{U_R^\pm,V}(P_\pm) = \mp\sigma_b^{U,V}(P_\pm).
    \end{equation}
\end{claim}
\begin{proof}[Proof of Claim \ref{claim-4c}]
Note that $U_R^+ \Delta U\subset U_R^\p$ and $U_R^- \Delta U^c\subset U_R^\p$. By Proposition \ref{prop-2c}, 
\[
\begin{split}
&\sigma_b^{U_R^+, V}(P_+) = \sigma_b^{U, V}(P_+),\qquad \sigma_b^{U_R^-, V}(P_-) = \sigma_b^{U^c, V}(P_-) = -\sigma_b^{U, V}(P_-). 
\end{split}
\]
Hence we get the second equality in \eqref{eq-3q}. WLOG we prove $\zeta_R^+(H_+) = -\sigma_b^{U_R^+, V}(P_+)$ below. Recall that 
\begin{equation}
\begin{split}
    \zeta_R^{+}(H_+) := \Tr\left(\int_\C \frac{\p {\tilde \rho}}{\p \bar{z}} (H_+ - z)^{-1}[H_+, \1_V]\big[(H_+ - z)^{-1}, \1_{U_R^{+}}\big]dz \wedge d\bar{z}\right)=:\Tr(A).
\end{split}
\end{equation}
Expand $[(H_+ - z)^{-1}, \1_{U_R^{+}}]$ and use $(H_+ - z)^{-1}[H_+, \1_V](H_+ - z)^{-1} = -\big[(H_+ - z)^{-1}, \1_V\big]$, we get 
\[
\begin{split}
A &= \int_\C \frac{\p {\tilde \rho}}{\p \bar{z}}\left(-\big[(H_+ - z)^{-1}, \1_V\big] \1_{U_R^+} - (H_+ - z)^{-1}[H_+, \1_V]\1_{U_R^+}(H_+ - z)^{-1}\right)dz \wedge d\bar{z}\\
&=: -A_1 - A_2.
\end{split}
\]
By the Helffer-Sj\"ostrand formula \eqref{eq: HS_1}, we have 
\begin{equation}
    \label{eq-3r}
A_1 = \big[\rho(H_+), \1_V\big]\1_{U_R^+}.
\end{equation}

To conlude, we adapt a trick from \cite{EGS05}.

\begin{lemma}\label{lemma-3c}
    We have $P_+ A_2 P_+ = P_+^\perp A_2 P_+^\perp = 0$.
\end{lemma}

\begin{proof}
By functional calculus,
    \[
    \begin{split}
            P_+A_2 P_+ &= \int_\C \frac{\p {\tilde \rho}}{\p \bar{z}}P_+(H_+ - z)^{-1}[H_+, \1_V]\1_{U_R^{+}}(H_+ - z)^{-1} P_+ dz \wedge d\bar{z}\\
            &= \int_\C \frac{\p {\tilde \rho}}{\p \bar{z}}P_+(P_+H_+P_+ - z)^{-1} [H_+, \1_V]\1_{U_R^{+}}(P_+H_+P_+ - z)^{-1} P_+ dz \wedge d\bar{z}
            \\
            & =:\int_\C \frac{\p {\tilde \rho}}{\p \bar{z}} K(z) dz\wedge d\bar{z}.
    \end{split}   
    \]
    Since $\lambda\in \MG\subset \Sigma(H_+)^c$ while $\rho\in \ME_\MG$, $\tilde \rho$ is the almost analytic extension of $\rho$, $\Sigma(P_+H_+P_+)\cap \supp(\tilde\rho) = \emptyset$. Hence $(P_+H_+P_+ - z)^{-1}$, and thus the whole integrand, denoted as $K(z)$, is analytic on $\supp(\tilde\rho)$. Choose $R$ large enough such that $\supp(\tilde \rho)\subset \Dd_R(0)$. By Stokes theorem, 
    \[
    \int_\C \frac{\p {\tilde \rho}}{\p \bar{z}} K(z) dz\wedge d\bar{z} = \oint_{\p \Dd_R(0)} \tilde \rho(z, \bar{z})K(z) dz = 0. 
    \]
    A similar computation gives 
    \[
    \begin{split}
        P_+^\perp A_2P_+^\perp &= \int_\C \frac{\p (1 - \tilde \rho)}{\p \bar{z}} P_+^\perp (P_+^\perp H_+P_+^\perp  - z)^{-1}[H_+, \1_V]\1_{U_R^{+}} (P_+^\perp H_+P_+^\perp  - z)^{-1} P_+ dz\wedge d\bar{z}\\
        &=:\int_\C \frac{\p (1 - \tilde \rho)}{\p \bar{z}} J(z) dz \wedge d\bar{z}\\
        &= \oint_{\p \Dd_R(0)} (1 - \tilde\rho)(z, \bar{z}) J(z) dz = \oint_{\p \Dd_R(0)} J(z) dz.
    \end{split}
    \]
    However, when $\bfz\in \p \Dd_R(0)$, 
    \[
    \Vert J(z) \Vert \leq C d\big(z, \Sigma(H_+)\big)^{-2} \leq C(R - \Vert H \Vert)^{-2}. 
    \]
    Hence 
    \[
    \Vert P_+^\perp A_2 P_+^\perp \Vert \leq \frac{CR}{(R - \Vert H_+\Vert)} \to 0, \qquad R\to \infty. 
    \]
    This completes the proof. 
\end{proof}

Now write $A = P_+^2 A + (P_+^\perp)^2A$. Since $A$ is trace-class, Lemma \ref{lemma-3c} and cyclicity \eqref{eq-0l} yields:
\[
\begin{split}
    \Tr(A) &= \Tr\big(P_+^2A + (P_+^\perp)^2 A\big) = \Tr\big(P_+AP_+ + P_+^\perp A P_+^\perp\big)\\
    &= -\Tr\big(P_+A_1P_+ + P_+^\perp A_1 P_+^\perp\big). 
\end{split}
\]
Using $A_1 = [\rho(H_+), \1_V]\1_{U_R^+}$ and $\rho(H_+) = 1 - P_+ = P_+^\perp$, we see 
\[
\begin{split}
    P_+A_1P_+ + P_+^\perp A_1 P_+^\perp &= P_+[P_+^\perp, \1_V]\1_{U_R^+}P_+ + P_+^\perp [P_+^\perp, \1_V]\1_{U_R^+}P_+^\perp\\
    &= -P_+\1_VP_+^\perp \1_{U_R^+}P_+ + P_+^\perp \1_V \1_{U_R^+}P_+^\perp - P_+^\perp \1_V P_+^\perp \1_{U_R^+}P_+^\perp\\
    &= -P_+\1_VP_+^\perp \1_{U_R^+}P_+ + P_+^\perp \1_V P_+ \1_{U_R^+}P_+^\perp
\end{split}
\]
where 
\[
\begin{split}
    \Tr\big(P_+^\perp \1_V P_+ \1_{U_R^+} P_+^\perp\big) &= \Tr\big(P_+^\perp \1_V P_+ \cdot P_+\1_{U_R^+} P_+^\perp\big) = \Tr\big([P_+^\perp, \1_V]P_+ \cdot P_+[\1_{U_R^+}, P_+^\perp]\big)\\
    &= \Tr\big(P_+[\1_{U_R^+}, P_+^\perp][P_+^\perp, \1_V]P_+\big) = \Tr\big(P_+\1_{U_R^+} P_+^\perp \1_V P_+\big).
\end{split}
\]
Thus 
\[
\begin{split}
    \zeta_R^+(H_+) &= \Tr(A) = -\Tr\big(P_+A_1P_+ + P_+^\perp A_1 P_+^\perp\big) = \Tr\big(P_+\1_{V} P_+^\perp \1_{U_R^+} P_+-P_+\1_{U_R^+}P_+^\perp \1_V P_+\big)\\
    &=  \Tr\big(P_+\1_{U_R^+} P_+\1_V P_+ - P_+\1_VP_+\1_{U_R^+}P_+ \big).
\end{split}\]
On the other hand, by direct computation, 
\[
\sigma_b^{U_R^+, V}(P_+) = -\Tr\big(P_+\big[[P_+, \1_{U_R^+}],[P_+, \1_V]\big]\big) = \Tr\big(P_+\1_V P_+ \1_{U_R^+} P_+ - P_+\1_{U_R^+}P_+\1_VP_+ \big). 
\]
Thus $\zeta_R^+(H_+) = -\sigma_b^{U_R^+, V}(P_+)$. 
\end{proof}
\end{proof}

\section{Intersection number between simple transverse sets}\label{sec-5}
In this section, we define an intersection number $\II_{U, V}$ between simple transverse sets $U, V$ (see Definition \ref{def-8}). This integer will emerge when expressing geometric bulk conductances in terms of Hall conductances: see  Theorem \ref{thm:1} below.

\begin{definition}[Simple path]\label{def:2} A continuous map $\gamma : \R \rightarrow \R^2$ is called a simple path if:
\begin{itemize}
    \item[(i)] $\gamma$ is injective and proper (the preimage of a compact set is compact);
    \item[(ii)] There exists a discrete closed set $S \subset \R$ such that $\gamma$ is smooth on $\R \setminus S$; 
    \item[(iii)] For all $t \in \R$, the left and right derivatives of $\gamma$ at $t$ exist and have norm $1$.
\end{itemize}
If (ii) and (iii) hold but $\gamma$ is periodic and injective over its period, we call $\gamma$ a simple loop.
\end{definition}

\begin{prop}\label{prop:4} Let $\gamma_1, \gamma_2$ be two simple paths with ranges $\Gamma_1, \Gamma_2$, such that $\Gamma_1\Delta \Gamma_2$ is compactly supported. Let $A_1$ be a connected component of $\R^2 \setminus \Gamma_1$. Then there exists a unique connected component $A_2$  of $\R^2 \setminus \Gamma_2$ such that $A_1 \Delta A_2$ is bounded.
\end{prop}
\begin{figure}[ht] 
\floatbox[{\capbeside\thisfloatsetup{capbesideposition={right,center},capbesidewidth=0.6\textwidth}}]{figure}[\FBwidth]
    {\caption{Simple path splits $\R^2$ into two halves. When two simple paths $\Gamma_1$, $\Gamma_2$ differ by a compact set, given $A_1$ connected component of $\R^2 \setminus \Gamma_1$, there is a connected component $A_2$ of $\R^2 \setminus \Gamma_2$ such that $A_1$, $A_2$ differ by a compact set. }} 
    {\includegraphics[width=0.9\linewidth]{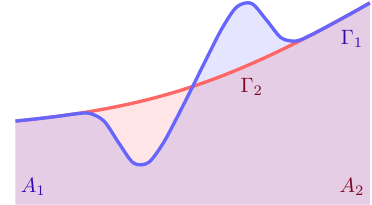} \label{fig-prop}}
\end{figure}

Proposition \ref{prop:4} has a flavor reminiscent of the Jordan curve Theorem. While visually obvious (see Figure \ref{fig-prop}), its proof requires some work. We give a sketch here and defer the proof to Appendix \ref{app:A}:
\begin{itemize}
    \item We show that if $\gamma_1$ is a proper injective curve, then $\R^2 \setminus \gamma_1(\R)$ has two connected components (Lemma \ref{lem:4});
    \item We then justify that a perturbation $\Gamma_2$ of $\Gamma_1$ on a compact set perturbs the two connected components by a bounded set only (Lemma \ref{lem:5});
    \item We unambiguously define the left side 
    of a simple path (Proposition \ref{prop:1}) and justify that the left components of $\R^2 \setminus \gamma_1(\R)$, $\R^2 \setminus \gamma_2(\R)$ have bounded symmetric difference (see Figure \ref{fig:P_left} for a pictorial representation and Proposition \ref{prop:1} for a proper definition of the left of a simple curve).
\end{itemize}
\begin{figure}[ht] 
\floatbox[{\capbeside\thisfloatsetup{capbesideposition={right,center},capbesidewidth=0.6\textwidth}}]{figure}[\FBwidth]
    {\caption{At each point on $\gamma$, we can find small enough disk split in two by $\gamma$ and define the left side of $\gamma$ as the ``conventional" left side when travelling along the path $\gamma$.}} 
    {\includegraphics[width=0.85\linewidth]{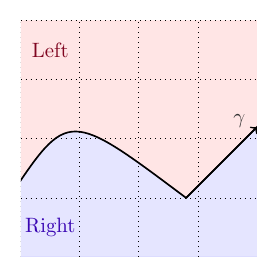} \label{fig:P_left}}
\end{figure}

\begin{definition}[Simple set]\label{def-8} An open subset $A$ of $\R^2$ is simple if it is connected and its boundary is the range of a simple path or of a simple loop. \end{definition}

Associated with the distinction between simple paths and simple loops, there are two types of simple sets: those with bounded boundaries (given by simple loops) and those with unbounded boundaries (given by simple paths). By Jordan's theorem, those with bounded boundaries are bounded or have bounded complements. 

\subsection{Intersection number}\label{sec-5.2}

\begin{definition}[Intersection number for simple sets]\label{def:5} Let $U,V$ be transverse sets such that $U$ is simple. If $\p U$ is unbounded, let $\gamma$ be a simple path with range $\p U$, such that $U$ lies to the left of $\gamma$ (see Figure \ref{fig:P_left}  for a pictorial representation and Proposition \ref{prop:1} for a proper definition). We define the intersection number between $U, V$ as:
\begin{equation}\label{eq2-6e}
\II_{U, V}  \de \II_+(U,V) - \II_-(U,V), \qquad \II_\pm(U,V) \de \lim_{t \rightarrow \pm \infty}   \1_V \circ \gamma(t).
\end{equation}
If $\p U$ is bounded, we set $\II_{U, V}  = 0$.
\end{definition}

We show in Appendix \ref{app:B} that $\II_{U, V}$ is correctly defined and independent of the choice of curve $\gamma$ describing its boundary. In \S\ref{sec-7.2} we extend Definition \ref{def:5} to more general transverse sets $U,V$, such that (loosely speaking) $\p U$ is made of well-separated curves.

\section{Bulk index computation}\label{sec-6}

\subsection{Theorem \ref{thm:1} for transverse simple sets}

The main theorem here is:

\begin{theorem}\label{thm:1} Under Assumption \ref{ass-all}, we have
\begin{equation}
\sigma_b^{U, V}(P_\pm) = \II_{U, V} \cdot \sigma_b(P_\pm).
\end{equation}
\end{theorem}

Since $U$ is simple, by Definition \ref{def:5}, $\II_{U, V} \in \{0,\pm 1\}$. It suffices to prove the following: 
\begin{itemize}
    \item[\textbf{1.}] If $\II_{U, V} = 0$, then $\sigma_b^{U, V}(P) = 0$ (Lemma \ref{lem:16}).
    \item[\textbf{2.}] If Theorem \ref{thm:1} holds when $\II_{U, V} = 1$, then it also holds when  $\II_{U, V} = -1$ (Lemma \ref{lem:8}).
    \item[\textbf{3.}] Theorem \ref{thm:1} holds when $\II_{U, V} = 1$ (\S\ref{sec:3.5} - \S\ref{sec-5.6}).
\end{itemize}
The core of the proof is the third step.

\subsection{Steps 1 and 2}

\begin{lemma}\label{lem:16} Let $U,V$ be transverse simple sets such that $\II_{U, V} = 0$, and $P$ be an ESR projector. Then $\sigma_b^{U, V}(P) = 0$.      
\end{lemma}

\begin{proof} We recall the notation $A^\CCC = \Int A^c$.

1.  Assume first that $\p U$ is bounded. Since $U$ is a simple set, by definition, $\p U$ is a simple loop. Hence either $U$ or $U^\CC$ is bounded by Jordan's theorem. In the former case, Proposition \ref{prop-2c} implies $\sigma_b^{U, V}(P) = \sigma_b^{\emptyset, V}(P) = 0$. In the latter case, Proposition \ref{prop-2c} implies $\sigma_b^{U, V}(P) = \sigma_b^{\R^2, V}(P) = 0$. Therefore we assume $\p U$ is unbounded in the rest of the proof.

2. Since $\II_{U, V} = 0$, $\II_+(U,V) = \II_-(U,V)$; potentially replacing $V$ by $V^\CCC$ we may assume that $\II_+(U,V) = \II_-(U,V) = 1 = \lim\limits_{t\to \pm \infty} \1_V\circ \gamma(t)$. Hence when $T$ is large enough, $\gamma(\{ |t| > T \}) \subset V$.  Given $R>0$, define 
\[
V_R \de \Dd_R(V) = \big\{ \bfx \in \R^2: \  d_l(\bfx,V) < R \big\}. 
\]
When $R>\max \big\{ d_l(\gamma(t),V) : t \in [-T,T] \big\}$, we have $\p U \subset V_R$.

3. By continuity of distance function, $d_l(\p V_R, V) \leq R$ and $d_l(\p V_R, \p V)\leq R$. As a result, 
\[
\begin{split}
    d_l(\bfx, \p V_{2R})&\geq   d_l(\bfx, \p V) -  d_l(\p V, \p V_{2R}) \geq d_l(\bfx, \p V) - 2R
\end{split}
\]
By \eqref{eq-0x}, 
\begin{align}
      2\ln \Psi_{U, V_{2R}}(\bfx) &\geq  d_l(\bfx, \p V_{2R})  +  d_l(\bfx, \p U)\\
      &\geq d_l(\bfx, \p V) + d_l(\bfx, \p U) - 2R \geq \ln \PUV(\bfx) - 2R. \label{eq2-8c}
\end{align}
On the other hand, if $d_l(\p V_R, V)<R$, then there is $\bfx \in \p V_R$ such that $d_l(\bfx, V)<R$, by definition this implies $\bfx \in V_R$. But $V_R$ is an open set so $\bfx \in V_R \cap \p V_R = \emptyset$. We get a contradiction; thus $d(\p V_R, V) = R$. In particular, for any $\bfx \in \p V_{2R}$, 
\[
    d_l(\bfx, V_R) \geq d_l(\bfx, V_R) - d_l(V, V_R) \geq 2R - R = R \quad \Rightarrow \quad d_l(V_R, \p V_{2R}) \geq R. 
\]
As a result, by \eqref{eq-0x},
\begin{equation}\label{eq2-8b}
2\ln \Psi_{U, V_{2R}}(\bfx) \geq d_l(\p U, \p V_{2R}) \geq d_l(V_R, \p V_{2R})\geq R. 
\end{equation}
Interpolating between \eqref{eq2-8c} and \eqref{eq2-8b} gives:
\begin{equation}\label{eq2-5b}
2 \ln \Psi_{U, V_{2R}}(\bfx) \geq \dfrac{1}{4} (\ln \Psi_{U,V}(\bfx) - 2R) + \dfrac{3}{4} R = \dfrac{\ln \Psi_{U,V}(\bfx)+R}{4}. 
\end{equation}

4. Since $V_{2R} \Delta V \subset \{\bfx: d_l(\bfx, \p V) \leq 2R\}$, by Proposition \ref{prop-2c}:
\begin{equation}
    \big| \sigma_b^{U, V}(P) \big|  =  \big| \sigma_b^{U, V_{2R}}(P) \big|.
\end{equation}
By \eqref{eq-2q} and \eqref{eq2-5b}, we have 
\begin{equation}
\big| \sigma_b^{U, V_{2R}}(P) \big| \leq C_N \left(\sum_{\bfx \in \Z^2} e^{-N\ln \Psi_{U, V_{2R}}(\bfx)}\right)^2 \leq C_N e^{-\frac{NR}{4}} \left(\sum_{\bfx\in \Z^2} e^{-\frac{N}{8} \ln\Psi_{U,V}(\bfx)} \right)^2. 
\end{equation}

Since $U, V$ are transverse sets, by \eqref{eq-0v}, the last sum is finite when $N$ is large enough and the sum does not depend on $R$. Therefore, taking $R \rightarrow \infty$ yields $\sigma_b^{U, V}(P) = 0 = \II_{U, V} \cdot \sigma_b(P)$. This proves Theorem \ref{thm:1} when $\II_{U, V} = 0$. \end{proof}

\begin{lemma}\label{lem:8} Let $P$ be an ESR projector. Assume that for all transverse simple sets $U,V$ such that $\II_{U, V} = 1$, we have $\sigma_b^{U, V}(P) = \sigma_b(P)$. Then for all transverse simple sets $U,V$ such that $\II_{U, V} = -1$, we have $\sigma_b^{U, V}(P) = -\sigma_b(P)$.
\end{lemma}

\begin{proof} Assume that $\II_{U, V} = -1$. Then, $\II_{U, V^c} = 1$, and so by assumption we have $\sigma_b^{U, V^c}(P) = \sigma_b(P)$. But since $\sigma_b^{U, V^c}(P) = -\sigma_b^{U, V}(P)$, we deduce that $\sigma_b^{U, V}(P) = -\sigma_b(P)$.\end{proof}

Therefore, it remains to prove that Theorem \ref{thm:1} holds for pairs $(U,V)$ such that $\II_{U, V} = 1$.

\subsection{Uniformly transverse families}\label{sec:3.5} 

Recall that transversality condition \eqref{eq-0a} is equivalent to \eqref{eq-0b}:
\begin{equation}\label{eq2-0b}
    \exists c\in (0,1) \ \text{such that, } \ \forall |\bfx|\geq c^{-1}, \ \PUV(\bfx) \geq |\bfx|^c.
\end{equation}
When \eqref{eq2-0b} holds, we refer to $(U,V)$ as \textit{$c$-transverse}.

\begin{definition}[Uniformly transversality] Let $\FF = \big\{ (U_n,V_n): n \in \N\big\}$ be a family of transverse simple sets. We say that $\FF$ is uniformly transverse if there exists $c \in (0,1)$ such that for all $n$, $(U_n,V_n)$ is $c$-transverse.
We say that $\FF$ is equivalent to $(U,V)$ if for all $n$, $U \Delta U_n$ and $V \Delta V_n$ are bounded.    
\end{definition}

Uniform transversality boils down to \eqref{eq2-0b} holding uniformly in $n$:
\begin{equation}\label{eq2-0g}
    \exists c \in (0,1) \text{ such that, } \ \forall n, \ \ \forall |\bfx| \geq c^{-1}, \ \Psi_{U_n, V_n}(\bfx) \geq |\bfx|^c.
\end{equation}

Recall that $\Hh_i:= \{ x_i\geq 0\}$, $i = 1,2$ denote the right/upper half-plane. 

Our strategy to prove Theorem \ref{thm:1} goes as follows. We will construct a family of uniformly transverse sets $(U_n, V_n)$ (see Proposition \ref{lem-10} and \S \ref{sec-5.6}) such that $(U_n,V_n)$ are compact perturbations of $(U,V)$ and they look like $(\Hh_2, \Hh_1)$ inside the disk $\Dd_n(0)$. Because $(U_n,V_n)$ are compact perturbations of $(U,V)$, we will have $\sigma_b^{U_n, V_n}(P) = \sigma_b^{U, V}(P)$ by the robustness of geometric bulk conductance (Proposition \ref{prop-2b}). Meanwhile, since $(U_n, V_n)$ look like $\Hh_2$ and $\Hh_1$ in $\Dd_n(0)$, we will show $\sigma_b^{U_n, V_n}(P)$ converges to $ \sigma_b^{\Hh_2, \Hh_1}(P)$ as $n\to \infty$ by showing
\begin{itemize}
    \item $K_{U_n,V_n}(\bfx,\bfx) \to K_{\Hh_2,\Hh_1}(\bfx,\bfx)$ as $n\to \infty$ for each $\bfx$ (Lemma \ref{lem-11}); 
    \item the convergence is uniform because of uniform transversality. 
\end{itemize}
This will prove that $ \sigma_b^{U_n, V_n}(P)$ converges as $n \rightarrow \infty$ to $ \sigma_b^{\Hh_2,\Hh_1}(P)$, and in particular $\sigma_b^{U, V}(P) = \lim\limits_{n\to\infty} \sigma_b^{U_n, V_n}(P) = \sigma_b^{\Hh_2,\Hh_1}(P) = \sigma_b(P)$.

\begin{lemma}\label{lem-11} Assume that $\{ (U_n,V_n): \ n \in \N\}$ is a family of transverse sets in $\R^2$ such that for all $n$, $U_n \cap \Dd_n(0) = \Hh_2 \cap \Dd_n(0)$ and $V_n \cap \Dd_n(0) = \Hh_1 \cap \Dd_n(0)$. Then for any $\bfx \in \Z^2$:
\begin{equation}
    \lim_{n \rightarrow \infty} K_{U_n, V_n}(\bfx,\bfx) = K_{\Hh_2, \Hh_1}(\bfx,\bfx).
\end{equation}
\end{lemma}

Lemma \ref{lem-11} means that the geometric bulk conductance can be computed locally. This observation was key to the work \cite{DZ23} on emergence of edge spectrum for truncated topological insulators.

\begin{proof} 1. Recall that $K_{A,B} = P\big[[P, \1_A],[P, \1_B]\big]$ \eqref{eq-9f} satisfies the kernel estimate \eqref{eq-2k}: 
\begin{equation}
    \label{eq2-4e}
|K_{A,B}(\bfx, \bfy)| \leq C_N e^{-N(d_l(\bfx, \p A) + d_l(\bfx, \p A) + d_l(\bfy, \p B) + d_l(\bfy, \p B) + d_l(\bfx, \bfy))}.
\end{equation}
Since $K_{A, B}$ is linear in $\1_A$ and $\1_B$, we have   
\[
K_{A,B} - K_{A\cap \Dd_n(0), B\cap \Dd_n(0)} = K_{A\cap \Dd_n(0)^c, B} + K_{A\cap \Dd_n(0)^c, B\cap \Dd_n(0)^c}.
\]
Note that
\[
d_l(\bfx, \p (A\cap \Dd_n(0)^c)) \geq d_l(0,\p (A\cap \Dd_n(0)^c))) - d_l(\bfx, 0)\geq \ln (1 + n) - d_l(\bfx, 0).
\]
Hence we have  
\begin{equation} 
\big|K_{A\cap \Dd_n(0)^c, B}(\bfx, \bfx) \big|\leq C_N e^{-Nd_l(\bfx, \p (A\cap \Dd_n(0)^c))}\leq C_N e^{-N \ln (1 + n) + Nd_l(\bfx, 0)}.
\end{equation}
The same estimate holds for $\big|K_{A\cap \Dd_n(0)^c, B\cap \Dd_n(0)^c}(\bfx, \bfx)\big|$. Hence we obtain
\begin{equation}\label{eq2-5h}
\big|K_{A,B} - K_{A\cap \Dd_n(0), B\cap \Dd_n(0)}(\bfx, \bfx)\big|\leq C_N e^{-N\ln (1 + n)+ Nd_l(\bfx, 0) } = C_N(1 + n)^{-N} e^{-N d_l(\bfx, 0)}.
\end{equation}

2. We now apply \eqref{eq2-5h} to the pair $(U_n,V_n)$, then to the pair $(\Hh_2,\Hh_1)$:
\begin{equations}
    \big| K_{U_n,V_n}(\bfx,\bfx) - K_{U_n \cap \Dd_n(0),V_n \cap \Dd_n(0)}(\bfx,\bfx)\big|\leq     C_N(1 + n)^{-N} e^{Nd_l(\bfx, 0)},
\end{equations}
\begin{equations}
\big| K_{\Hh_2, \Hh_1}(\bfx,\bfx) - K_{\Hh_2 \cap \Dd_n(0),\Hh_1 \cap \Dd_n(0)}(\bfx,\bfx)     \big| \leq C_N(1 + n)^{-N} e^{Nd_l(\bfx, 0) }.
\end{equations}
Because $\Hh_2 \cap \Dd_n(0) = U_n \cap \Dd_n(0)$ and $\Hh_1 \cap \Dd_n(0)=V_n \cap \Dd_n(0)$, summing these two bounds gives
\begin{equations}
    \big| K_{U_n ,V_n}(\bfx,\bfx) - K_{\Hh_2, \Hh_1}(\bfx,\bfx)     \big| \leq  2C_N(1 + n)^{-N} e^{Nd_l(\bfx, 0) }.
\end{equations}
It suffices to take the limit as $n$ goes to $\infty$ to conclude.
\end{proof}

\begin{prop}\label{lem-10} Let $(U,V)$ be transverse simple sets in $\R^2$ such that $\II_{U, V} = 1$. There exists a family $\FF = \{ (U_n,V_n): \ n \in \N\}$ with the following properties:
\begin{itemize}
    \item[(a)] $\FF$ is uniformly transverse;
    \item[(b)] $\FF$ is equivalent to $(U,V)$;
    \item[(c)] In the disk $\Dd_n(0)$, $U_n$ is the upper half-plane $\Hh_2 = \{x_2 >0\}$ and $V_n$ is the right half-plane $\Hh_1 = \{x_1 >0\}$.
\end{itemize}    
\end{prop}

Proposition \ref{lem-10} is the key construction in the proof of Theorem \ref{thm:1} and we defer its proof to \S\ref{sec-5.4} - \S\ref{sec-5.6}. 

\begin{proof}[Proof of Theorem \ref{thm:1} assuming Proposition \ref{lem-10}] By Proposition \ref{lem-10}, there is a family  $\FF = \{ (U_n,V_n): \ n \in \N\}$ satisfying (a), (b), (c) above. Since $\FF$ is equivalent to $(U,V)$ and $U_n \Delta U$, $V_n \Delta V$ are bounded, by Proposition \ref{prop-2c}, we deduce that $\sigma_b^{U, V}(P) = \sigma_b^{U_n, V_n}(P)$ for all $n$. In particular,
\begin{equation}\label{eq2-8a}
    \sigma_b^{U, V}(P) = \lim_{n \rightarrow \infty} \sigma_b^{U_n, V_n}(P) = \lim_{n \rightarrow \infty} \sum_{\bfx \in \Z^2} K_{U_n,V_n}(\bfx,\bfx).
\end{equation}
Our plan is now to apply the dominated convergence theorem to the above series.

Define $k_n(\bfx) = K_{U_n,V_n}(\bfx)$. By Lemma \ref{lem-11}, for every $\bfx \in \Z^2$, $k_n(\bfx)$ converges to $K_{\Hh_2,\Hh_1}(\bfx,\bfx)$ as $n \rightarrow \infty$. Moreover, by \eqref{eq2-4e} and uniformly-admissiblility, when $|\bfx|\geq c^{-1}$,
\begin{equation}
    \big| k_n (\bfx) \big| = \big| K_{U_n,V_n} (\bfx,\bfx)\big| \leq C_N  \Psi_{U_n, V_n}(\bfx)^{-2N} \leq C_N |\bfx|^{-2Nc}=:k(\bfx).
\end{equation}
When $N$ is large enough, we have $k(\bfx)\in \ell^1(\Z^2)$. Thus we can apply dominated convergence theorem to get:

\begin{align}
    \lim_{n \rightarrow \infty} \sum_{\bfx \in \Z^2} K_{U_n,V_n}(\bfx,\bfx) & =
    \lim_{n \rightarrow \infty} \sum_{\bfx \in \Z^2} k_n (\bfx) 
    \\
    & =  \sum_{\bfx \in \Z^2} \lim_{n \rightarrow \infty} k_n (\bfx) = \sum_{\bfx \in \Z^2} K_{\Hh_2,\Hh_1}(\bfx,\bfx) = \sigma_b^{\Hh_2,\Hh_1}(P) = \sigma_b(P).
\end{align}
Going back to \eqref{eq2-8a} completes the proof of Theorem \ref{thm:1}.
\end{proof}

\subsection{Ordering entrance / exit points} \label{sec-5.4} In the following two subsections, we aim to prove Proposition \ref{lem-10}. Let $U, V$ be transverse simple sets. For $r > 0$, $J = U,V$, define:
\begin{equations}\label{eq2-4n}
    t_J^+(r) \de \sup \big\{ t : \ \gamma_J(t) \in \overline{\Dd_r(0)} \big\}, \quad
    t_J^-(r) \de \inf \big\{ t : \ \gamma_J(t) \in \overline{\Dd_r(0)} \big\}, \quad
    \bfz_J^\pm(r) \de \gamma_J \circ t_J^\pm(r).
\end{equations}

\begin{lemma}\label{lem:17} Let $U, V$ be transverse simple sets such that $\II_{U, V} = 1$. Let $R > 0$ such that $\p U \cap \p V \subset \Dd_R(0)$. For all $r > R$, there exists $\te_1 < \te_2 < \te_3 < \te_4 < \te_1 +2\pi$ (see Figure \ref{fig:P9}), such that
\begin{equation}
    \bfz_V^-(r) = r e^{i\te_1}, \quad \bfz_U^-(r) = r e^{i\te_2}, \quad \bfz_V^+(r) = r e^{i\te_3}, \quad \bfz_U^+(r) = r e^{i\te_4}.
\end{equation}
\end{lemma}
\begin{figure}[b] 
\floatbox[{\capbeside\thisfloatsetup{capbesideposition={right,center},capbesidewidth=0.6\textwidth}}]{figure}[\FBwidth]
    {\caption{The sets $U,V$ with the arguments $\te_1, \te_2, \te_3, \te_4$.}} 
    {\includegraphics[width=0.85\linewidth]{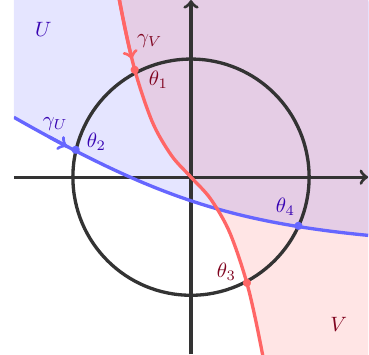} \label{fig:P9}}
\end{figure}
\begin{proof} 1. Fix $r > R$. Define
\begin{equation}
    \Gamma_U^+ = \gamma_U\big(( t_U^+(r),+\infty)\big), \qquad \Gamma_U^- = \gamma_U\big(( -\infty,t_U^-(r))\big)
\end{equation}
These are connected sets which do not intersect $\p V$. In particular they lies in $V$ or in $V^\CCC$. But since $\II_+(U,V) = 1$, for $t$ sufficiently large $\gamma_U(t) \in V$. It follows that $\Gamma_U^+ \subset V$. Likewise $\Gamma_U^- \subset V^\CCC$.

2. Let $\te_1$ such that $\bfz_V^-(r) = re^{i\te_1}$; let $\te_3 \in (\te_1,\te_1+2\pi)$ such that $\bfz_V^+(r) = r e^{i\te_3}$. 
Let $\gamma_W$ be the curve defined by:
\begin{equation}
    \gamma_W = \gamma_V \big|_{t \leq  t_V^-(r)} \oplus re^{-i\te}\big|_{\te \in (-\te_1,2\pi-\te_3)} \oplus \gamma_V \big|_{t \geq t_V^+(r)}.
\end{equation}
where we use the symbol $\oplus$ to denote the concatenation of two curves. See Figure \ref{fig:P10}.
Let $W$ the connected component of $\C \setminus \gamma_W(\R)$ to the left of $\gamma_W$. The function $\te \mapsto e^{-i\te}$ runs clockwise. Therefore (e.g. by considering a sufficiently small disk centered at $re^{i\te}$, $\te \in (\te_1,\te_2)$, simply split by $\gamma_W$) we see that $\Dd_r(0)$ lies in the component of $\C \setminus \Gamma_W$ to the right of $\gamma_W$, in particular it does not intersect $W$.

3. Note that $\Gamma_U^+$ is an unbounded connected subset of $V$. Because $V \Delta W$ is bounded (see Proposition \ref{prop:4}) $\Gamma_U^+$ intersects $W$. Moreover, it does not intersect $\p W$, so we have $\Gamma_U^+ \subset W$. In particular, 
\begin{equation}
   \bfz_U^+(r) = \lim_{t \rightarrow t_U^+(r)} \gamma_U(t) \in \overline{W}. 
\end{equation}
Because $|\bfz_U^+(r)| = r$, we obtain 
\begin{equation}\label{eq2-6f}
    \bfz_U^+(r) \in \overline{W} \cap \p \Dd_r(0) = \big\{ re^{-i\te}: \ \te \in (-\te_1,2\pi-\te_3) \big\} 
\end{equation}
Therefore, $\bfz_U^+(r) = r e^{\te_4}$ for some $\te_4 \in (\te_3,\te_1+2\pi)$.

\begin{figure}[b] 
\floatbox[{\capbeside\thisfloatsetup{capbesideposition={right,center},capbesidewidth=0.6\textwidth}}]{figure}[\FBwidth]
    {\caption{Construction of $\gamma_W$ from $\gamma_V$ and determination of $\theta_4\in (\theta_3, \theta_1 + 2\pi)$.}} 
    {\includegraphics[width=0.85\linewidth]{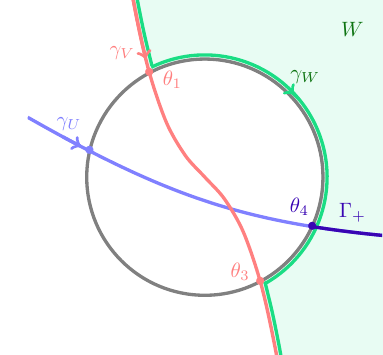} \label{fig:P10}}
\end{figure}

4. By a similar argument, $\Gamma_U^-$ lies in the component of $\C \setminus \Gamma_W$ to the right of $\gamma_W$. Because $\Gamma_U^-$ does not intersect $\Dd_r(0)$ and $|\bfz_U^-(r)| = r$, we deduce that as in \eqref{eq2-6f} that
\begin{equation}
    \bfz_U^-(r) \in \big\{ re^{i\te}: \ \te \in (\te_1,\te_3) \big\},
\end{equation}
which implies $\bfz_U^-(r) = r e^{\te_2}$ for some $\te_3 \in (\te_1,\te_3)$. This completes the proof.  
\end{proof}

\subsection{Construction of uniformly transverse family of sets}\label{sec-5.6}

Let $(U,V)$ be a $c$-transverse pair with $\II_{U, V} = 1$. Now we can construct a family of uniformly transverse sets $(U_n, V_n)$ that is equivalent to $(U, V)$ such that $U_n \cap \Dd_n(0) = \Hh_1 \cap \Dd_n(0)$, $V_n \cap \Dd_n(0) = \Hh_2 \cap \Dd_n(0)$. 

By Lemma \ref{lem:17}, when $n$ is large enough, for some $\te_1 < \te_2 < \te_3 < \te_4 <\te_1 +2\pi$, we have 
\begin{equation}
    \bfz_V^-(4n) = 4n e^{i\te_1}, \quad \bfz_U^-(4n) = 4n e^{i\te_2}, \quad \bfz_V^+(4n) = 4n  e^{i\te_3}, \quad \bfz_U^+(4n) = 4n  e^{i\te_4}.
\end{equation}
For such an $n$, we define four functions $\alpha_k : [0,4] \rightarrow \R$, $k \in \{1,2,3,4\}$ by
\begin{equation}
    \alpha_k(s) = \systeme{ \frac{k \pi}{2}, & s \in [0,2] \\ 
    (3-s) \frac{k \pi}{2} + (s-2) \te_k, & s \in [2,3]\\
    \te_k, & s \in [3,4]}
\end{equation}
and four curves $\bfz_k : [0,4] \rightarrow \C$ (see Figure \ref{fig:P11}) by
\begin{equation}\label{eq2-5r}
    \bfz_k(s) = ns \cdot e^{i\alpha_k(s)}.
\end{equation}
\begin{figure}[b]
  \centering
  \includegraphics[width=1\textwidth]{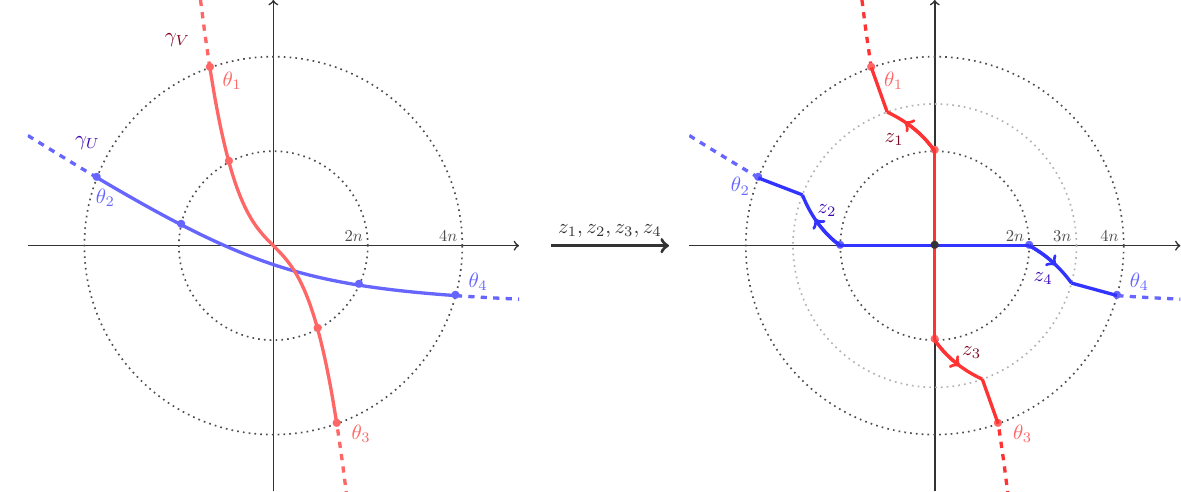}
  \caption{The curves $\bfz_1, \bfz_2, \bfz_3, \bfz_4$.}
  \label{fig:P11}
\end{figure}

We will make use of the following result:

\begin{lemma}\label{lemma-14} Let $c \in [0,1], n \in \N$. Assume that $t,s \in [0,4]$ are such that $|\bfz_1(t)| - |\bfz_2(s)| <2^{-7}n^c$. Then 
         \begin{equation}\label{eq-6p}
        2n \left|\sin \left(\frac{\alpha_1(t) - \alpha_2(s)}{2}\right)\right|\geq 2^{-4} n^{ c}.
        \end{equation}
\end{lemma}

\begin{proof}
         We first estimate $\alpha_1(4) - \alpha_2(4) = \theta_1 - \theta_2$. Without loss of generality, we can assume $|\theta_1 - \theta_2|< \pi$ by choosing $\theta_2\in (\theta_1 - \pi, \theta_1 + \pi)$. Recall that $\bfz_1(4) = \bfz_V^-(4n)\in \p V$ by definition. By transversality, we see
         \[
         \begin{split}
             (4n)^c &\leq \PUV\big(\bfz_1(4)\big) = d\big(\bfz_1(4), \p U\big)\leq \big|\bfz_1(4) - \bfz_2(4)\big|\\
             &= 4n\big|e^{i\theta_1} - e^{i\theta_2}\big| \leq 4n|\theta_1 -\theta_2|.
         \end{split}
         \]
         Since $c<1$, for any $n\geq 1$, $(4n)^{c - 1}\leq 1\leq \frac{\pi}{2}$. Recall that by definition, $\alpha_1(t) - \alpha_2(t)$ is monotone over $[0,4]$. Hence we obtain
         \begin{equation}
             \label{eq-6n}
         \begin{split}
             |\alpha_1(t) - \alpha_2(t)| &\geq \min \big\{|\alpha_1(4) - \alpha_2(4)|, |\alpha_1(0) - \alpha_2(0)|\big\}\geq \min \big\{|\theta_1 - \theta_2|, \frac{\pi}{2}\big\}\\
             &\geq \min \left\{(4n)^{c-1}, \frac{\pi}{2}\right\}\geq (4n)^{c - 1} \geq \frac{1}{4}n^{c - 1}.
         \end{split}
         \end{equation}
         Meanwhile, by the definition of $\alpha_k(s)$ and the assumption $|\bfz_1(t)| - |\bfz_2(s)| = nt - ns \leq 2^{-7} n^{c}$, we have 
         \begin{equation}
             \label{eq-6o}
         |\alpha_2(t) - \alpha_2(s)|\leq |t - s|\cdot\max\limits_{r\in[0,4]}|\alpha_2'(r)|\leq |t - s|\cdot  4\pi\leq 2^{-3} n^{c - 1}.
         \end{equation}
         Combining \eqref{eq-6n} and \eqref{eq-6o}, we obtain 
    \[
    \begin{split}
        |\alpha_1(t) - \alpha_2(s)|&\geq |\alpha_1(t) - \alpha_2(t)| - |\alpha_2(t) - \alpha_2(s)|\\ 
   &\geq \frac{1}{4}n^{c -1} -\frac{1}{8}n^{c - 1} = \frac{1}{8}n^{c - 1}. 
   \end{split}
     \]
     Since $\left|\frac{\alpha_1(t) - \alpha_2(s)}{2}\right| = \left|\frac{\theta_1 - \theta_2}{2}\right| \leq \frac{\pi}{2}$ and when $| \alpha |\leq \frac{\pi}{2}$, $|\sin  \alpha| \geq |r|/2$, we see that 
     \[
     2n\left|\sin\left(\frac{\alpha_1(t) - \alpha_2(s)}{2}\right)\right|\geq \frac{2n|\alpha_1(t) - \alpha_2(s)|}{4} \geq  2^{-4} n^c
     \]
     for any $n\geq 1$. This proves Lemma \ref{lemma-14}.      \end{proof}

We now deform $U, V$ by using $\bfz_k$ as boundary functions. Specifically, we define two sets $U_n, V_n$ as the set lying to the left of the following boundaries (see Figure \ref{fig:P12}) -- recall that $\alpha_k$ and $\bfz_k$ depends on $n$, see \eqref{eq2-5r}:
\begin{equations}\label{eq2-7y}
    \p U_n = \gamma_U\big((-\infty,t_U^-(4n))\big) \cup \bfz_2\big([0,4]\big) \cup \bfz_4\big([0,4]\big) \cup \gamma_U\big((t_U^+(4n), +\infty)\big),
    \\
    \p V_n = \gamma_V\big((-\infty,t_V^-(4n))\big) \cup \bfz_1([0,4]) \cup \bfz_3([0,4]) \cup \gamma_V\big((t_V^+(4n),+\infty)\big).
\end{equations}
\begin{figure}[b]
  \centering
  \includegraphics[width=0.75\textwidth]{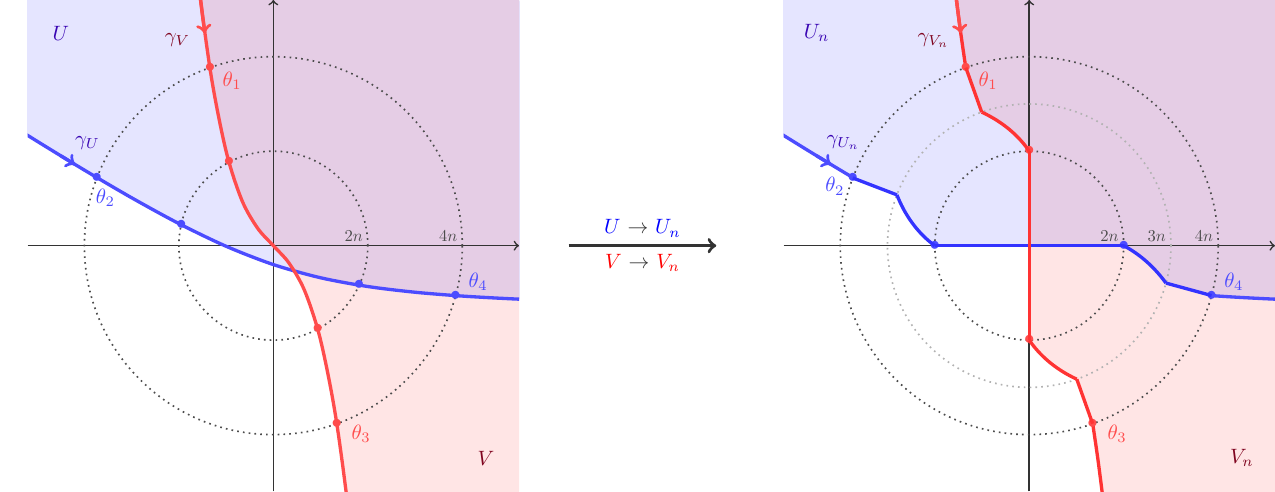}
  \caption{Deforming $U, V$ to $U_n, V_n$. Blue (red) region represents $U$ ($V$) and $U_n$ ($V_n$).}
  \label{fig:P12}
\end{figure}

\begin{lemma}\label{prop-1}   Assume $(U,V)$ is $c$-transverse, $c \in (0,1)$. There exists some $c' \in (0,1)$ such that $\big\{(U_n, V_n)\big\}_{n\geq 1}$ is $c'$-uniformly transverse.
\end{lemma}

\begin{proof}
     Given $A,B \subset \R^2$ and $\bfx\in \R^2$, let $\bfx_A$, $\bfx_B$ denote points on $\p A$, $\p B$ that minimizing the distances to $\p A, \p B$: 
    \[
    |\bfx -  \bfx_A| = \min_{\bfy\in \p A}|\bfx- \bfy|, \qquad |\bfx- \bfx_B| = \min_{\bfy\in \p B}|\bfx- \bfy|.
    \]
    Without loss of generalities, we may assume $|\bfx_{U_n}|\geq |\bfx_{V_n}|$. To prove the lemma, it suffices to show that $\Psi_{U_n, V_n}(\bfx) \geq |\bfx|^{\frac{c}{4}}$ when $|\bfx|$ is large enough, uniformly in $n$. 
    
   \noindent \textbf{Case 1: $|\bfx_{V_n}|\geq 4n$.} Then $|\bfx_{U_n}|\geq|\bfx_{V_n}|\geq 4n $; thus $\bfx_{U_n} = \bfx_U$, $\bfx_{V_n} = \bfx_V$ and 
    \[
    \Psi_{U_n, V_n}(\bfx) =  \PUV(\bfx) \geq |\bfx|^c
    \]
    when $|\bfx|\geq 1/c$, which is uniform in $n$. \\
    
   \noindent \textbf{Case 2: $|\bfx_{V_n}|<4n$ and $|\bfx|\geq 5n$.} Then $|\bfx| - 4n \geq |\bfx|/5$; thus 
    \[
     \Psi_{U_n,V_n}(\bfx)\geq  |\bfx -  \bfx_{V_n}| \geq  |\bfx| - |\bfx_{V_n}| \geq  |\bfx| - 4 n \geq  \frac{|\bfx|}{5} \geq |\bfx|^c,
    \]
    when $|\bfx|\geq 5$, which is uniform in $n$. \\
    
   \noindent  \textbf{Case 3: $|\bfx_{V_n}|<4n$, $|\bfx|<5n$, and $|\bfx_{U_n}| - |\bfx_{V_n}|\geq 2^{-7} n^c$.} Then 
    \[
    \begin{split}       
      \Psi_{U_n,V_n}(\bfx) &=  |\bfx - \bfx_{U_n}| + |\bfx - \bfx_{V_n}| \geq  |\bfx_{U_n} - \bfx_{V_n}| \geq  |\bfx_{U_n}| - |\bfx_{V_n}|\\
     &\geq  2^{-7} n^c \geq 2^{-7}\left(  \frac{|\bfx|}{5}\right)^c \geq |\bfx|^{\frac{c}{2} },
    \end{split}
    \]
    when $|\bfx|$ is large enough, uniformly in $n$. \\
       
    \noindent \textbf{Case 4: $|\bfx_{V_n}|<4n$, $|\bfx|<5n$, and  $|\bfx_{U_n}| - |\bfx_{V_n}|< 2^{-7} n^c$.} Since $c<1$ and $n\geq 1$, $|\bfx_{U_n}| - |\bfx_{V_n}|< n$ and we have $|\bfx_{U_n}|\in[0,5n)$. We divide it into three cases.
    
    Assume first that $|\bfx_{U_n}|\in[0,2n)$. Then we have $|\bfx_{V_n}|\leq |\bfx_{U_n}|< 2n$. Since in $\Dd_{4n}(0)$, the boundaries $\p U_n$ and $\p V_n$ coincides with the $x$ and $y$-axis from $-4n$ to $4n$, if we write $\bfx = (\bfx_1, \bfx_2)$, we have
    \[
     \Psi_{U_n, V_n}(\bfx) =  |\bfx - \bfx_{U_n}| + |\bfx - \bfx_{V_n}| =  |\bfx_1| + |\bfx_2| \geq   |\bfx|\geq |\bfx|^c
    \]
    when $|\bfx|\geq 1$, which is uniform in $n$.  
    
    Assume now that $|\bfx_{U_n}|\in 
(4n,5n)$. Then $|\bfx_{V_n}|\in [3n,4n]$. Since $\bfx_{V_n}\in \p V_n\cap \Dd_{4n}(0)$, by definition, $\bfx_{V_n}\in [3ne^{i\theta_k}, 4ne^{i\theta_k}]$ for some $k \in \{1,2,3,4\}$. Set $\tilde \bfx_V:= 4ne^{i\theta_k}\in \p V$. Then 
\[
|\bfx_{V_n} -  \tilde \bfx_V| = \min_{|\bfy|\geq 4n} |\bfx_{V_n} - \bfy| \leq |\bfx_{V_n} - \bfx_{U_n}| \leq \Psi_{U_n, V_n}(\bfx). 
\]
As a result,
\[
\begin{split}
    2\Psi_{U_n, V_n}(\bfx)&\geq |\bfx - \bfx_{V_n}| + |\bfx- \bfx_{U_n}| + | \bfx_{V_n}- \tilde \bfx_V| \geq |\bfx- \tilde \bfx_V| + |\bfx- \bfx_{U_n}|\\
    &\geq |\bfx- \bfx_V| + |\bfx- \bfx_U| = \PUV(\bfx)\geq |\bfx|^c,
\end{split}
\]
when $|\bfx|\geq 1/c$, which is uniform in $n$. 
    
    Finally, assume that $|\bfx_{U_n}|\in [2n,4n]$. Then we have $|\bfx_{V_n}|\in [n,3n]$. So both $\bfx_{U_n}$ and $\bfx_{V_n}$ sit on some $\bfz_j$ curves; we write without loss of generalities $\bfx_{U_n} = \bfz_1(t) = nte^{i\alpha_1(t)}$, $\bfx_{V_n} = \bfz_2(s) = nse^{i\alpha_2(s)}$. Then we have 
    \begin{align}\label{eq-9b}
     \Psi_{U_n, V_n}(\bfx)^2 & \geq |\bfx_{U_n} - \bfx_{V_n}|^2 
     \geq 4n^2\sin^2\left(\frac{\alpha_1(t) - \alpha_2(s)}{2}\right)
     \geq 2^{ - 4} n^c\geq 2^{ - 4}\left(\frac{|\bfx|}{5}\right)^{c}\geq |\bfx|^{\frac{c}{2}},
    \end{align}
    when $|\bfx|$ is large enough, uniformly in $n$. In \eqref{eq-9b}
    we first applied the inequality
    \begin{equation}
        |r_1 e^{i\te_1}-r_2e^{i\te_2}|^2 = r_1^2 + r_2^2 - 2 r_1 r_2 \cos(\te_1-\te_2) \geq 2r_1 r_2 \big( 1- \cos(\te_1-\te_2) \big) = 4 r_1 r_2 \sin^2\left(\dfrac{\te_1-\te_2}{2}\right).
    \end{equation} 
    together with $|\bfx_{U_n}|, |\bfx_{V_n}| \geq n$; then we applied  \eqref{eq-6p}.\\

\noindent \textbf{Summary.} Cases 1-4 imply that for all $n\geq 1$, $\Psi_{U_n,V_n}(\bfx)\geq |\bfx|^{\frac{c}{2}}$  when $|\bfx|$ is large enough, uniformly in $n$. Hence there is some $c'<c$ such that $(U_n, V_n)$ is $c'$-uniformly transverse. \end{proof}

\begin{proof}[Proof of Proposition \ref{lem-10}] Let $(U,V)$ be transverse sets and $(U_n,V_n)$ defined by \eqref{eq2-7y}. Introduce
\begin{equation}
    \FF = \big\{ (U_n, V_n) : n \geq 1 \big\}.
\end{equation}
By Lemma \ref{prop-1}, $\FF$ is uniformly transverse; this proves (a). Moreover, because $U,U_n$ respectively lie to the left of $\gamma_U, \gamma_{U_n}$, and $\gamma_{U_n}$ is a compact perturbation of $\gamma_U$, $U \Delta U_n$ is bounded -- see Proposition \ref{prop:4}. Likewise $V \Delta V_n$ is bounded. It follows that $\FF$ is equivalent to $(U,V)$; this proves (b). Finally, $\gamma_{U_n}$ simply splits $\Dd_n(0)$, and $\Hh_1 \cap \Dd_n(0)$ is the left of $\gamma_{U_n}$ in $\Dd_n(0)$ (see Definition \ref{def:3}). It follows from Proposition \ref{prop:1} that 
\begin{equation}
    \Hh_1 \cap \Dd_n(0) = U_n \cap \Dd_n(0).
\end{equation}
Likewise, $\Hh_2 \cap \Dd_n(0) = V_n \cap \Dd_n(0)$. This proves (c). 
\end{proof}

This completes the proof of Theorem \ref{thm:1}.

\section{Extension}\label{sec-7}

Here we extend Theorem \ref{thm:1} to general transverse sets $(U,V)$, when neither $U$ nor $V$ are assumed to be simple, but instead are \textit{good:}

\begin{definition}\label{def:4} An open subset $A$ of $\R^2$ is good if:
\begin{itemize}
    \item[(a)] The set $A^\CCC := \Int A^c$ has boundary $\p A$;
    \item[(b)] The connected components $\{ \Gamma_k : k \in \N \}$ of $\p A$ are the ranges of simple paths or loops and satisfy $\inf\limits_{j \neq k} d(\Gamma_j, \Gamma_k) > 0$.
    \item[(c)] $\p A$ does not intersect $\Z^2$. 
\end{itemize}
\end{definition}

It turns out that one can naturally extend the intersection number between transverse simple sets $\II_{U, V}$ (\S\ref{sec-5.2}) to transverse good sets -- see \S\ref{sec-7.2}. In rough terms, $\II_{U, V}$ counts the signed number of times that $\p U$ (oriented so that $U$ lies to its left) enters $V$. We then have the following extension of Theorem \ref{thm:1}: 

\begin{theorem}\label{thm:2} Assume that Assumption \ref{ass-all} holds and that $U,V$ are transverse good sets. Then
\begin{equation}\label{eq2-7f}
    \sigma_b^{U, V}(P_\pm) = \II_{U, V} \cdot \sigma_b(P_\pm).
\end{equation}
\end{theorem}

Theorem \ref{thm:2} generalizes Theorem \ref{thm:1} to transverse good sets. Our proof of Theorem \ref{thm:2}, given in \S\ref{sec:5.5}, consists of decomposing $U$ and $V$ in simple sets (see \S\ref{sec:5.1}) and applying additivity properties of Chern and intersection numbers (see \S\ref{sec-7.2}-\ref{sec:5.4}).

For our puprose, good sets are not restrictive. The next result implies that for any $\AAA \subset \R^2$, there exists a good set $A$ such that $\AAA \cap \Z^2 = A \cap \Z^2$. In particular, $\1_A$ and $\1_\AAA$, seen as multiplicative operators on $\ell^2(\Z^2)$, are equal.

\begin{lemma}\label{lem-9} Given $\Aa \subset \Z^2$, define
\begin{equation}\label{eq2-7a}
    A = \big\{ \bfx \in \R^2 : d_1(\bfx,\Z^2 \setminus \Aa) > 3/4 \big\},
\end{equation}
where $d_\infty$ denotes the $|\cdot|_\infty$ distance. Then $A$ is a good set and $A \cap \Z^2 = \Aa$. 
\end{lemma}
The set $A$ is obtained from $\Aa$ by essentially filling the squares with vertices in $A$, and lightly enlarging the filled set -- see Figure \ref{fig:2}.

\begin{figure}[t] 
\floatbox[{\capbeside\thisfloatsetup{capbesideposition={right,center},capbesidewidth=0.6\textwidth}}]{figure}[\FBwidth]
    {\caption{The sets $\Aa$ (blue dots) and $A$ (blue filling) with the boundary $\p A$ (blue polygonal curves). }} 
    {\includegraphics[width=0.9\linewidth]{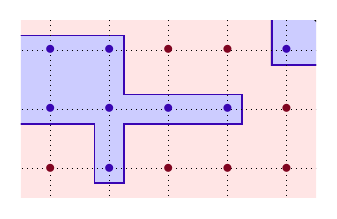} \label{fig:2} } 
\end{figure}

We prove this result in Appendix \ref{app:C}.

\subsection{Properties of good sets.}\label{sec:5.1}

We will need the following properties of good sets:

\begin{lemma}\label{lem-13} Let $A$ be a good set. Then $A^\CCC$ is good.  \end{lemma}

\begin{lemma}\label{lem:20} Let $A$ be a good set, $\{ \Gamma_k : k \in \N\}$ be the connected components of $\p A$ and $\{ A_\ell: \ell \in \N\}$ be the connected components of $A$. There exists a partition $\KK_0, \KK_1, \dots$ of $\N$ such that
\begin{equation}\label{eq2-7k}
    \p A_\ell = \bigsqcup\limits_{k \in \KK_\ell} \Gamma_k.
\end{equation}
In particular, the connected components of good sets are good. 
\end{lemma}

\begin{lemma}\label{lem:18b} Let $A$ be a connected good set, $\{ \Gamma_k : k \in \N\}$ be the connected components of $\p A$. Let $\AAA = A^\CCC$. There exists a labelling $\{ \AAA_k : k \in \N\}$ of the connected components of $\AAA$ such that for all $k \in \N$, $\p \AAA_k = \Gamma_k$. In particular, the sets $\AAA_k$ are simple. 
\end{lemma}

\begin{figure}[b] 
\floatbox[{\capbeside\thisfloatsetup{capbesideposition={right,center},capbesidewidth=0.6\textwidth}}]{figure}[\FBwidth]
    {\caption{The boundary of $A$ is composed of $\Gamma_1$, $\Gamma_2$, $\Gamma_3$, which are respectively boundary of the disjoint connected components $\mathcal A_1$, $\mathcal A_2$ and $\mathcal A_3$ of $A^\CCC$.}} 
    {\includegraphics[width=0.85\linewidth]{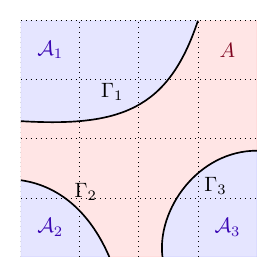} \label{fig: P_Ac}}
\end{figure}

We prove these results in Appendix \ref{app:D}; see Figure \ref{fig: P_Ac} for a pictorial representation.

\subsection{Extending the intersection numbers to pairs of good sets}\label{sec-7.2}

In \S\ref{sec-5} we defined the intersection number $\II_{U, V}$  between transverse simple sets $U, V$. We now extend it to transverse good sets $U, V$.

\begin{lemma}\label{lem:24} Let $U,V$ be transverse sets such that $U$ is simple and $V$ is good. Let $V_\ell$ be the connected components of $V$. Then $U,V_\ell$ are transverse sets, $\II_{U, V_\ell} \neq 0$ for all but at most two values of $\ell$ and
\begin{equation}\label{eq2-7m}
    \II_{U, V} = \sum_{\ell \in \N} \II_{U,V_\ell}.
\end{equation}
\end{lemma}

\begin{proof} Fix $\ell \in \N$. By Lemma \ref{lem:20}, we have $\p V_\ell \subset \bigsqcup\limits_{m \in \N} \p V_m = \p V$. Therefore, $\Psi_{U,V_\ell} \geq \Psi_{U,V}$, so $U,V_\ell$ are transverse sets. 

Let $\gamma$ be a simple path or loop with range $\p U$. If $\gamma$ is a simple loop then $\II_{U, V} = 0$ and $\II_{U, V_\ell} = 0$ for all $\ell$; in particular \eqref{eq2-7m} holds. We now assume that $\gamma$ is a simple path.

Let $c > 0$ such that $(U,V)$ are $c$-transverse sets, see \eqref{eq2-0b}. Since $\gamma$ is proper,
\begin{equation}
    \exists T > 0, \ \forall t \geq T, \ |\gamma(t)| > c^{-1}.
\end{equation}
As a result, by \eqref{eq-0b}, for every $ t \geq T$, $\ln \Psi_{U,V} ( \gamma(t))\geq c\ln |\gamma(t)| > 0$; thus $\big\{\gamma(t):t\geq T\big\} \cap \p V = \emptyset$. 

Denote $V_{-1}:= V$ for convenience. For each $\ell \geq -1$, $\big\{\gamma(t):t\geq T\big\}\cap \p V_\ell = \emptyset$; thus for each $\ell\geq -1$, 
\begin{equation}\label{eq2-5d}
 \text{either} \quad  \big\{ \gamma(t):t\geq T\big\} \subset  V_\ell, \quad \text{or} \quad \big\{ \gamma(t):t\geq T\big\} \subset  V_\ell^\CCC.
\end{equation}
Because the sets $\{ V_\ell : \ell \in \N\} $ are disjoint, there is at most one value of $\ell \in \N$ such that the first statement is true -- in particular $\II_+(U,V_\ell) \neq 0$ for at most one value of $\ell \in \N$. On the other hand, \eqref{eq2-5d} implies $\1_{V_\ell}\circ \gamma(t)$ is constant when $t\geq T$ for any $\ell\geq -1$. As a result, 
\begin{align}
    \II_+(U,V) = \lim_{t \rightarrow +\infty} \1_V \circ \gamma(t) = \1_V \circ \gamma(T) = \sum_{\ell \in \N} \1_{V_\ell} \circ \gamma(T) = \sum_{\ell \in \N} \lim_{t \rightarrow +\infty} \1_{V_\ell} \circ \gamma(t) = \sum_{\ell \in \N} \II_+(U,V_\ell).
\end{align}
A similar identity holds for $\II_-(U,V)$, and this completes the proof. 
\end{proof}

\begin{lemma}\label{lem:23} Let $U,V$ be transverse sets with $U$ simple, $V$ good and connected. Let $\VV = V^\CCC$ and $\VV_k$ be the connected components of $\VV$. Then $U,\VV_k$ are transverse simple sets, $\II_{U,\VV_k} \neq 0$ for at most two values of $k$ and 
\begin{equation}\label{eq2-7r}
    \II_{U, V} = -\sum_{k \in \N} \II_{U,\VV_k}.
\end{equation}
\end{lemma}

\begin{proof} 1. We first note that 
\begin{equation}\label{eq2-7n}
    \II_{U, V} = \lim_{t \rightarrow +\infty} \1_V \circ \gamma(t) - \lim_{t \rightarrow -\infty} \1_V \circ \gamma(t) = \lim_{t \rightarrow -\infty} \1_{V^c} \circ \gamma(t) - \lim_{t \rightarrow +\infty} \1_{V^c} \circ \gamma(t).
\end{equation}
As in the proof of Lemma \ref{lem:24}, because $\gamma$ is proper and $U,V$ are transverse sets, there exists $T > 0$ such that if $|t| > T$, $\gamma(t) \notin \p V = \p V^c$. Therefore, when $|t| \geq T$, we have
\begin{equation}\label{eq2-7o}
    \1_{V^c} \circ \gamma(t) = \1_{V^\CCC} \circ \gamma(t) = \1_\VV \circ \gamma(t)
\end{equation}
It follows from \eqref{eq2-7n} and \eqref{eq2-7o} that $\II_{U, V} = -\II_{U, \VV}$.

2. The connected components of $\VV$ are the sets $\VV_k$; these are simple sets by Lemma \ref{lem:18b}. Therefore, by Lemma \ref{lem:24}, at most two values of $k$ are such that $\II_{U, \VV_k} \neq 0$ and
\begin{equation}\label{eq2-7q}
     \sum_{k \in \N} \II_{U, \VV_k} = \II_{U, \VV} = -\II_{U, V}.
\end{equation}
This completes the proof.     
\end{proof}

\subsection{Extending the intersection number} Here we extend $\II_{U, V}$ to all transverse good sets.

We shall make use of the following fact. Let $X$ be a bounded subset of $\R^2$. Then
\begin{equation}\label{eq2-7g}
    \inf_{\bfx \neq \bfy \in X} |\bfx-\bfy| > 0 \quad \Rightarrow \quad \# X < \infty.
\end{equation}
Indeed, if $X$ was infinite, then it would contain a sequence with pairwise distinct values, convergent without loss of generalities because $X$ is bounded, and hence Cauchy. But that's not possible because of the condition in \eqref{eq2-7g}.

\begin{lemma}\label{lem:21} Let $U, V$ be transverse good sets with $U$ connected. Let $\UU = U^\CCC$ and $\{ \UU_k : k \in \KK\}$ be the connected components of $\UU$. Then:
\begin{itemize}
    \item For each $k \in \KK$,  $\UU_k$ is simple and $\UU_k,V$ are transverse sets. 
    \item $\II_{\UU_k, V}$ is zero for all but finitely many $k$.
\end{itemize}
\end{lemma}

\begin{proof} By Lemma \ref{lem:20}, $\p \UU_k \subset \p U$, therefore $\Psi_{\UU_k,V} \geq \Psi_{U,V}$. It follows that $\UU_k,V$ are transverse sets. Define
\begin{equation}
    S = \big\{ k \in \N : \ \II_{\UU_k, V} \neq 0\big\}.
\end{equation}
We must show that $S$ is finite. Let $\{ \Gamma_k : k \in \N\}$ be the connected components of $\p U$, labelled so that $\Gamma_k = \p \UU_k$ -- see Lemma \ref{lem:18b}. If $k \in S$, $\Gamma_k  = \p \UU_k$ and $\p V$ must intersect; let $\bfx_k$ in the intersection. Let $c \in (0,1)$ such that $(U,V)$ are $c$-transverse sets. We note that $\bfx_k \in \Gamma_k \subset \p U$, and because we also have $\bfx_k \in \p V$, $\Psi_{U,V}(\bfx_k) = 0$ so $\ln|\bfx_k| \leq c^{-1}$.

In particular, the set $\{ \bfx_k: k \in S\}$ is finite. Moreover, $\bfx_k \in \Gamma_k$, so -- because $U$ is good -- the pairwise distances between the points $\bfx_k$ are bounded below. By \eqref{eq2-7g}, $S$ is finite. This completes the proof. \end{proof}

Lemma \ref{lem:21} allows us to define $\II_{U, V}$ for transverse good sets $U,V$ such that $U$ is connected. We let $\UU_k$ the connected components of $U^\CCC$ and we set:
\begin{equation}\label{eq2-7h}
    \II_{U, V} \de -\sum_{k \in \N} \II_{\UU_k, V}.
\end{equation}
Lemma \ref{lem:21} ensures that $\II_{\UU_k, V}$ is well-defined and that the sum is finite. Because after adequate labelling, $\p \UU_k = \Gamma_k$ (where the sets $\Gamma_k$ are the connected components of $\p U$), we can interpret \eqref{eq2-7h} as the signed number of entrances in $V$ of trajectories following $\p U$, oriented so that $U$ lies to their left.

\begin{lemma}\label{lem:22} Let $U, V$ be transverse good sets. Let $\{ U_\ell : \ell \in \N \}$ be the connected components of $U$. Then:
\begin{itemize}
    \item For each $\ell$, $U_\ell,V$ are transverse sets. 
    \item $\II_{U_\ell, V}$ is zero for all but finitely many $\ell$.
\end{itemize}
\end{lemma}

\begin{proof} Because of Lemma \ref{lem:20}, $U_\ell$ is good; and $\p U_\ell \subset \p U$. This implies $\Psi_{U_\ell V} \geq \Psi_{U,V}$ so $U_\ell,V$ are transverse sets. As in the proof of Lemma \ref{lem:21}, we define
\begin{equation}
    S = \big\{ \ell \in \N: \  \II_{U_\ell, V} \neq 0\big\}.
\end{equation}
Assume $\ell \in S$. Let $\UU_\ell = U_\ell^\CCC$ and $\{ \UU_{\ell k} :  k \in \N\}$ be the connected components of $\UU_\ell$. By \eqref{eq2-7h}, $\II_{\UU_{\ell k},V} \neq 0$ for some $k \in \N$ (depending on $\ell$), and there exists $\bfx_\ell \in \p \UU_{\ell k} \cap \p V \subset \p U_\ell \cap \p V \subset \p U \cap \p V$. As in the proof of Lemma \ref{lem:21}, we have $\Psi_{U,V}(\bfx_\ell) = 0$ and therefore $\{ \bfx_\ell : \ell \in S\}$ is bounded.

Moreover, because of Lemma \ref{lem:20}, the distances between the sets $\p U_\ell$ are bounded below,  hence so are the pairwise distances between the points $\bfx_\ell$. By \eqref{eq2-7g}, $S$ is finite. This completes the proof. 
\end{proof}

Lemma \ref{lem:22} allows us to remove the condition that $U$ is connected in \eqref{eq2-7h}: if $U,V$ are transverse sets, we split $U$ in its connected components $U_\ell$ and we define:
\begin{equation}\label{eq2-7i}
    \II_{U, V} \de \sum_{\ell \in \N} \II_{U_\ell, V}.
\end{equation}
The quantities $\II_{U_\ell, V}$ are well-defined because $U_\ell,V$ are transverse good sets (see Lemma \ref{lem:20} and Lemma \ref{lem:22}). The sum is finite because of Lemma \ref{lem:22}. Because of \eqref{eq2-7k}, we interpret again $\II_{U, V}$ as the signed number of entrances in $V$ of trajectories following $\p U$, oriented so that $U$ lies to their left.

\begin{figure}[b]
     \centering
     \begin{subfigure}[ht]{0.3\textwidth}
         \centering
         \includegraphics[width=\textwidth]{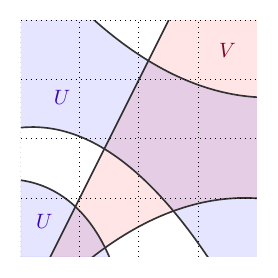}
         \caption{}
         \label{fig: P_multi_1}
     \end{subfigure}
     \hfill
     \begin{subfigure}[ht]{0.3\textwidth}
         \centering
         \includegraphics[width=\textwidth]{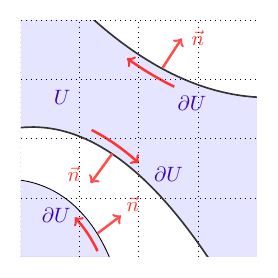}
         \caption{}
         \label{fig: P_multi_2}
     \end{subfigure}
     \hfill
     \begin{subfigure}[ht]{0.3\textwidth}
         \centering
         \includegraphics[width=\textwidth]{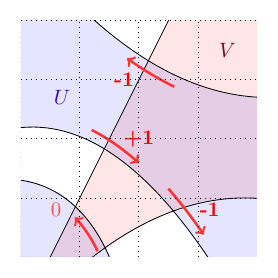}
         \caption{}
         \label{fig: P_multi_3}
     \end{subfigure}
        \caption{We define the intersection number $\mathcal X_{U, V}$ between transverse sets $U, V$ in two steps. We first orient $\p U$ such that $U$ is to its left according to the outward-pointing normal, see (a) and (b); then we count how many times the oriented $\p U$ enters $V$, see (c). In particular, here $\mathcal X_{U, V} = +1 - 1 - 1 + 0= -1$. }
        \label{fig: P_multi}
\end{figure}

\subsection{Chern numbers and connected components}\label{sec:5.4}

\begin{lemma}\label{lem:19} Let $U, V$ be transverse good sets. Let $\{ V_k : k \in \N\}$ denote the connected components of $V$. Then for all $k \in \N$, $U, V_k$ are transverse sets; the series $\sigma_b^{U,V_k}(P)$ is summable; and 
\begin{equation}\label{eq2-7d}
    \sigma_b^{U, V}(P) = \sum_{k \in \N} \sigma_b^{U,V_k}(P).
\end{equation}
\end{lemma}
Since $\sigma_b^{U, V}(P)$ is antisymetric in $(U,V)$, \eqref{eq2-7d} also holds when $U$ and $V$ are switched. See Figure \ref{fig: P_multi}.

\begin{proof} 1. Because $V$ is good, for every $k$ we have $\p V_k \subset \p V$, and therefore $\ln \Psi_{U,V_k}(\bfx) \geq \ln \Psi_{U,V}(\bfx)$: $U,V_k$ are transverse sets. Define now
\begin{equation}
    W_K = \bigsqcup\limits_{k > K} V_k.
\end{equation}
We claim that $\p W_K = \bigsqcup\limits_{k > K} \p V_k$. Indeed, for any $k > K$, $V_k$ is a connected component of $W_K$ and therefore $\p V_k \subset \p W_K$ -- see e.g. \cite[Theorem IV.3.1]{N51}. This proves $\bigsqcup\limits_{k > K} \p V_k \subset \p W_k$. 

To prove the reverse inclusion, fix $\bfx$ in $\p W_K$ and define
\begin{equation}
    \rho_0 = \inf_{j \neq k} d_1(\Gamma_j,\Gamma_k) > 0,
\end{equation}
where $\{ \Gamma_k : k \in \N\}$ are the connected components of $\p V$. Let $\Dd_n = \Dd_{1/n}(\bfx)$ be a neighborhood of $\bfx$ with $n \geq 2\rho_0^{-1}$. Because $\bfx \in \p W_K$, $\Dd_n$ intersects both $W_K = \bigsqcup V_k$ and $W_K^c = \bigcap V_k^c$, so (at least) one of the sets $\{ V_k : k > K\}$, denoted $V_{k_n}$; and all of the sets $\{ V_k^c : k >K\}$. In particular, $\Dd_n$ contains a point of $\p V_{k_n}$. Because $\Dd_n$ has diameter less than $\rho_0$, we deduce that $k_n$ does not depend on $n$, so there exists $k_*$ such that $\Dd_n \cap \p V_{k_*} \neq \emptyset$. Taking $n \rightarrow \infty$ proves that $\bfx_n \in \p V_{k_*}$. This implies that $\p W_K \subset \bigsqcup\limits_{k > K} \p V_k$, completing the proof of $\p W_K = \bigsqcup\limits_{k > K} \p V_k$.

2. Hence $\p W_K \subset \p V$ and $U,W_K$ are transverse sets. By expanding $K_{U,V}$ according to the finite disjoint union $V = \bigsqcup\limits_{k \leq K} V_k \bigsqcup W_K$, we deduce that
\begin{equation}
    \sigma_b^{U, V}(P) = \sum_{k=0}^K \sigma_b^{U,V_k}(P)  + \sigma_b^{U,W_K}(P).
\end{equation}
It then remains to prove that
\begin{equation}\label{eq2-7c}
    \lim_{K \rightarrow \infty} \sigma_b^{U,W_K}(P) = 0.
\end{equation}

3. Since $\p W_K \subset \p V$, we have $ \Psi_{U,W_K}(\bfx) \geq \Psi_{U,V}(\bfx)$. Therefore by \eqref{eq-2q} and \eqref{eq-0x}, we have
\begin{align}
    \big| \sigma_b^{U,W_K}(P) \big| \leq C_N \left( \sum_{\bfx \in \Z^2}  \Psi_{U,W_K}(\bfx)^{-N}\right)^2
    & \leq C_N \cdot \sup_{\bfx \in \Z^2} e^{-N\ln \Psi_{U,W_K}(\bfx)} \cdot \left(\sum_{\bfx \in \Z^2}  \Psi_{U,V}(\bfx)^{-\frac{N}{2}} \right)^2
    \\
    & \leq C_N  \cdot e^{- \frac{N}{2} d_l(\p U, \p W_K)} \cdot \left( \sum_{\bfx \in \Z^2} \Psi_{U,V}(\bfx)^{-\frac{N}{2}}\right)^2.
\end{align}
The last sum is finite by \eqref{eq-0v}, and from $\p W_k = \bigsqcup\limits_{k > K} V_k$, we have
\begin{equation}
    d_l(\p U, \p W_K) = \inf_{k > K} d_l(\p U, \p V_k).
\end{equation}
So, to prove \eqref{eq2-7c} it suffices to show that 
\begin{equation}
\lim_{K \rightarrow \infty} \inf_{k > K} d_l(\p U, \p V_k) = \liminf_{k \rightarrow \infty}  d_l(\p U, \p V_k) = +\infty,
\end{equation}
or equivalently that $d_l(\p U, \p V_k)$ goes to infinity as $k$ goes to infinity.

4. Let $\bfx_k \in \p V_k$ with $d_l(\p U, \bfx_k) = d_l(\p U, \p V_k)$. By Lemma \ref{lem:20}, we have:
\begin{equation}
    d_l(\bfx_k, \bfx_j) \geq \ln \big(1 + d_1(\bfx_k,\bfx_j)\big) \geq \ln\big(1 + d_1(\p V_k, \p V_j)\big) \geq \ln(1 + \inf_{m \neq \ell} d_1(\Gamma_m,\Gamma_\ell))  = \ln(1 +  \rho_0)>0.
\end{equation}
Hence the set $\{ \bfx_k : k \in \N\}$ is made of points that are at least $\rho_0$ pairwise distant; by \eqref{eq2-7g}, $\bfx_k$ has no bounded subsequence:
\begin{equation}
    \limsup_{k \rightarrow \infty} |\bfx_k| = \infty.
\end{equation}
When $|\bfx_k| \geq c^{-1}$,
\begin{equation}
    d_l(\p U, \p V_k) = d_l(\p U, \bfx_k) = \ln \Psi_{U,V_k}(\bfx_k) \geq \ln \Psi_{U,V}(\bfx_k) \geq c\ln |\bfx_k|.
\end{equation}
Therefore $d_l(\p U, \p V_k)$ goes to infinity as $k \rightarrow \infty$, as claimed. This completes the proof.     
\end{proof}

\subsection{Proof of Theorem \ref{thm:2}}\label{sec:5.5}

\begin{proof}[Proof of Theorem \ref{thm:2}] Let $U,V$ be transverse sets. We claim that \eqref{eq2-7f} is valid if one of the following scenarios holds (ordered according to generality):
\begin{itemize}
    \item[A.] $U$ and $V$ are simple.
    \item[B.] $U$ is simple and $V$ is good and connected.
    \item[C.] $U$ is simple and $V$ is good.
    \item[D.] $U$ is good and connected and $V$ is good.
    \item[E.] And finally, $U$ and $V$ are good.
\end{itemize}
Of course A is Theorem \ref{thm:1}, which we already proved; and E is Theorem \ref{thm:2}.

B. Assume that $U,V$ are transverse sets with $U$ simple and $V$ good and connected; let $\VV = V^\CCC$. Then $U, \VV$ are also transverse sets: indeed $\VV$ is good (Lemma \ref{lem-13}) and $\p \VV = \p V$. Moreover, because $\p V \cap \Z^2 = \emptyset$ (see the Definition \ref{def:4} of good sets), $V^c \cap \Z^2 = V^\CCC \cap \Z^2 = \VV $ a
\begin{equation}\label{eq2-8f}
    \sigma_b^{U, V}(P) = -\sigma_b^{U, V^c}(P) = -\sigma_b^{U,\VV}(P),
\end{equation}
Let $\{ \VV_k : k \in \N \}$ be the connected components of $\VV$. By Lemma \ref{lem:19}, the sets $U, \VV_k$ are transverse, $\sigma_b^{U, \VV_k}(P)$ is summable, and:
\begin{equation}
   \sigma_b^{U, V}(P) = - \sigma_b^{U,\VV}(P) = -\sum_{k \in \N} \sigma_b^{U, \VV_k}(P).
\end{equation}
By Lemma \ref{lem:18b}, $\VV_k$ is simple. Hence, $U, \VV_k$ are transverse simple sets, and by part A,
\begin{equation}\label{eq2-7e}
    \sigma_b^{U, \VV_k}(P) = \II_{U, \VV_k} \cdot \sigma_b(P).
\end{equation}
Summing both sides of \eqref{eq2-7e} over $k$ and applying Lemma \ref{lem:23}, we obtain:
\begin{equation}
   \sigma_b^{U, V}(P)  = -\sum_{k \in \N} \sigma_b^{U, \VV_k}(P) =  -\sum_{k \in \N} \II_{U, \VV_k} \cdot \sigma_b(P) = \II_{U, V} \cdot \sigma_b(P).
\end{equation}

C. Assume now that $U$ is simple and $V$ is good. Let $V_\ell$ be the connected components of $V$. By Lemma \ref{lem:19}, the sets $U,V_\ell$ are transverse, $\sigma_b^{U,V_\ell}(P)$ is summable, and 
\begin{equation}
    \sigma_b^{U, V}(P) = \sum_{\ell \in \N} \sigma_b^{U,V_\ell}(P).
\end{equation}
Now, $U$ is simple and $V_\ell$ is connected, so by part B, we have $\sigma_b^{U,V_\ell}(P) = \II_{U, V_\ell} \cdot \sigma_b(P)$. Therefore, by Lemma \ref{lem:24}:
\begin{equation}
    \sigma_b^{U, V}(P) = \sum_{\ell \in \N} \II_{U, V_\ell} \cdot \sigma_b(P) = \II_{U, V} \cdot \sigma_b(P).
\end{equation}

D. Assume now that $U$ is good and connected and $V$ is good; set $\UU = U^\CCC$. As for \eqref{eq2-8f}, $(\UU,V)$ is good and we have $\sigma_b^{U, V}(P) = -\sigma_b^{\UU,V}(P)$. Let $\{ \UU_m : m \in \N\}$ be the connected components of $\UU$. By Lemma \ref{lem:19}, the sets $\UU_m,V$ are transverse, the terms $\sigma_b^{\UU_m,V}(P)$ are summable, and:
\begin{equation}\label{eq2-7s}
   \sigma_b^{U, V}(P) = - \sigma_b^{\UU,V}(P) = -\sum_{m \in \N} \sigma_b^{\UU_m,V}(P).
\end{equation}
Now because $\UU,V$ are transverse, so are $\UU_m,V$, but now $\UU_m$ is simple and $V$ is good -- see Lemma \ref{lem:24}. By Step C, $\sigma_b^{\UU_m,V}(P) = \II_{\UU_m, V} \cdot \sigma_b(P)$. Using the definition \eqref{eq2-7h} of $\II_{U, V}$ when $U,V$ are transverse sets and $U$ is connected, we obtain once again:
\begin{equation}
    \sigma_b^{U, V}(P)  = -\sum_{m \in \N} \sigma_b^{\UU_m,V}(P) = -\sum_{m \in \N} \II_{\UU_m, V} \cdot \sigma_b(P) =  \II_{U, V} \cdot \sigma_b(P).
\end{equation}

E. Assume finally that $U$ and $V$ are both good, and let $\{ U_n : n \in \N\}$ be the connected components of $U$. We apply (in this order) Lemma \ref{lem:19}, part D, and the definition \eqref{eq2-7i} of $\II_{U, V}$ for good sets:
\begin{equation}
    \sigma_b^{U, V}(P) = \sum_{n \in \N} \sigma_b^{U_n,V}(P) =  \sum_{n \in \N} \II_{U_n, V} \cdot \sigma_b(P) = \II_{U, V} \cdot \sigma_b(P).
\end{equation}
This completes the proof. 
\end{proof}

\subsection{Proof of Theorem \ref{thm-main}}
Brought together, Theorems \ref{thm-4a} and \ref{thm:2} and Lemma \ref{lem-9} complete the proof of Theorem \ref{thm-main}. Indeed, in the setup of Assumption \ref{ass-all}, we have (by Theorem \ref{thm:2}):
\begin{equation}\label{eq-9j}
    \sigma_e(H_e) = \sigma_b^{U,V}(P_+) - \sigma_b^{U,V}(P_-).
\end{equation}
Let now $\UU, \VV \subset \R^2$ be given by 
\begin{equation}
    \UU = \{ \bfx \in \R^2 : d(\bfx,\Z^2 \setminus U) > 3/4\}, \qquad  VV = \{ \bfx \in \R^2 : d(\bfx,\Z^2 \setminus V) > 3/4\}.
\end{equation}
By Lemma \ref{lem-9}, $\UU$, $\VV$ are good transverse sets, and $\UU \cap \Z^2 = U$, $\VV \cap \Z^2 = V$. Therefore, by Theorem \ref{thm-4a}:
\begin{equation}\label{eq-9i}
    \sigma_b^{U,V}(P_\pm) = \sigma_b^{\UU,\VV}(P_\pm) = \chi_{\UU,\VV} \cdot \sigma_b(P_\pm).
\end{equation}
Theorem \ref{thm-main} follows from plugging \eqref{eq-9i} in \eqref{eq-9j}.

\appendix
\section{Proofs for Section \ref{sec-2}}\label{app-ESR}
\begin{proof}[Proof of Lemma \ref{lemma-1a}]
    (1) Following \cites{EGS05,AW15}, we introduce
    \[
    S_\alpha := \sup_{\bfx\in \Z^2} \sum_{\bfy\in \Z^2} |H(\bfx, \bfy)|(e^{\alpha d_1(\bfx, \bfy)} - 1)<+\infty, \qquad \alpha<2\nu
    \]
and note that, for $\alpha\in (0,\nu]$, we have \cite{DZ23}*{(3.3)}: 
\begin{equation}
    \label{eq-1m}
    S_\alpha / \alpha \leq S_\nu/\nu \leq 16/\nu^4<+\infty.
\end{equation}
Recall the Combes-Thomas estimate in \cite[Theorem 10.5]{AW15}: for any $\alpha<\nu$ and $z\in \C$ such that $\Delta_z = d(z,\Sigma(H))>S_\alpha$, we have 
        \begin{equation}\label{eq-2s}
        |(H - z)^{-1}(\bfx, \bfy)|\leq \frac{1}{\Delta_z - S_\alpha}e^{-\alpha d_1(\bfx, \bfy)}.
        \end{equation}
    
    The estimate \eqref{eq-9a} follows from \eqref{eq-2s} by taking $\alpha = 2^{-5}\nu^4|\Im z|$. In this case, since $\nu<1$ and $0<|\Im z|< 1$, we see that $\alpha<\nu$. By \eqref{eq-1m}, 
    \[
    S_\alpha<16\alpha/\nu^3<|\Im z|/2<|\Im z|\leq \Delta_z.
    \]
    
    (2) We recall the Helffer-Sj\"ostrand formula \cite{Z22}*{Theorem 14.8}: for any $g\in C_c^\infty(\R)$, 
\begin{equation}
    \label{eq: HS_1}
    g(H) = \frac{1}{2\pi i }\int_\C \frac{\p \tilde g}{\p \bar{z}} (H - z)^{-1} dz\wedge d\bar{z},
\end{equation}
where $\tilde g = \tilde g(z, \bar{z}):\C \to \C$ is an \textit{almost analytic extension} of $g$:
\begin{equation}
    \label{eq: almost_analytic}
\tilde g \big|_\R = g, \qquad \p_{\bar{z}}\tilde g = O(|\Im z|^\infty) \ \ \text{as $\Im z \rightarrow 0$}, \qquad \supp(\tilde g)\subset \{z:|\Im z| \leq 1\}. 
\end{equation}
Above, $\p_{\bar{z}}\tilde g = O(|\Im z|^\infty)$ means for any $N>0$, there is $C_N > 0$ (depending on $g$), such that $|\p_{\bar{z}}\tilde g|\leq C_N |\Im z|^N.$ We now have, for $\bfx \neq \bfy$, by \eqref{eq-9a},
\[
\begin{split}
    |g(H)(\bfx, \bfy)| &\leq \frac{1}{2\pi}\int_{\C} |\p_{\bar{z}} \tilde g(z)| \left| \left(H - z\right)^{-1}(\bfx, \bfy)\right| |dz \wedge d\bar{z}| \\
    &\leq C_{N}\int_{\supp(g)\times[-1,1]} |\Im z|^N \frac{2}{|\Im z|}e^{-\frac{\nu^4|\Im z|}{32}d_1(\bfx, \bfy)} |dz \wedge d\bar{z}|\\
    &\leq C_N d_1(\bfx,\bfy)^{-N} \cdot \int_0^{1}  w^{N - 1}e^{-\frac{\nu^4}{32} w} dw  \\
    &\leq C_N d_1(\bfx, \bfy)^{-N}, \label{eq-9c}
\end{split}
\]
where we used the substitution $w = |\Im z| d_1(\bfx,\bfy)$. When $\bfx = \bfy$, since $g$ is bounded, $|g(H)(\bfx, \bfx)|\leq \Vert g \Vert_\infty$. Combining this bound with the estimate \eqref{eq-9c}, we obtain $g(H)$ is PSR:
\[
\big|g(H)(\bfx, \bfy)\big|\leq C_N\big(1 + d_1(\bfx, \bfy)\big)^{-N} = C_N e^{-Nd_l(\bfx, \bfy)}.
\]
\end{proof}

\begin{proof}[Proof of Lemma \ref{lemma-1b}]
Note that $\1_U S \1_U(\bfx, \bfy) = S(\bfx, \bfy) \1_U(\bfx) \1_U(\bfy)$ and $\1_U(\bfx)\leq e^{-cd_*(\bfx, U)}$ for any $c>0$. Hence we obtain \eqref{eq-2p}
\[
    \big|\1_U S_1 \1_U(\bfx, \bfy)\big|\leq \big|S_1(\bfx, \bfy)\big|\1_U(\bfx) \1_U(\bfy)\leq C e^{-cd_*(\bfx, \bfy) - cd_*(\bfx, U) - cd_*(\bfy, U)},
\]
and
\[
\begin{split}
    \big|\1_U S_1 \1_{U^c}(\bfx, \bfy)\big|& \leq \big|S_1(\bfx, \bfy)\big|\1_U(\bfx) \1_{U^c}(\bfy) \leq Ce^{-cd_*(\bfx, \bfy) -  c d_*(\bfx, U) - c d_*(\bfy, U^c)}.
\end{split}
\]
Note that 
\[
\begin{split}
    &d_*(\bfx, \bfy) + d_*(\bfx, U) + d_*(\bfy, U^c)\geq d_*(\bfy, U) + d_*(\bfy, U^c) \geq d_*(\bfy, \p U)\\
    &d_*(\bfx, \bfy) + d_*(\bfx, U) + d_*(\bfy, U^c)\geq d_*(\bfx, U) + d_*(\bfx, U^c) \geq d_*(\bfx, \p U).
\end{split}
\]
Therefore we obtain \eqref{eq-2b}:
\[
\big|\1_U S_1 \1_{U^c}(\bfx, \bfy)\big| \leq Ce^{-cd_*(\bfx, \bfy) - \frac{2c}{3}d_*(\bfx, U) - \frac{2c}{3}d_*(\bfy, U^c)}\leq C e^{-\frac{c}{3}d_*(\bfx, \bfy) - \frac{c}{3}d_*(\bfx, \p U) - \frac{c}{3}d_*(\bfy,\p U)}.
\]
Finally \eqref{eq-2c} follows from $[\1_ U , S_1] = \1_ U  S_1 \1_{ U ^c} - \1_{ U ^c} S_1 \1_ U $.
\end{proof}

    \begin{proof}[Proof of Lemma \ref{lemma-1c}]
    We note that for $C = \prod\limits_{j = 1}^n C_j$, we have
        \[
    \begin{split}
        &\left|\left(\prod_{j = 1}^n S_{j}\right)(\bfx, \bfy)\right| \leq C \sum\limits_{\bfz_1, \cdots, \bfz_{n - 1}} e^{-cd_*(\bfx, \bfz_1) - \cdots -cd_*(\bfz_{n -1}, \bfy)} e^{-cd_*(\bfz_{p - 1}, \p U ) - cd_*(\bfz_p , \p U ) - cd_*(\bfz_{q - 1}, \p V )  - cd_*(\bfz_{q}, \p V )}\\
        &\leq Ce^{-\frac{c}{4}d_*(\bfx, \p U ) - \frac{c}{4}d_*(\bfx, \p V ) - \frac{c}{4}d_*(\bfy, \p U ) - \frac{c}{4}d_*(\bfy, \p V ) - \frac{c}{4}d_*(\bfx, \bfy)} \sum\limits_{\bfz_1,\cdots, \bfz_{n - 1}} e^{-\frac{c}{4}d_*(\bfx, \bfz_1) -\cdots -\frac{c}{4}d_*(\bfz_{n -2}, \bfz_{n - 1})}\\
        &\leq CM_{d_*, \frac{c}{4}}^{n - 1} e^{-\frac{c}{4}d_*(\bfx, \bfy) -\frac{c}{4}d_*(\bfx, \p U ) - \frac{c}{4}d_*(\bfx, \p V ) - \frac{c}{4}d_*(\bfy, \p U ) - \frac{c}{4}d_*(\bfy, \p V )}. 
    \end{split}
    \]
     To get the second line, we used the triangle inequality (the quantities $d_*(\bfx, \p V )$, $d_*(\bfy, \p V )$, and $d_*(\bfx, \bfy)$ being obtained similarly.):
    \[
    d_*(\bfx, \bfz_1) + d_*(\bfz_1, \bfz_2) + \cdots + d_*(\bfz_{p - 2}, \bfz_{p - 1}) + d_*(\bfz_{p - 1}, \p U ) \geq d_*(\bfx, \p U ).
    \]
    This completes the proof. 
\end{proof}

\begin{proof}[Proof of Corollary \ref{cor-1a}]
    To show $\prod\limits_{j = 1}^n S_j$ is trace-class, it is enough to show 
    \[
    \sum\limits_{\bfx, \bfy}\left|\left(\prod\limits_{j = 1}^n S_{j}\right)(\bfx, \bfy)\right|<+\infty.
    \]
    Using that $S_j$ are PSR, \eqref{eq-2d} in Lemma \ref{lemma-1c} and \eqref{eq-0v}, we have 
    \begin{equation}
        \label{eq-1r}
    \begin{split}
    \sum\limits_{\bfx, \bfy}\left|\left(\prod_{j = 1}^n S_{j}\right)(\bfx, \bfy)\right|&\leq CM_{d_l, \frac{N}{4}}^{n - 1}\sum\limits_\bfx e^{-\frac{N}{4}d_l(\bfx, \p  U ) - \frac{N}{4} d_l(\bfx, \p  V ) }\left(\sum\limits_\bfy e^{-\frac{N}{4}d_l(\bfx, \bfy)}\right)\\
    &\leq CM_{d_l, \frac{N}{4}}^{n } \sum\limits_\bfx e^{- \frac{N}{4}d_l(\bfx, \p  U ) - \frac{N}{4}d_l(\bfx, \p  V ) }\\
    &\leq CM_{d_l, \frac{N}{4}}^n \sum\limits_\bfx \Psi_{U,V}(\bfx)^{-\frac{N}{4}}<+\infty
    \end{split}
    \end{equation}
    Thus $\prod\limits_{j = 1}^n S_j$ is trace-class.    This completes the proof.
\end{proof}
\section{Proof of Proposition \ref{prop:4}}\label{app:A}

\subsection{Sets separated by proper injective curves}

\begin{lemma}\label{lem:4} Let $\gamma : \R \rightarrow \R^2$ be a continuous, proper and injective map with range $\Gamma$. Then $\C \setminus \Gamma$ has precisely two connected components.

Moreover, if $A$ is one of these connected components:
\begin{itemize}
    \item[(i)] $A$ is unbounded;
    \item[(ii)] $A$ has boundary $\Gamma$;
    \item[(iii)] For any compact set $K \subset \C$, there exists $\Omega \subset \C \setminus K$ open such that $A \cap \Omega$ is connected and $\C \setminus \Omega$ is compact.
\end{itemize}
\end{lemma}

\begin{proof} In this proof we identify $\R^2$ with $\C$. We denote by $\Dd_r(\bfz)$ the open disk centered at $\bfz \in \C$, of radius $r$. We will use the Riemann sphere $S^2 = \C \cup \{\infty\}$. This is a topological space with topology generated by the sets
\begin{equation}
    \DD_r(\bfz_*) = \systeme{ \Dd_r(\bfz_*) & \text{ if $\bfz_* \in \C$, $r > 0$}; \\ S^2 \setminus \overline{\Dd_r(0)} & \text{ if $\bfz_* = \infty$, $r > 0$.}}
\end{equation}
See for instance \cite[\S6.2]{T11}.

After reparametrizing $\gamma$, we may assume that it maps $S^1 \setminus \{1\}$ to $\C$. 
Extend $\gamma$ to $\tgamma : S^1 \rightarrow S^2$ by requiring $\tgamma(1) = \infty$. We claim that $\tgamma$ is injective. Indeed, if $\tgamma(t) = \tgamma(t')$ for some $t, t'$, then either the joint value is $\infty$, in that case $t=t'=1$; or it is in $\C$, so then $\gamma(t) = \gamma(t')$ and $t=t'$ since $\gamma$ is injective.  We claim that $\tgamma$ is continuous.  Indeed, if $\bfz_* \in \C$ and $r>0$ then $\DD_r(\bfz_*) \subset \C$, so
\begin{equation}
   \tgamma^{-1}\big(\DD_r(\bfz_*)\big) = \tgamma^{-1}\big(\Dd_r(\bfz_*)\big) = \gamma^{-1}\big(\Dd_r(\bfz_*)\big)
\end{equation}
is open in $S^1 \setminus \{1\}$, hence in $S^1$. Likewise, if $\bfz_* =\infty$ and $r>0$ then 
\begin{equation}
\tgamma^{-1}\big(\DD_r(\infty)\big) = \tgamma^{-1}\big(S^2 \setminus \overline{\Dd_r(0)}\big) = S^1 \setminus \gamma^{-1}\big(\overline{\Dd_r(0)}\big).
\end{equation}
Now, $\gamma^{-1}(\overline{\Dd_r(0)})$ is a compact subset of $S^1 \setminus \{1\}$, hence of $S^1$ (because a covering by open subsets in $S^1 \setminus \{1\}$ is a covering by open subsets in $S^1$). In particular it is closed in $S^1$ and $\tgamma^{-1}\big(\DD_r(\infty)\big)$ is open in $S^1$. We proved that preimages under $\tgamma$ of topological basis elements are open in $S^1$, so $\tgamma : S^1 \rightarrow S^2$ is continuous.

This means that $\tgamma$ is a Jordan curve in $S^2$. By \cite[Theorem 63.4]{M13}, $S^2 \setminus \tGamma$ (with its subspace topology) has precisely two connected component, each of them having $\tGamma$ as boundary. Let $A$ be one of them. Since $\infty \in \tGamma$, $A$ is open in $S^2 \setminus \tGamma$, which is open in $\C$; so $A$ is open in $\C$. It follows that $\C \setminus \Gamma$ has precisely two connected components, which are the same as $S^2 \setminus \tGamma$.

Since $\infty \in \tGamma$, $\infty$ is in the $S^2$-closure of $A$. Thus for any $r > 0$, $\DD_r(\infty) \cap A \neq \emptyset$, so $A$ is unbounded. This proves (i).  The closure of $A$ in $S^2$ is $\overline{A} \cup \{\infty\}$, so the boundary of $A$ in $S^2$ is equal to $\p A \cup \{\infty\} = \tGamma = \Gamma \cup \{\infty\}$ (where $\overline{A}$ and $\p A$ denote the closure and boundary of $A$ in $\C$, respectively). We deduce that the boundary of $A$ in $\C$ is $\Gamma$. This proves (ii).

Finally, let $w \in \C \setminus \overline{A}$ and $F_1 : S^2 \rightarrow S^2$ defined by $F_1(z) = 1+(z-w)^{-1}$ for $z \neq \infty, w$, $F_1(w) = \infty$, and $F_1(\infty) = 1$. This is a homeomorphism of $S^2$ --  see e.g. \cite[Theorem 6.3.1]{T11}. Since $A$ is open and $w \notin \overline{A}$, the set $F_1(\tGamma)$ is a closed Jordan curve in $\C$, and the bounded component of $S^2 \setminus F_1(\tGamma)$ is $F_1(A)$. By Schoenflies's theorem \cite[Theorem 3.1]{T92}, 
there exists a homeomorphism $F_2 : \C \rightarrow \C$ such that $F_2 \circ F_1(A) = \Dd_1(0)$ and $F_2(1) = 1$; extend $F_2$ to a homeomorphism $S^2 \rightarrow S^2$ by $F_2(\infty) = \infty$. Set $F = F_2 \circ F_1$, which is a homeomorphism of $S^2$ that sends $A$ to $\Dd_1(0)$.

Fix $K \subset \C$ compact, and set $V = \C \setminus K \cup \{\infty\} = S^2 \setminus K$, which is a neighborhood of $\infty$ in $S^2$. Then $F(V)$ is a neighborhood of $F(\infty) = 1$, so there exists $r > 0$ such that $\Dd_r(1) \subset F(V)$. The set $\Dd_r(1) \cap \Dd_1(0)$ is connected, therefore so is 
\begin{equation}
    F^{-1}\big( \Dd_r(1) \cap \Dd_1(0) \big) = F^{-1}\big( \Dd_r(1) \big) \cap W_1 = \tOmega \cap A = \Omega \cap A, \quad \tOmega \de F^{-1}\big( \Dd_r(1) \big), \quad \Omega \de \tOmega \setminus \{\infty\}.
\end{equation}
Note that $\tOmega$ is a neighborhood of $\infty$, so the set $\C \setminus \Omega$ is compact; and $\Omega \subset V \setminus \{\infty\} \subset \C \setminus K$. This proves (iii).
\end{proof}

\begin{figure}[t]
  \centering
  \includegraphics[width=1\textwidth]{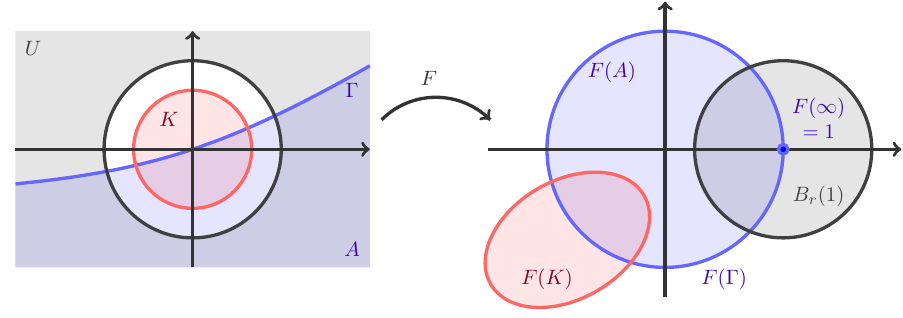}
  \caption{Pictorial representation of the proof of Lemma \ref{lem:4}(iii). After applying a homeomorphism of $S^2 = \C \cup \{\infty\}$ that sends $\infty$ to $1$, we can think of $F(A)$ as the unit disk and $F(K)$ as a compact set that does not contain $1$. Thus we can take $F(\Omega)$ as a disk centered at $1$ that does not intersect $F(K)$.}
  \label{fig:1}
\end{figure}

\begin{lemma}\label{lem:5} Let $\Gamma_1, \Gamma_2$ be the ranges of two proper injective continuous maps such that $\Gamma_1 \Delta \Gamma_2$ is bounded. Let $A_1$ be a connected component of $\R^2 \setminus \Gamma_1$. Then there exists  a unique connected component $A_2$  of $\R^2 \setminus \Gamma_2$ such that $A_1 \Delta A_2$ is bounded.
\end{lemma}

\begin{proof} Let $A_j^+, A_j^-$ denote the two connected components of $\R^2 \setminus \Gamma_j$, labelled so that $A_1^+ = A_1$; and $K = \Gamma_1 \Delta \Gamma_2$, which is compact. If both $A_1^+ \Delta A_2^-$ and $A_1^+ \Delta A_2^+$ were bounded, then we would have $\R^2 \setminus \Gamma_2 = A_2^- \sqcup A_2^+ = A_1^+ \cup B$, for a bounded set $B$. Let $R > 0$ such that $\Gamma_1 \Delta \Gamma_2 \subset \Dd_R(0)$ and $B \subset \Dd_R(0)$. Then:
\begin{equation}
   \big(A_1^+ \sqcup A_1^-\big) \cap \Dd_R(0)^c = \big(\R^2 \setminus \Gamma_1 \big) \cap \Dd_R(0)^c = \big(\R^2 \setminus \Gamma_2\big) \cap \Dd_R(0)^c = A_1^+ \cap \Dd_R(0)^c,
\end{equation}
so $A_1^- \subset \Dd_R(0)$: this is a contradiction because $A_1^-$ is unbounded.

We now work on existence. 
Let $\Omega_j^\pm \subset \C \setminus K$ open, such that $A_j^\pm \cap \Omega_j^\pm$ is connected and $\C \setminus \Omega_j^\pm$ is compact -- it exists because of Lemma \ref{lem:4}(iii). Let $\Omega = \bigcap_{j,\pm} \Omega_j^\pm$; since $A_1^+$ is unbounded and $\C \setminus \Omega$ is compact, $A_1^+ \cap \Omega$ is not empty. Fix $\bfx \in A_1^+ \cap \Omega$. Note that $\bfx \in \Omega$, so $\bfx \notin K$; and $\bfx \in A_1^+$, so $\bfx \notin \Gamma_1$; from which it follows $\bfx \notin \Gamma_2$, because $\Gamma_2 \subset \Gamma_1 \cup K$. After a potential relabelling of $A_2^+, A_2^-$, the set $A_2^+$ is the connected component of $\R^2 \setminus \Gamma_2$ containing $\bfx$.

We now show that $A_2^+ \cap \Omega = A_1^+ \cap \Omega$. Indeed, write 
\begin{align}
    A_1^+ \cap \Omega 
    & = A_1^+ \cap \Omega \cap \big(\R^2 \setminus \Gamma_2\big) \\
    & = A_1^+ \cap \Omega \cap \big(A_2^+ \sqcup A_2^-\big) = \big(A_1^+ \cap \Omega \cap A_2^+\big) \sqcup \big(A_1^+ \cap \Omega \cap A_2^-\big).
\end{align}
Since $A_1^+ \cap \Omega$ is connected and $A_1^+ \cap \Omega \cap A_2^+$, $A_1^+ \cap \Omega \cap A_2^-$ are disjoint open sets, we deduce from $\bfx \in A_1^+ \cap \Omega \cap A_2^+$ that $A_1^+ \cap \Omega  \subset A_2^+ \cap \Omega$. Similarly, 
$A_2^+ \cap \Omega \subset A_1^+ \cap \Omega$.

It follows that $A_1^+ \Delta A_2^+ \subset \C \setminus \Omega$, which is compact. In particular, $A_1^+ \Delta A_2^+$ is bounded. This completes the proof.\end{proof}

\subsection{Right and left of a simple path}\label{sec:2.3}

In this section, we review notions from \cite[\S2]{T11}, adapted to our context. 

\begin{definition}\cite[Definition 2.7.10]{T11} Let $\gamma$ be a simple path and $\Dd$ be an open disk in the plane. We say that $\gamma$ simply splits $\Dd$ if $I = \gamma^{-1}(\Dd)$ is an open interval and $\Dd \setminus \gamma(I)$ has precisely two connected components.
\end{definition}

\begin{prop}\cite[Theorem 2.7.11]{T11} Let $R > 0$, $\gamma$ be a simple path and $\bfz \in \gamma(\R)$. There exists $r < R$ such that $\Dd_r(\bfz)$ is simply split by $\gamma$. 
\end{prop}

Our goal now is to define unambiguously the connected component lying to the left of a  simple path. We fix a simple path $\gamma$, with range denoted by $\Gamma$. Given $D = \Dd_r(\bfz)$ an open disk centered at a point $\bfz \in \Gamma$, of radius $r$, simply split by $\Gamma$, we write $\gamma^{-1}(D) = (a,b)$ and (using that $\gamma(a), \gamma(b)$ are both on $\p D$)
\begin{equation}
    \gamma(a) = z + r e^{i\alpha}, \qquad \gamma(b) = z + r e^{i\beta}, \qquad \alpha \in (\beta, \beta+2\pi).
\end{equation}
We then define a Jordan curve $\ell : [-1,1] \rightarrow \C$ by:
\begin{equation}\label{eq2-6c}
    \ell(t) = \systeme{\gamma\big(b (1+t)-a t\big), & t \in [-1,0] 
    \\ z + r e^{i(1-t)\beta + it \alpha} & t \in [0,1]}.
\end{equation}
By the Jordan curve theorem, it encloses a bounded region.

\begin{definition}\label{def:3} The left of $\gamma$ in the simply split disk $\Dd$, denoted $L_\Dd$, is the bounded component of $\C \setminus \ell([-1,1])$.
\end{definition}

We let $A_\Dd$ be the connected component of $\C \setminus \Gamma$ that contains $L_\Dd$. A priori $A_\Dd$ depends on the choice of simply split disk. However we have the following result:

\begin{prop}\label{prop:1} If $\Dd_-, \Dd_+$  are two disks with centers on $\Gamma$, simply split by $\gamma$, then $A_{\Dd_-} = A_{\Dd_+}$.
\end{prop}

Proposition \ref{prop:1} allows us to talk unambiguously about the component of $\C \setminus \Gamma$ to the left of $\gamma$. Likewise, we can now talk about the component of $\C \setminus \Gamma$ to the right of $\gamma$, as the other connected component of $\C \setminus \Gamma$. We will need the following lemma:

\begin{lemma}\label{lem:13} Let $\Dd_-, \Dd_+$ be two disks with centers on $\Gamma$, simply split by $\gamma$.
\begin{itemize}
    \item[(i)] If $\Dd_- \subset \Dd_+$, then $L_{\Dd_-} \subset L_{\Dd_+}$.
    \item[(ii)] If $\gamma^{-1}(\Dd_-) \cap \gamma^{-1}(\Dd_+) \neq \emptyset$, then $L_{\Dd_-} \cap L_{\Dd_+} \neq \emptyset$.
\end{itemize}
\end{lemma}

\begin{figure}[t] 
    \centering 
    \floatbox[{\capbeside\thisfloatsetup{capbesideposition={right,top},capbesidewidth=0.4\textwidth}}]{figure}[\FBwidth]
    {\caption{The contours $\ell, \ell', \ell''$. The point $w$ is outside $\ell-\ell'$.}} 
    {\includegraphics[width=1\linewidth]{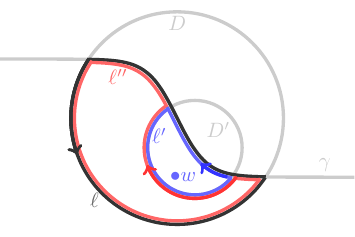}\label{fig:p8}} 
\end{figure}

\begin{proof} (i) Let $\ell_\pm$ be the contour given by \eqref{eq2-6c} for the simply split disk $\Dd_\pm$. Because $L_{\Dd_\pm}$ is the bounded component of $\C \setminus \ell_\pm(\R)$, it has an index characterization:
\begin{equation}\label{eq2-6a}
    L_{\Dd_\pm} = \big\{ w \in \C \setminus \ell([-1,1]) : \ \Ind_{\ell_\pm}(w) \neq 0\big\}, 
\end{equation}
see \cite[Theorem 4.2.8]{T11}. The index in \eqref{eq2-6a} is defined as the line integral
\begin{equation}\label{eq2-6d}
   \Ind_{\ell_\pm}(w) = \dfrac{1}{2\pi i}\oint_{\ell_\pm} \dfrac{dz}{z-w};
\end{equation}
note that $\ell_\pm$ is piecewise smooth so \eqref{eq2-6d} makes sense.

Let $w \in L_{\Dd_-}$ (in particular $\Ind_{\ell_-}(w) \neq 0$ and $w \in \Dd_- \subset \Dd_+$) and define $m = \ell_+ \ominus \ell_-$; see Figure \ref{fig:p8}. The contour $m$ is contained in $\Dd_+ \setminus \Dd_-$. Therefore, $w$ belongs to the unbounded component of $\C \setminus m(\R)$. It follows that $\Ind_m(w) = 0$ -- again by \cite[Theorem 4.2.8]{T11} -- and 
\begin{equation}
    0 = \Ind_m(w) = \Ind_{\ell_+}(w) - \Ind_{\ell_-}(w).
\end{equation}
Hence, $\Ind_{\ell_+}(w) = \Ind_{\ell_-}(w) \neq 0$. Thus $w$ is a point of $\Dd_+$ with non-zero index with respect to $\ell_+$: $w \in L_{\Dd_+}$. This proves $L_{\Dd_-} \subset L_{\Dd_+}$. 

(ii) Assume that  $\gamma^{-1}(\Dd_-) \cap \gamma^{-1}(\Dd_+) \neq\emptyset$ and let $t$ in this set. Then $\bfz = \gamma(t)$ is in the open set $\Dd_- \cap \Dd_+$, so there exists $R > 0$ such that $\Dd_R(\bfz) \subset \Dd_- \cap \Dd_+$. Let $r < R$ such that $\Dd = \Dd_r(\bfz)$ is simply split by $\gamma$. Because $\Dd \subset \Dd_\pm$, we have $L_\Dd \subset R_{\Dd_\pm}$ by Part (i). It follows that $L_{\Dd_-}$ and $L_{\Dd_+}$ intersect.
\end{proof}

\begin{proof}[Proof of Proposition \ref{prop:1}] Write $\Dd_\pm = \Dd_{r_\pm}(\bfz_\pm)$, $\bfz_\pm = \gamma(t_\pm)$. Without loss of generalities $t_- \leq t_+$. Given any $t \in [t_-,t_+]$, let $\Dd_t$ be an open disk centered at $\gamma(t)$ simply split by $\gamma$; without loss of generalities $\Dd_{t_+} = \Dd_+$, $\Dd_{t_-} = \Dd_-$. Let $I_t = \gamma^{-1}(\Dd_t)$, which is an open interval containing $t$. 

We now cover $[t_-,t_+]$ by the sets $I_t$ and pass to a finite covering $\CC_0 \subset \{ I_t : t \in [t_-,t_+]\}$; without loss of generalities $I_{t_-}$ and $I_{t_+}$ are elements of $\CC_0$. We now define a refinement $\CC_1$ of $\CC_0$ as follows. If there exist $I, I' \in \CC_0$ such that  $I \subset I'$, then $\CC_1 = \CC_0 \setminus \{I\}$; otherwise $\CC_1 = \CC_0$. Repeating the procedure on $\CC_1$ and so on, we obtain nested coverings $\dots \subset \CC_2 \subset \CC_1 \subset \CC_0$. Note that for some $n$, $\CC_n = \CC_{n-1}$ because $\CC_0$ is finite and the number of elements decreases by one when $\CC_j \neq \CC_{j-1}$. We define $\CC = \CC_n$. 

We now order the elements of $\CC$ according to their infimum:
\begin{equation}
    \CC = \{ I_0, I_1, \dots, I_N \}, \qquad \inf I_0 < \inf I_1 < \dots  < \inf I_N.
\end{equation}
Without loss of generalities $I_{t_-} \subset I_0$ and $I_{t_+} \subset I_N$. Note that we may not have for any $K$, $\inf I_K = \inf I_{K+1}$: otherwise, $I_K \subset I_{K+1}$ or $I_{K+1} \subset I_K$, which is impossible by our construction of $\CC$. Moreover, for all $K$, $I_K$ and $I_{K+1}$ must intersect: if they do not, 
\begin{equation}
    \sup I_K \leq \inf I_{K+1} < \dots  < \inf I_N,
\end{equation}
which implies that the connected set $[t_-,t_+]$ is covered by the disjoint sets $\bigcup_{k \leq K} I_k$ and $\bigcup_{k > K} I_k$. But both of these sets intersect $[t_-,t_+]$ so this is impossible.

Let $L_k$ be the left of $\gamma$ in the disk $\Dd_{t_k}$ -- see \eqref{eq2-6a}. Because of Lemma \ref{lem:13}(ii) and $I_k \cap I_{k+1} \neq \emptyset$, we have $L_k \cap L_{k+1} \neq \emptyset$. Therefore, the set $L = \bigcup_{k = 0}^N L_k$ is a connected subset of $\C \setminus \Gamma$.

Recall that $I_{t_-} \subset I_0 \in \CC$; let $t_0 \in [t_-,t_+]$ with $I_{t_0} = I_0$: in particular, $\gamma^{-1}(\Dd_{t_-}) \subset \gamma^{-1}(\Dd_{t_0})$, hence $L_{\Dd_-} \subset L_m$ by Lemma \ref{lem:13}(i). In particular $L_{\Dd_-} \subset L$. Likewise $L_{\Dd_+} \subset L$. We deduce that $L$ is a connected subset of $\C \setminus \Gamma$ that intersects the connected components $A_{\Dd_-}$ and $A_{\Dd_+}$ of $\C \setminus \Gamma$, hence $A_{\Dd_-} = A_{\Dd_+}$. This completes the proof.\end{proof}

If one perturbs a path $\gamma_1$ in a compact set to a path $\gamma_2$, then $\C \setminus \gamma_1(\R)$ and $\C \setminus \gamma_2(\R)$ have connected components that coincide outside a compact set (Lemma \ref{lem:5}). Proposition \ref{prop:4} states that these connected components are those that lie on the same side of $\gamma_1, \gamma_2$. We are finally ready for its proof.

\begin{proof}[Proof of Proposition \ref{prop:4}]  Let $V$ be the component of $\C \setminus \Gamma_2$ such that $A_1 \Delta V$ is bounded -- see Lemma \ref{lem:5}. Let $R > 0$ such that $A_2 \Delta V$ and $\Gamma_1 \Delta \Gamma_2$ are contained in $\Dd_R(0)$. Let $\bfz \in \Gamma_1$ with $|z| \geq R+1$; note $\bfz \in \Gamma_2$. 

Let $D = \Dd_r(\bfz)$ be an open disk with radius $r < 1$, simply split by $\gamma_2$ and $L$ be the left of $\gamma_1$ in $D$. Because $\gamma_1$ and $\gamma_2$ coincide in $\Dd_1(\bfz)$, $D$ also is simply split by $\gamma_2$, and the left of $\gamma_2$ in $\Dd_2$ is equal to $L$. Therefore, $L \subset A_2$.

We claim that $L \subset V$. Indeed, $L \subset A_1$ and $L$ does not intersect $\Dd_R(0)$. Because $A_1 \Delta V \subset \Dd_R(0)$, we have $L \subset V$. It follows that $A_2$ and $V$ are two intersecting components of $\C \setminus \Gamma_2$, in particular they are equal. This completes the proof.   
\end{proof}

\section{Intersection numbers are well-defined}\label{app:B}

The goal here is to prove that intersection numbers are well-defined for transverse sets $U,V$ such that $U$ is simple. From Definition \ref{def:5}, if $\p U$ is bounded (that is, it is the range of a simple loop), then $\II_{U, V}$ is clearly well-defined. If $\p U$ is unbounded, then it is the range of a simple path $\gamma$ such that $U$ lies to the left of $\gamma$, and 
\begin{equation}
\II_{U, V} \de \II_+(U,V) - \II_-(U,V), \qquad
\II_\pm(U,V) \de \lim_{t \rightarrow \pm \infty}  \1_V \circ \gamma(t).
\end{equation}

\begin{lemma}\label{lem:9} If $U,V$ are transverse sets such that $U$ is simple, then:
\begin{itemize}
    \item[(i)]  The limits \eqref{eq2-6e} exist;
    \item[(ii)] They do not depend on the choice of $\gamma$ such that $U$ is to the left of $\gamma$. 
\end{itemize}  
\end{lemma}

\begin{proof} If $\p U$ is bounded, there is nothing to do. So we assume below that $\p U$ is unbounded; in particular $\gamma$ is a simple path.

(i) If $U,V$ are transverse sets, then $\p U \cap \p V$ is bounded.  Since $\gamma$ is proper, there exists $T > 0$ such that for all $t \geq T$, $\gamma(t) \notin \p U \cap \p V$, so $\gamma(t) \notin \p V$. Therefore $\gamma([T,\infty))$ is a connected subset of $\R^2 \setminus \p V$, it is fully contained in $V$ or in $V^c$. We conclude that $\1_V \circ \gamma(t)$ is constant for $t \geq T$, so $\II_+(U,V)$ is well-defined. Likewise, $\II_-(U,V)$ is well-defined. 

(ii) We note that 
\begin{equation}\label{eq2-6j}
    \II_+(U,V) = \systeme{1 & \text{ if $\gamma(\R^+) \setminus V$ is bounded;} \\ 0 & \text{ otherwise.}}
\end{equation}
Let us prove that the RHS of \eqref{eq2-6j} does not depend on the choice of curve $\gamma$ as long as $U$ lies to the left of $\gamma$. If $\tgamma$ is another such simple path, then we claim that $\gamma(\R^+) \Delta \tgamma(\R^+)$ is bounded. In particular, this would prove that the condition $\gamma(\R^+) \cap V$ bounded, arising in \eqref{eq2-6j}, is independent of $\gamma$, and hence that $\II_{U, V}$ does not depend on $\gamma$.

Let us argue by contradiction. Then because $\gamma$ is proper, 
\begin{equation}
    \gamma^{-1} \left( \gamma(\R^+) \Delta \tgamma(\R^+) \right) = \R^+ \Delta \gamma^{-1} \big( \tgamma(\R^+) \big)
\end{equation}
is unbounded.

By continuitiy of $\gamma$ and because $\gamma, \tgamma$ have same range, the set  $\gamma^{-1} \big( \tgamma(\R^+) \big)$ is an interval, so we have $\tgamma(\R^+) = \gamma\big( (-\infty,t_-) \big)$ for some $t_- \in \R$.

We define $\alpha : \R \rightarrow \C$ by:
\begin{equation}
    \alpha(t) = \systeme{ \gamma(t+t_-), \quad t \geq 0, \\ \tgamma(-t), \quad t < 0}
\end{equation}
Because $\gamma\big( [t_-,+\infty) \big)$ and $\tgamma(\R^+) = \gamma\big( (-\infty,t_-) \big)$ do not intersect, the map $\gamma$ is a simple path. Its range is
\begin{equation}
    \alpha(\R) = \gamma\big( [t_-,+\infty) \big) \cup \tgamma(\R^+) = \gamma\big( [t_-,+\infty) \big) \cup \gamma\big( (-\infty,t_-) \big) = \gamma(\R) = \p U.
\end{equation}

Let $W$ be the component of $\C \setminus \p U$ to the left of $\alpha$. Set $\bfz_+ = \alpha(1)$, and $\Dd_+$ an open disk simply split by $\alpha$, centered at $\bfz_+$, such that $\alpha^{-1}(\Dd_+) \subset \gamma^{-1}(\R^+)$ -- we can realize this by taking a small enough radius because $\alpha$ is injective. In $\Dd_+$, $\alpha$ and $\gamma$ are just same-speed affine reparametrization of each other, so the left of $\alpha$ and $\gamma$ in $\Dd_+$ coincide. We deduce that that $U$ and $W$ are two intersecting connected components of $\C \setminus \p U$, so $U=W$.  

Set now $\bfz_- = \alpha(-1)$, and $\Dd_-$ an open disk simply split by $\alpha$, centered at $\bfz_-$, such that $\alpha^{-1}(\Dd_-) \subset \tgamma^{-1}(\R^+)$. In $\Dd_-$, $\alpha$ and $\tgamma$ are just opposite-speed affine reparametrization of each other, so the left of $\alpha$ and the right of $\tgamma$ in $\Dd_-$ coincide. We deduce that $U^\CCC = \Int U^c$ (the right of $\tgamma$) and $W$ intersect, so $U^\CCC=W$: that's a contradiction. This completes the proof.\end{proof}

\section{From discrete to good sets}\label{app:C}

\begin{proof}[Proof of Lemma \ref{lem-9}] In this proof we use the notation $\Ss_r(\bfx)$ for the open square centered at $\bfx = (x_1,x_2)$, of side $2r$ -- the $|\cdot|_\infty$ open ball centered at $\bfx$, of radius $r$:
\begin{equation}
    \Ss_r(\bfx) = (x_1-r,x_1+r) \times (x_2-r,x_2+r).
\end{equation}
We immediately see that $A \cap \Z^2 = \Aa$: if $\bfx \in \Aa$ then $d_1(\bfx,\Z^2 \setminus \Aa) = 1 > 3/4$; conversely if $\bfx \in A \cap \Z^2$ then $d_1(\bfx,\Z^2 \setminus \Aa) > 3/4$ so $\bfx \notin \Z^2 \setminus A$ and $\bfx \in A$. It now remains to show that the conditions (a), (b) and (c) from Definition \ref{def:4} hold. 

1. Note that (c) is immediate: if $\bfx \in \p A$ then $d_1(\bfx,\Z^2 \setminus \Aa) = 3/4$, so $\bfx$ cannot have integer coordinates. We now prove (a): $\p A^\CCC = \p A$. We have:
\begin{equation}
    \p A^\CCC = \ove{A^\CCC} \setminus A^\CCC, \qquad \p A=  \p A^c  =  A^c \setminus A^\CCC.
\end{equation}
Therefore, it suffices to show that $\ove{A^\CCC} = A^c$. The direct inclusion is immediate because $A^c$ is closed. As for the reverse inclusion, we write
\begin{equation}
    A^c = \big\{ \bfx \in \R^2 : d_1(\bfx,\Z^2 \setminus \Aa) \leq 3/4 \big\} = \bigcup_{\bfy \in \Z^2 \setminus \Aa} \ove{\Ss_{3/4}(\bfy)}. 
\end{equation}
Hence, if $\bfx \in A^c$, there exists $\bfy \in \Z^2 \setminus \Aa$ such that $\bfx \in \ove{\Ss_{3/4}(\bfy)}$. Let now $\bfx_n \in \Ss_{3/4}(\bfy)$ with $\bfx_n \rightarrow \bfx$, we have $\bfx_n \in \Ss_{3/4}(\bfy) = \Int \Ss_{3/4}(\bfy) \subset A^\CCC$ and hence $\bfx \in \ove{A^\CCC}$. This proves (a). 

2. We now focus on proving (b): $\p A$ is the disjoint union of uniformly separated simple paths or loops. Let $\Lambda = p + (\Z/2)^2$, $p = (1/4,1/4)$. We call two points $\bfx,\bfy \in \Lambda$ neighbors if $[\bfx,\bfy] \subset \p A$. We shall show the following properties:
\begin{itemize}
    \item[(1)] Any point $\bfx \in \Lambda \cap \p A$ has exactly two neighbors.
    \item[(2)] If $\bfz \in \p A$, there exists $\bfx,\bfy$ neighbors such that $\bfz \in [\bfx,\bfy]$.
\end{itemize}

We first prove (1). Let $\bfx \in \lambda \cap \p A$, let $m \in \Z^2$ closest to $\bfx$. Note that $|\bfx-m|_\infty = 1/4$, hence $m \in \Aa$. Up to a rotation, and a translation, we may assume that $m = (0,0)$ and $\bfx=p$. We now perform a case analysis, depending if the points $(0,1), (1,0)$ and $(1,1)$ are in $\Aa$ or $\Z^2 \setminus \Aa$. There are $2^3 = 8$ cases to consider -- in fact, just $7$ because at least one of these points is in $\Z^2 \setminus \Aa$ as $d_\infty(\bfx,\Z^2 \setminus \Aa) = 3/4$. We proceed through Figure \ref{fig:12}.

\begin{figure}[t]
  \centering
  \includegraphics[width=1\textwidth]{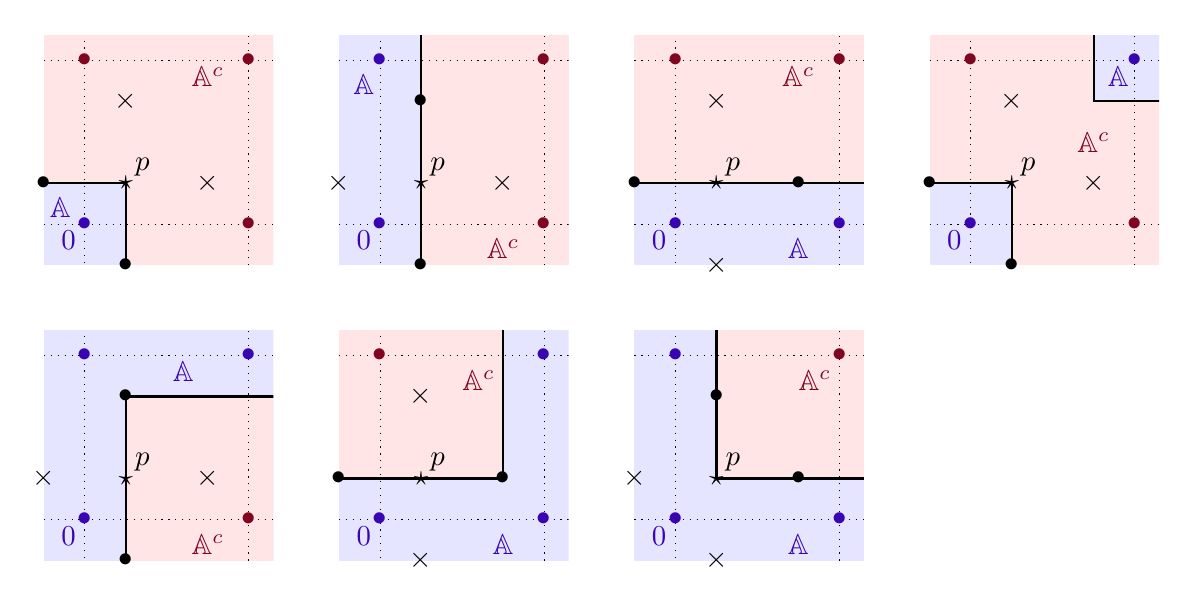}
  \caption{The 7 possible cases. On each picture, the blue bullets are points of $\Aa$ and the red bullets are points of $\Aa^c$. The set $A$ is filled in blue and $A^c$ is filled in red. The point $p$ is at the star. The black line is the boundary of $A$. We labelled each of the 4 points of $\Lambda$ closest to $p$ with a black bullet if it is a neighbor of $p$ and a cross if it is not a neighbor of $p$. In each case, $p$ has exactly two neighbors.}
  \label{fig:12}
\end{figure}

We now prove (2). Assume that $\bfz \in \p A$. If $\bfz \in \Lambda$, then the statement is clear by (1). So let's assume that $\bfz \notin \Lambda$. We note that $A^c$ is the union of closed squares of sides $3/2$, centered at points in $\Z^2 \setminus \Aa$. Each of these squares being union of squares of sides $1/2$, centered at points in $(\Z/2)^2$, we deduce that there are points $w_1, w_2, \dots \in (\Z/2)^2$ such that 
\begin{equation}\label{eq2-7w}
    A^c = \bigcup_{j \in \N} \ove{\Ss_{1/4}(w_j)}.
\end{equation}
In particular, 
\begin{equation}
 \p A = \p (A^c) = \p \bigcup_{j \in \N} \ove{\Ss_{1/2}(w_j)} \subset \bigcup_{j \in \N} \p \ove{\Ss_{1/2}(w_j)} \subset \bigcup_{w \in (\Z/2)^2} \p \ove{\Ss_{1/2}(w)} =  \bigcup_{\substack{\bfx, \bfy \in \Lambda \\ |\bfx-\bfy| = 1/2}} [\bfx,\bfy].
\end{equation}
In particular, there exists two points $\bfx,\bfy \in \Lambda$ such that $\bfz \in [\bfx,\bfy]$. We now show that $[\bfx,\bfy] \subset \p A$. 

If $\bfz \in \Lambda$, then by (1), $[\bfx,\bfy] \subset \p A$. Otherwise, there exists precisely two closed squares $\Ss_0, \Ss_1$ of the form $\ove{\Ss_{1/4}(\bfw)}$, $\bfw \in (\Z/2)^2$, containing $\bfz$. If neither appear in \eqref{eq2-7w}, then $\bfz \in \Int A = A$ -- that's impossible. If both of them appear in \eqref{eq2-7w}, then $\bfz \in A^\CCC$ so $\bfz \notin \p A$ -- that's also impossible. So, one of them must appear in \eqref{eq2-7w} -- say $\Ss_0$ -- and hence the interior of $\Ss_1$ does not intersect $A^c$. It follows that $\Ss_0 \cap \Ss_1 = [\bfx,\bfy]$ is in $\p A^c = \p A$. This completes the proof of (2).

3. Let $\Gamma$ a connected component of $\p A$. Fix $\bfz \in \Gamma$. By (2), there exists $\bfx,\bfy \in \p A$ with $\bfz \in [\bfx,\bfy] \subset \p A$. We now construct a sequence 
\begin{equation}
   \gamma(0) = \bfx, \ \gamma(1) = \bfy, \ \gamma(2), \ \dots
\end{equation}
such that for all $n \in \Z$, $\gamma(n-1)$ and $\gamma(n+1)$ are the two neighbors of $\gamma(n)$. Start by defining $\gamma(2)$ as the neighbor of $\bfy = \gamma(1)$ that is not $\bfx = \gamma(0)$. Then assume we constructed $\gamma(0), \dots, \gamma(n)$ -- in particular $\gamma(n)$ and $\gamma(n-1)$ are neighbors. Define then $\gamma(n+1)$ as the neighbor of $\gamma(n)$ that is not $\gamma(n-1)$. Similarly construct $\gamma(-1), \gamma(-2), \dots$. 

We then extend $\gamma : \Z \rightarrow \R^2$ to $\gamma : \R \rightarrow \R^2$ by:
\begin{equation}
    \gamma(t) = (1-t+n)\gamma(n) + (t-n)\gamma(n+1), \qquad t \in [n,n+1].
\end{equation}
This defines a curve $\gamma$ with $\gamma(\R) \subset \p A$, intersecting $\Gamma$ at $\bfz$. Because $\gamma(\R)$ is connected, we have $\gamma(\R) \subset \Gamma$. Note now that 
\begin{equation}\label{eq2-7u}
    \Gamma = \Gamma_1 \sqcup \Gamma_2, \quad \Gamma_1 \de \big\{\bfx \in \Gamma : d\big(\bfx,\gamma(\R)\big) \leq 1/5 \big\}, \quad \Gamma_2 \de \big\{ \bfx \in \Gamma : d\big(\bfx,\gamma(\R)\big) > 1/5 \big\}.
\end{equation}
Moreover, $\Gamma_1 = \gamma(\R)$. Indeed, assume that $\bfz \in \Gamma$ is such that $d\big(\bfx,\gamma(\R)\big) \leq 1/4$. By (2) above, there exists neighbors $\bfx,\bfy$ such that $\bfz \in [\bfx,\bfy]$. We have either $|\bfx-\bfz| \leq 1/5$ or $|\bfy-\bfz| \leq 1/5$ -- say the latter. Then $d_1(\bfx,\gamma(\R)) < 1/2$. 

Note that $\gamma(\R)$ is the union of segments of the form $[u,v]$, $u,v \in \Lambda$. No such segments can intersect a unit square centered at $\bfx$, unless one of the vertices is $\bfx$. Therefore  $d_1(\bfx,\gamma(\R)) < 1/2$ implies $\bfx \in \gamma(\R)$. Because $\Lambda \cap \gamma(\R) = \gamma(\Z)$, $\bfx = \gamma(n)$ for some $n \in  \Z$. And because $\bfy$ is a neighbor of $\bfx$, $\bfy$ is either $\gamma(n+1)$  or $\gamma(n-1)$, and $\bfz \in [\bfx,\bfy] \subset \gamma(\R)$. So $\Gamma_1 \subset \gamma(\R)$ and (the other inclusion being immediate) $\Gamma_1 = \gamma(\R)$ indeed.

In particular,
\begin{equation}
    \big\{\bfx \in \Gamma : d\big(\bfx,\gamma(\R)\big) = 1/5 \big\} = \emptyset.
\end{equation}
Indeed, if $d\big(\bfx,\gamma(\R)\big) = 1/5$ then $\bfx \in \Gamma_1 = \gamma(\R)$, but then $d\big(\bfx,\gamma(\R)\big) = 0$  -- a contradiction. This implies that
\begin{equation}
    \Gamma_1 = \big\{\bfx \in \Gamma : d\big(\bfx,\gamma(\R)\big) < 1/5 \big\}.
\end{equation}
In particular, $\Gamma_1$ is an open subset of $\Gamma$. So is $\Gamma_2$; therefore \eqref{eq2-7u} is a separation of $\Gamma$ in two disjoint open set, and using connectedness, one of them must be empty. Because $\Gamma_1 \neq \emptyset$, this prompts $\Gamma_2 = \emptyset$, so finally $\Gamma = \gamma(\R)$. 

4. We now justify that $\gamma$ is a simple path or loop. On the regularity side, $\gamma$ has non-zero left and right derivatives at every point (with norm $1$ after arclength parametrization), and is non-differentiable only at integer values. 

Assume now that $\gamma$ is injective. If $\gamma$ was not proper, there would exist $t_k \rightarrow \infty$ with $\gamma(t_k)$ bounded; by definition of $\gamma$, this implies that there be a sequence $n_k$ of distinct integers with $\gamma(n_k) \in B$, $B$ bounded. But $\gamma(n_k) \in \Lambda$ which is discrete, hence such that $\Lambda \cap B$ is finite; this contradicts injectivity. Hence $\gamma$ is proper indeed: this proves that $\gamma$ is a simple path.

Assume now that $\gamma$ is not injective and define:
\begin{equation}
    s_* = \min \big\{ s \in \N : \ \exists n \in \Z, \ \gamma(n+s) = \gamma(n) \big\}
\end{equation}
Note that $s_* \geq 2$; let $n_* \in \Z$ with $\gamma(n_* + s_*) = \gamma(n_*)$. Note that $\gamma(n_*+1)$ and $\gamma(n_*+s_*-1)$ are both neighbors of $\gamma(n_*) = \gamma(n_*+s_*)$; by minimality of $s_*$ they are different. It follows that $\gamma(n_*-1) = \gamma(n_*+s_*-1)$. Because $\gamma(n+1)$ is fully determined by the values of $\gamma(n)$ and $\gamma(n-1)$, we obtain
\begin{equation}
 \forall n \geq n_*, \qquad   \gamma(n+s_*) = \gamma(n).
\end{equation}
Likewise, $\gamma(n_*+1) = \gamma(n_*+s_*+1)$, and by the same argument,
\begin{equation}
 \forall n \leq n_*, \qquad   \gamma(n+s_*) = \gamma(n).
\end{equation}
It follows that $\gamma$ is periodic on $\Z$, hence on $\R$. This implies that $\gamma$ is a simple loop, and completes the proof. \end{proof}

\begin{proof}[Proof of Lemma \ref{lem-13}] Write $B = A^\CCC$. If $A$ is good, then by Definition \ref{def:4}(b) and (c), $\p B = \p A$ is the union of separated ranges of simple paths or loopsthat do not intersect $\Z^2$ . 

Therefore we just have to show that $\p B = \p B^\CCC$. We first show that $\Int \ove{A} = A$. Indeed, $\Int \ove{A}$ contains $A$; assume conversely that there exists $\bfx \in \Int \ove{A} \setminus A$. Then by Definition \ref{def:4}(a) and the identity $A^\CCC = \ove{A}^c$, $\bfx \in \p A = \p A^\CCC = \p \ove{A}^c = \p \ove{A}$, but that's impossible for $\bfx \in \Int \ove{A}$. So $\Int \ove{A} = A$ indeed.

Using again  $A^\CCC = \ove{A}^c$, we have now:
\begin{equation}
    \p B^\CCC = \p \Int (A^\CCC)^c = \p \Int \ove{A} = \p A = \p B.
\end{equation}
This completes the proof.    
\end{proof}

\section{Properties of good sets}\label{app:D}

\begin{proof}[Proof of Lemma \ref{lem:20}] 1. Let $A$ be good, and $B$ a connected component of $A$. We have $\p B \subset \p A$ -- see e.g. \cite[Theorem IV.3.1]{N51}. 

Now consider $\bfz_- \in \p B$. We have $\bfz_- \in \p A$, so $\bfz_-$ belongs to one of the $\Gamma_k$: say $\bfz_- \in \Gamma_1$. Fix now $\bfz_+ \in \Gamma_1$, and $\gamma_1$ a simple path or loop with range $\Gamma_1$. Let $\rho_0 = \inf_{j \neq k} d(\Gamma_j, \Gamma_k) > 0$. We present the proof when $\gamma$ is a simple path; the proof for a simple loop is similar. As in the proof of Proposition \ref{prop:1}, we can construct disks $\Dd_0, \dots, \Dd_N$ such that:
\begin{itemize}
    \item For all $k$, $\Dd_k$ is centered on $\Gamma_1$, is simply split by $\gamma_1$, and has radius less than $\rho_0$;
    \item $\bfz_- \in \Dd_0$ and $\bfz_+ \in \Dd_N$
    \item If $L_k$ denotes the left of $\gamma_1$ in $\Dd_k$ then $L = \bigcup_{k=0}^N L_k$ is connected; and if $R_k$ denotes the right of $\gamma_1$ in $\Dd_k$ then $R = \bigcup_{k=0}^N R_k$ is connected.
\end{itemize}

Fix $n \in \N$, large enough so that $\Dd_{1/n}(\bfz_-) \subset \Dd_0$ and $\Dd_{1/n}(\bfz_+) \subset \Dd_N$. Let $\bfx_n \in \Dd_{1/n}(\bfz_-) \cap B$; $\bfx_n$ exists because $\bfz_- \in \Dd_0 \cap \p B$). Note that $\Dd_{1/n}(\bfz_-) \cap B$ is non-empty open set, in particular it has positive Lebesgue measure; and $\Gamma_1$ has null Lebesgue measure. So without loss of generalities $\bfx_n \notin \Gamma_1$. 

Assume that $\bfx_n$ is to the left of $\gamma_1$ in $\Dd_0$. Let $\bfy_n$ to the left of $\gamma_1$ in $\Dd_{1/n}(\bfz_+)$. Because $\Dd_{1/n}(\bfz_+) \subset \Dd_N$, we have $\bfy_n \in L$, see Lemma \ref{lem:13}(i). Let $\gamma$ a path connecting $\bfx_n$ to $\bfy_n$ in $L$ ($L$ is a connected subset of $\R^2$, so it is path-connected). Then $\gamma$ may not intersect $\Gamma_1$, nor may it intersect other components of $\p A$ because $L \subset \{ \bfx : d(\bfx,\Gamma_1) < \rho_0\}$. So $\gamma$ starts in $B$ and does not cross $\p B$, therefore it remains in $B$. It follows that $\bfy_n \in B$. A similar argument holds if $\bfx_n$ was to the right of $\gamma_1$ in $\Dd_0$. 

Now making $n \rightarrow \infty$ proves that $\bfz_+ \in \ove{B}$. Because $\bfz_+ \in \Gamma_1$, we have $\bfz_+ \notin A$ and a fortiori $\bfz_+ \notin B$, so $\bfz_+ \in  \p B$. This proves $\Gamma_1 \subset \p B$. 

We conclude that for each connected component $A_\ell$ of $A$, there exists $\KK_\ell \subset \N$ such that 
\begin{equation}
    \p A_\ell = \bigsqcup\limits_{k \in \KK_\ell} \p \Gamma_k.
\end{equation}
It remains to show that the sets $\KK_\ell$ partition $\N$.

2. Assume that $\KK_\ell \cap \KK_m \neq \emptyset$; let $k$ in this intersection and $\Dd$ a disk simply by $\Gamma_k$, of radius less than $\rho_0$. Let $L_\Dd$ and $R_\Dd$ be the left and right of $\gamma_1$ in $\Dd$. We note that $L_\Dd$ and $R_\Dd$ are at most $\rho_0$-distant from $\Gamma_1$, and do not intersect $\Gamma_k$, so they do not intersect $\p A$.

We claim that $A_\ell$ contains $L_\Dd$ or $R_\Dd$. Indeed, because $\p A_\ell$ contains $\Gamma_1$ and $\Dd \setminus \Gamma_1 = R_\Dd \sqcup L_\Dd$, $A_\ell$ must intersect one of these two sets, say $L_\Dd$, at a point $\bfx_0$. If $\bfx_1$ is another point of $L_\Dd$, we can connect $\bfx_0$ and $\bfx_1$ by a path $\bfx_s$ in $L_\Dd$ ($L_\Dd$ is an open subset of $\R^2$, so it is path-connected). The path $\bfx_s$ does not intersect $\p A_\ell$ and starts in $A_\ell$, so it ends in $A_\ell$: this proves $L_\Dd \subset A_\ell$. A similar argument implies that $R_\Dd \subset A_m$. It follows that
\begin{equation}
    \Dd \setminus \Gamma_\ell = L_\Dd \sqcup R_\Dd \subset A_\ell \sqcup A_m  \subset A.
\end{equation}
But then $A^\CCC$ may not have boundary $\Gamma_\ell$: this is a contradiction. So $\KK_\ell \cap \KK_m = \emptyset$.

3. We now prove that the  sets $\KK_\ell$ have union $\N$, or equivalently, that
\begin{equation}\label{eq2-7j}
    \p A = \bigcup_{\ell \in \N} \p A_\ell. 
\end{equation}
We already have the reverse inclusion because of \cite[Theorem IV.3.1]{N51}. Now if $\bfx_* \in \p A$, $\bfx_* \in \Gamma_k$ for some $k$; let $\Dd$ a disk simply split by $\Gamma_k$, of radius smaller than $\rho_0$, centered at $\bfx_*$.

The set $A$ contains points in either $L_\Dd$ or $R_\Dd$; say in $L_\Dd$. Therefore, there exists $\ell \in \N$ such that $A_\ell \cap L_\Dd \neq \emptyset$. As in step 2, any path in $L_\Dd$ may not cross $\p A$, so it must be contained in $A_\ell$. This proves $L_\Dd \subset A_\ell$.

Let now $n \geq N$, with $N$ large enough so that $\Dd_n \de \Dd_{1/n}(\bfx_*) \subset \Dd$. By Lemma \ref{lem:13}(i), we have $L_{\Dd_n} \subset L_\Dd \subset A_\ell$; by picking $\bfx_n \in L_{\Dd_n} \subset \Dd_n$, and taking the limit $n \rightarrow \infty$, we deduce that $\bfx_* \in \ove{A_\ell}$. Because $\bfx_* \notin A$, we have $\bfx_* \notin A_\ell$ and hence $\bfx_* \in \p A_\ell$. This completes the proof of \eqref{eq2-7j} and of \eqref{eq2-7k}.  

4. We now justify that connected components of good sets are good. Their boundaries are the union of well-separated traces of simple paths or loops that do not intersect $\Z^2$ because of \eqref{eq2-7k}; these are the conditions (b) and (c) in definition \ref{def:4}. Now if $\bfx \in \p A_\ell$, then $\bfx \in \p A$ and so $\bfx \in \p A^\CCC$. It follows that any neighborhood $\Omega$ of $\bfx$ intersects $A^\CCC$. Now we have:
\begin{equation}
    A^\CCC = \Int \bigcap_{m \in \N} A_m^c \subset  \bigcap_{m \in \N} \Int A_m^c = \bigcap_{m \in \N} A_m^\CCC.
\end{equation}
So $\Omega$ intersects $A_\ell^\CCC$. Because $\Omega$ also intersects $(A_\ell^\CCC)^c = \ove{A_\ell}$ at the point $\bfx$, we have $\bfx \in \p A_\ell^\CCC$. This proves $\p A_\ell \subset \p A^\CCC$; the other inclusion being always satisfied, we obtain the condition (a) in Definition \ref{def:4}.
\end{proof}

\begin{proof}[Proof of Lemma \ref{lem:18b}] By Lemma \ref{lem:20}, it suffices to show that if one of the sets $\AAA_k$ cannot have boundary containing two of the sets $\{ \AAA_\ell : \ell \in \N\}$. We argue by contradiction: (up to relabelling) assume that $\p \AAA_1$ contains $\Gamma_1$ and $\Gamma_2$. 

Let $A_1^\pm, A_2^\pm$ be the connected components of $\C \setminus \Gamma_1, \C \setminus \Gamma_2$, respectively. The set $\AAA_1$ is a connected subset of $\C \setminus \Gamma_1$ so it is contained in $A_1^+$ or $A_1^-$ -- say $\AAA_1 \subset A_1^+$.  Likewise $\AAA_1 \subset A_2^+$. In particular, $\p \AAA_1 \subset \ove{A_1^+} = A_1^+ \cup \Gamma_1$. Because $\Gamma_1$ and $\Gamma_2$ do not intersect, $\Gamma_2 \subset A_1^+$.

Let $\rho = \inf_{j \neq k} d(\Gamma_j, \Gamma_k) > 0$. Fix $\bfx \in \Gamma_1$, and $\bfy \in A_1^-$ with $d(\bfy,\Gamma_1) < \rho$. Note $\bfy \notin \AAA_1$; also $\bfy \notin \AAA_k$ for any $k \geq 2$: indeed, if $\bfy \in \AAA_k$, then because $\bfx \notin \AAA_k$, there would be a point $\bfz \in [\bfx,\bfy] \cap \p \AAA_k$, but that contradicts that $\rho$ is the above infimum. Therefore, $\bfy \in A$. In particular, $A$ and $A_1^-$ intersect. Because $A$ is a connected subset of $\C \setminus \Gamma_1$, $A \subset A_1^-$. 

This yields $\p A \subset \ove{A_1^-} = A_1^- \cup \Gamma_1$. Because $\Gamma_1$ and $\Gamma_2$ do not intersect, $\Gamma_2 \subset A_1^-$. That's a contradiction as $A_1^+ \cap A_1^- = \emptyset$.
\end{proof}

\bibliographystyle{amsxport}
\bibliography{ref.bib}
\end{document}